\documentclass[
preprint,
aps,
floatfix,
nofootinbib,
superscriptaddress,10pt]{revtex4-1}
\pdfoutput=1

\usepackage{tikz}
\usepackage{graphicx}
\usepackage{dcolumn}
\usepackage{bm}
\usepackage{amssymb} 

\usepackage{xcolor} 
\usepackage{amsthm} 
\usepackage{tensor} 
\usepackage{amsmath}
\usepackage{amssymb}
\usepackage{upgreek}
\usepackage[bottom, stable]{footmisc} 
\usepackage{verbatim}
\usepackage{mathtools}
\usepackage{enumitem}
\usepackage[mathscr]{euscript}
\usepackage[auto]{contour}
\usepackage{bbold}
\usepackage[colorlinks, allcolors=dblue!70!black, linktocpage]{hyperref}
\usepackage{cleveref}
\usepackage{longtable,array}

\definecolor{linkblue}{rgb}{0.1,0.3,.7}
\definecolor{forestgreen}{rgb}{0.13, 0.55, 0.13}
\definecolor{lava}{rgb}{0.81, 0.06, 0.13}
\definecolor{dblue}{rgb}{0.55,0.25,4.0} 

\newcommand{\nulleq}{~\widehat{=}~} 

\newtheorem{theorem}{Theorem}

\newtheorem{definition}{Definition}
\newtheorem{corollary}{Corollary}

\newtheorem{lemma}{Lemma}

\let\oldc\c \let\oldd\d \let\oldi\i

\newcommand{\beq}{\begin{equation}}
\newcommand{\eeq}{\end{equation}}
\newcommand{\bea}{\begin{align}}
\newcommand{\eea}{\end{align}}
\newcommand{\bos}{\boldsymbol}
\newcommand{\bom}{\boldsymbol{\omega}}
\newcommand{\boep}{\boldsymbol{\epsilon}}
\newcommand{\veps}{\varepsilon}

\newcommand{\boj}{\bos{\mathit{J}}}

\newcommand{\ed}{\mathrm{d}}
\newcommand{\tom}{\mathtt{\Omega}}
\newcommand{\boq}{\boldsymbol{Q}}
\newcommand{\bob}{\boldsymbol{B}}

\newcommand{\boe}{\boldsymbol{E}}
\newcommand{\boc}{\boldsymbol{\mathscr{C}}}
\newcommand{\bots}{\boldsymbol{\mathtt{S}}}
\newcommand{\bog}{\boldsymbol{\mathtt{e}}_\mathfrak{G}}
\newcommand{\bth}{\boldsymbol{\theta}}
\newcommand{\bTh}{\boldsymbol{\Theta}}

\newcommand{\bmq}{\bar{\mathsf{q}}}
\newcommand{\ul}{\underline}
\renewcommand{\i}{\mathbin{\lrcorner}}

\def\cB{{\cal{B}}}
\def\cC{{\cal{C}}}

\def\cL{{\cal{L}}}
\def\cH{{\cal{H}}}
\def\cP{{\cal{P}}}
\def\cV{{\cal{V}}}
\def\cM{{\cal{M}}}
\def\cN{{\cal{N}}}
\def\cS{{\cal{S}}}
\def\cE{{\cal{E}}}

\def\cZ{{\cal{Z}}}
\def\cT{{\cal{T}}}
\def\cU{{\cal{U}}}
\def\cF{{\cal{F}}}
\def\cQ{{\cal{Q}}}
\def\cR{{\cal{R}}}
\def\cW{{\cal{W}}}
\def\sq{{\mathscr{Q}}}
\def\sf{{\mathscr{F}}}
\def\fs{{\mathfrak{s}}}
\def\fa{{\mathfrak{a}}}

\def\fg{{\mathfrak{g}}}
\def\fq{{\mathfrak{q}}}
\def\ft{{\mathfrak{t}}}
\def\fw{{\mathfrak{w}}}

\def\sfa{{\mathsf{a}}}
\def\sfL{{\mathsf{L}}}

\def\d{\delta}
\def\ud{\updelta}

\def\c{\chi}
\def\z{\zeta}
\def\p{\phi}
\def\a{\alpha}
\def\k{\kappa}

\newcommand{\vs}{\varsigma}

\newcommand{\dg}{\delta g}
\def\P{\Phi}
\def\S{\Sigma}

\def\vr{\varrho}
\def\pa{\partial}

\theoremstyle{remark}
\newtheorem*{remark}{Remark}

\allowdisplaybreaks

\begin{document}


\title{Dynamical entropy and thermodynamics of causal diamonds}

\author{Avik Chakraborty}
\email{avikchakraborty589@gmail.com}
 
\author{Pratik K. Das}
\email{dask.pratik@gmail.com}
 
\author{Manavendra Mahato}%
 \email{manav@iiti.ac.in}
\affiliation{ Department of Physics,\\
 Indian Institute of Technology Indore, \\
 Simrol, Khandwa Road, 453552, Indore, India
}%


\begin{abstract}
\vspace{.3cm}
We study the dynamical entropy of causal diamonds using Wald-Zoupas charges and fluxes on their null boundaries. For a chosen boundary symplectic potential, the boost charge defines an entropy made of an area term plus an expansion correction, which reduces to one quarter of the area at the edge on the diamond. In Einstein gravity, we find an exact balance law relating the entropy on any two cuts of the null boundary to a signed flux, so the entropy decreases only when this flux is non-negative. We obtain a linearised physical process first law without assuming a background Killing symmetry, both in the presence and absence of external matter; and extend it for semiclassical gravity based on the improved quantum null energy condition. We also examine causal diamonds in maximally symmetric spacetimes and in Einstein-scalar gravity.
\end{abstract}

\maketitle



\section{Introduction}
\label{sec:intro}

The laws of black hole mechanics \cite{Bardeen:1973gs}, together with the Bekenstein entropy \cite{Bekenstein:1973ur, Bekenstein:1974ax} and Hawking radiation \cite{Hawking:1975vcx}, suggest a close relation between gravitational dynamics, geometry and entropy. The same
relation appears for cosmological horizons \cite{Gibbons:1977mu} and, more generally, for local causal horizons, from which the Einstein equation can be recovered as an equation of state \cite{Jacobson:1995ab} or as a condition of entanglement equilibrium \cite{Jacobson:2015hqa}. In a diffeomorphism invariant theory of gravity, the covariant phase space formalism \cite{Zuckerman:1986vzu, Crnkovic:1986ex, Crnkovic:1987tz,Lee:1990nz} gives this relation a nice interpretation: the entropy entering the first law of a stationary black hole is the Noether charge of the horizon-generating Killing field \cite{Wald:1993nt, Iyer:1994ys}. Extending this description to
evolving boundaries requires an entropy assigned to their individual cross-sections and a balance law describing its change. Laws of this kind have been formulated for trapped surfaces and dynamical horizons \cite{Hayward:1993wb, Ashtekar:2004cn}.A linearised version of the second law has been derived for higher curvature theories \cite{Wall:2015raa}, and a second law extending beyond linear order has been formulated for effective field theories of gravity \cite{Hollands:2022fkn}. It remains to understand how much of gravitational thermodynamics survives for a finite spacetime
region whose boundary need not admit any Killing symmetry.

We explore this question in the case of causal diamonds. A finite causal diamond is the region between two timelike separated points, bounded by portions of their future and past light cones \cite{Gibbons:2007nm}. Such regions, also called double cones, are the basic local domains of quantum field theory \cite{Dappiaggi:2010iq, Witten:2019qhl}. For a conformal field theory in flat space, the vacuum modular flow of a ball is the geometric flow of the conformal Killing field that preserves its diamond \cite{Casini:2011kv}. We consider diamonds whose null boundaries $\cN^-$ and $\cN^+$ meet at a smooth spacelike codimension-two edge $\cB$. The construction does not require an asymptotic region. For spherical diamonds in maximally symmetric spacetimes, Jacobson and Visser derived a comparison first law and a Smarr relation using the conformal Killing field that preserves the diamond \cite{Jacobson:2018ahi}. Their first law relates variations of
the edge area, the volume of the maximal spatial slice, the matter Hamiltonian, and the cosmological constant, and its thermodynamic interpretation assigns a negative temperature to the diamond \cite{Jacobson:2018ahi, JacobsonVisser:2019negative}. Related thermodynamic descriptions of light cones and stretched light cones are given in \cite{DeLorenzo:2017tgx, Barbosa:2026bmv, Svesko:2020yxo, Parikh:2017aas}. These results concern equilibrium configurations or particular slices. They motivate a description that applies to arbitrary cross-sections of the null boundary and keeps track of the flux through it.

A related development concerns the entropy of dynamical black holes. Hollands, Wald and Zhang constructed an entropy that satisfies a first law on arbitrary horizon cross-sections for nonstationary perturbations of a stationary black hole \cite{Hollands:2024vbe}. In Einstein gravity, this entropy adds an expansion correction to the area. Subsequently, it was interpreted as an improved Noether charge, and the associated physical process law was derived from the linearised Raychaudhuri equation \cite{Visser:2024pwz}. The construction has also been applied to perturbations of cosmological horizons \cite{Zhao:2025zny}, and related approaches to entropy balance on null hypersurfaces have been developed in \cite{Hopfmuller:2018fni, Rignon-Bret:2023fjq, Shajiee:2026coz}. These results make it useful to examine explicitly how such an entropy depends on the boundary symmetry, the boundary symplectic potential, and the background geometry, particularly when the background itself is not stationary.

Our starting point is the Wald-Zoupas prescription for localised charges and corresponding fluxes \cite{Wald:1999wa}. Its application to general null boundaries, with non-vanishing expansion and shear, was developed in \cite{Chandrasekaran:2018aop}. It fits into a broader analysis of gravitational subsystems, edge modes, and boundary symplectic structures \cite{Donnelly:2016auv, Speranza:2017gxd, Harlow:2019yfa} and of the phase space of null surfaces \cite{Hopfmuller:2016scf, Chandrasekaran:2020wwn}. For causal diamonds, the boundary symmetry algebra, the matching conditions across the edge and the associated conservation laws were established in \cite{Chandrasekaran:2019ewn}. In particular, the charge of the
constant boost at the edge is proportional to its area. We use this framework to study the entropy of cuts of the future boundary and the balance law between them. Because the background boundary is allowed to be nonstationary, the gravitational flux and its variations must be retained throughout.

For the boundary boost $\xi^a=\kappa_\beta v\ell^a$, where $v$ is an affine parameter vanishing at $\cB$ and $\kappa_\beta>0$, a constant, the Wald-Zoupas charge on a cross-section $\cC$ in Einstein gravity is given by
\begin{align}
 \sq_\xi[\cC]
 &=-\frac{\kappa_\beta}{2\pi}\cS[\cC],
 \nonumber\\
 \text{where} \quad
 \cS[\cC]
 &=\frac14\int_{\cC}\boep^+_{d-2}\,(1-v\vartheta).
\end{align}
Here $\ell^a$ is the future directed null generator, $\vartheta$ its expansion, and $\boep^+_{d-2}$ the positive area form in the orientation adopted below. At the edge, the expansion term vanishes and $\cS[\cB]=A(\cB)/4$, in agreement with \cite{Chandrasekaran:2019ewn,Jacobson:2018ahi}. Away from the edge, the entropy also depends on how the cross-sectional area changes along the generators, and it has the same functional form as the dynamical black hole entropy of \cite{Hollands:2024vbe}.

Before perturbing, we establish an exact balance law in Einstein
gravity. We denote the square root of the determinant of the codimension-two metric (of the codimension-two surface $\cC$) by $\sfa$, and call it the area density, as integrating over it gives the area of the codimension-two surface. The expansion term removes the first derivative of the area density $\sfa$: one has
$\partial_v [\sfa(1-v\vartheta)]=-v\,\partial_v^2\sfa$, so the entropy responds only to the focusing of the generators. Combined with the Raychaudhuri equation, this gives, for arbitrary and not necessarily uniform cuts $\cC_1$ and $\cC_2$ with $\cC_2$ to the future of $\cC_1$,
\beq
 -\frac{\kappa_\beta}{2\pi}\big(\cS[\cC_2]-\cS[\cC_1]\big)
 =\Delta\cE+\sf_\xi[\cN^+_{12}]
 =\int_{\cN^+_{12}}\varphi\,\boep^+_{d-1},
\eeq
where $\Delta\cE$ is the signed matter flux, $\sf_\xi$ the
gravitational Wald-Zoupas flux, and $\varphi$ a scalar flux density relative to the positive null measure. Consequently, the entropy is non-increasing between two cuts, and under every deformation of a cut supported in an open region, if and only if the corresponding flux is non-negative. The expansion makes $\varphi$ positive, while shear and matter obeying the null energy condition make it negative. The required positivity is therefore a genuine restriction rather than a consequence of the null energy condition, and in three dimensions it cannot hold for such matter at all. Nevertheless, in four or higher dimensions, the entropy always decreases overall, from $A(\cB)/4$ at the edge to zero at a regular corner.

For perturbations sourced by external matter, using the covariant phase space identity, we obtain the physical process version of the first law,
\beq
 -\frac{\kappa_\beta}{2\pi}\Delta\d\cS
 =\Delta\d\cE+\d\sf_\xi[\cN^+_{12}],
\eeq
where $\Delta$ compares two cuts and $\d$ denotes a perturbation on a fixed boundary structure and boost normalisation. This law holds for any metric theory in which the charge-entropy relation is available. The matter term is a signed flux defined with the same orientation as the charge balance law. Hence, non-negativity of the combined flux is equivalent to $\Delta\d\cS\le0$, with a corresponding local statement. In contrast with the stationary black hole setting
\cite{Hollands:2024vbe, Visser:2024pwz}, the boundary boost is
generally not a Killing field of the background, so the gravitational flux contributes already at first order, even in vacuum. At second order, the balance depends on the variation of the symplectic current as a background-dependent bilinear map and on the second order change of the metric along the chosen family of solutions. The reduction to the modified canonical energy of \cite{Hollands:2012sf, Hollands:2024vbe} requires the generator to be a Killing vector.

We also examine a semiclassical extension. The generalised second law \cite{Bekenstein:1973ur, Bekenstein:1974ax} has been established for semiclassical causal horizons \cite{Wall:2011gsl}. In the dynamical entropy framework, the quantum null energy condition \cite{Bousso:2015mna, Bousso:2015wca, Balakrishnan:2017bjg} combines with the physical process law to give a generalised second law involving a modified von Neumann entropy \cite{Hollands:2024vbe}. Since the null boundary of a causal diamond has a non-zero classical expansion, we instead use the improved quantum null energy condition (INEC) \cite{Ben-Dayan:2023inz, BenDayan:2026inec}. Under the assumptions needed for its perturbative and integrated versions, an expansion dependent matter entropy can be combined with the gravitational entropy in a semiclassical balance law. We keep the remainder that measures the departure from saturation of the assumed inequality, so that the condition for the reduction of the combined entropy is stated in terms of an effective flux that includes this remainder.

For spherical diamonds in maximally symmetric spacetimes, we compare this null boundary construction with the conformal Killing description of \cite{Jacobson:2018ahi}. The conformal Killing field $\zeta^a$ and the boundary boost $\xi^a$ have different profiles along the null generators, though they coincide at the edge. Keeping this distinction explicit, we derive the comparison identity and the Smarr relation on a generic spatial slice, including the contributions that vanish on the maximal slice. Finally, we extend the null boundary analysis to minimally coupled Einstein-scalar gravity \cite{Faulkner:2010fh, Anabalon:2023efw}. With the boundary potential chosen here, the scalar contributes to the symplectic flux while the entropy retains its geometrical form; its evolution depends on the scalar through the coupled field equations, and the scalar, like ordinary matter, acts to increase the entropy.

The paper is organised as follows. In \cref{sec:background}, we review the necessary background: covariant phase space formalism, Wald-Zoupas prescription, null boundary symmetries, and Gaussian null coordinates. In \cref{sec:dyn_entropy_cd} we obtain the local charges for causal diamonds and identify the dynamical entropy. Then, we derive the exact entropy balance in Einstein gravity, the physical process law, and the associated local and integrated entropy inequalities, including first and second order vacuum perturbations in  \cref{sec:physical_process_1st_law}.
In \cref{sec:generalise_second_law_cd}, we study the semiclassical
extension of the entropy inequality based on the INEC. We discuss the notion of dynamical entropy and thermodynamics of spherical
diamonds in maximally symmetric spacetimes in \cref{sec:CD_MSS}. Finally, in \cref{sec:EH-scalar_gravity}, we explore the null surface, local charges, fluxes, balance laws and thermodynamics in Einstein-scalar gravity. We conclude in \cref{sec:conclusion}. The appendices contain a review of the equilibrium thermodynamics of diamonds (\cref{sec:cd_mss}), an independent derivation from the Raychaudhuri equation (\cref{sec:dyn_entropy_raychaudhuri}), explicit examples (\cref{app:wz-examples}), the black hole entropy inequality using INEC (\cref{sec:bh_gsl}), and the symmetries of a null surface in Einstein-scalar gravity (\cref{sec:ES-null-symmetry}).

\subsection*{Notation and conventions}
We use $G=\hbar=c=k_B=1$ unless stated otherwise, metric signature $(-,+,\ldots,+)$ and the curvature convention
$[\nabla_a,\nabla_b]\cV^c=\cR^c{}_{dab}\cV^d$. The conventions below will be employed.
 \begin{itemize}
\item Any bold font Roman/Greek letter (e.g. $\bos{K}, \bos{\a}$) denotes a differential form in spacetime manifold $\cM$. 

\item  Tensors on the spacetime $\cM$ will be denoted by lowercase Roman abstract indices $a,\, b,\, c$, etc. from the first part of the alphabet.

\item Tensors on any hypersurface will be denoted by the lowercase Roman abstract indices $i,\, j,\, k$, etc. from the second part of the alphabet.

\item Tensors on any cross-section of a hypersurface will be denoted by the uppercase Roman indices $A,\, B,\, C,$ etc.
\end{itemize}
The main symbols are listed below.

\vspace{.3cm}
\begingroup
\small
\setlength{\parindent}{0pt}
\setlength{\parskip}{0pt}
\newcommand{\glossaryentry}[2]{%
  \par\noindent
  \parbox[t]{0.27\linewidth}{\raggedright\strut #1\strut}%
  \hfill
  \parbox[t]{0.69\linewidth}{\raggedright\strut #2\strut}%
  \par\addvspace{0.5ex}%
}
\noindent\rule{\linewidth}{0.4pt}\par
\glossaryentry{\textbf{Symbol}}{\textbf{Meaning}}
\noindent\rule{\linewidth}{0.4pt}\par\addvspace{0.5ex}

\glossaryentry{$\mathscr{D}_{p^-}^{p^+}$, $p^\pm$, $\cB$}
{Causal diamond, its corners, and its edge $\cB=\cN^+\cap\cN^-$.}

\glossaryentry{$\cN^\pm$, $\cC$, $\cN^+_{12}$}
{Future and past null boundaries; a cut; the segment of $\cN^+$
between two cuts.}

\glossaryentry{$\ell^a$, $n^a$, $v$}
{Future directed affine generators of $\cN^+$ and $\cN^-$, with
$\ell^an_a=-1$; affine parameter on $\cN^+$, with $v=0$ at $\cB$.}

\glossaryentry{$q_{AB}$, $\sfa$}
{Induced metric on a cut and its area density $\sqrt{\det q_{AB}}$.}

\glossaryentry{$\vartheta$, $\sigma_{AB}$, $c_d$}
{Expansion and shear of $\ell^a$; $c_d=(d-3)/(d-2)$.}

\glossaryentry{$\fs$}
{Boundary structure $[\ell^a,\k,\hat\ell_a]$, with $\k$ the
non-affinity.}

\glossaryentry{$\xi^a$, $\kappa_\beta$}
{Boundary boost, $\xi^a=\kappa_\beta v\,\ell^a$ on $\cN^+$, with
constant $\kappa_\beta=-\beta>0$.}

\glossaryentry{$\bth$, $\bom$, $\boq_\xi$}
{Symplectic potential, symplectic current and Noether charge.}

\glossaryentry{$\bTh$, $\bos{\vr}$}
{Wald-Zoupas boundary potential, $\bTh=\ul{\bth}-\delta\bos{\vr}$.}

\glossaryentry{$\sq_\xi[\cC]$, $\sf_\xi[\cN_{12}]$}
{Wald-Zoupas charge of a cut and flux through a segment.}

\glossaryentry{$\cS[\cC]$, $\bots$}
{Dynamical entropy
$\frac14\int_\cC\boep^+_{d-2}(1-v\vartheta)$ and its $(d-2)$-form.}

\glossaryentry{$T_{(\beta)}$}
{Diamond temperature $-\kappa_\beta/(2\pi)$, with
$\sq_\xi=T_{(\beta)}\cS$.}

\glossaryentry{$\bos e$, $\Delta\cE$, $\Delta\cE_{\rm out}$}
{Signed matter flux form $-\xi^bT_b{}^c\epsilon_c$, its integral, and
the outward flux $-\Delta\cE$.}

\glossaryentry{$\varphi$, $\varphi_\delta$}
{Signed flux densities of $\bos e+\bos{\sf}_\xi$ and of its variation,
relative to $\boep^+_{d-1}$.}

\glossaryentry{$\ud$, $\bog$}
{Derivative along a one-parameter family of metrics; second order
gravitational flux form.}

\glossaryentry{$\zeta^a$, $\alpha$, $\kappa$}
{Diamond-preserving conformal Killing vector, its conformal factor,
and its surface gravity.}

\glossaryentry{$\Sigma_t$, $K$, $V_\zeta$}
{Spatial slices of a spherical diamond; mean curvature of $\cB$;
conformal Killing volume.}

\glossaryentry{$\cS_{vN}$, $\cS_q$, $\mathsf R$}
{Von Neumann entropy, modified matter entropy, and remainder of the
quantum energy bound.}

\glossaryentry{$\Phi$, $U(\Phi)$}
{Minimally coupled scalar field and its potential.}

\noindent\rule{\linewidth}{0.4pt}\par
\endgroup

\section{Background}\label{sec:background}
This section briefly reviews the geometric and covariant phase space ingredients necessary for the subsequent analysis. We begin by reviewing the covariant phase space formalism and the Wald–Zoupas prescription for defining localised charges and fluxes in the presence of null boundaries. We then summarise the intrinsic and extrinsic geometry of a generic null hypersurface, together with the corresponding boundary symmetry algebra and the associated charges. Finally, we introduce Gaussian null coordinates, which furnish a convenient local parametrisation of a null boundary and will be employed throughout the analysis of causal diamonds. Readers already familiar with these constructions may proceed directly to  \cref{sec:dyn_entropy_cd}.

\subsection{Covariant phase space and Wald-Zoupas prescription} \label{sec:cps_review}
The covariant phase space formalism \cite{Zuckerman:1986vzu, Crnkovic:1986ex, Crnkovic:1987tz, Lee:1990nz, Iyer:1994ys, Wald:1993nt} is a framework for formulating Hamiltonian mechanics in which the phase space is defined as the space of solutions to the equations of motion, rather than being constructed from generalised coordinates and their conjugate momenta. This formalism is manifestly covariant, hence suitable for analysing the phase space and associated symplectic structures of diffeomorphism invariant (and gauge invariant) theories. Here, we will briefly review the Iyer-Wald method \cite{Iyer:1994ys} for constructing the symplectic form, Noether charge, and global charges. Then, we will proceed to the Wald-Zoupas prescription \cite{Wald:1999wa} to define local charges and their associated fluxes.

In general, the canonical phase space is considered to be a manifold $\mathsf{P}$ with a symplectic $2$-form $\Omega$, and the spacetime is a $d$-dimensional manifold $\cM$. Then, from the symplectic form, one can define the Hamiltonian. In the covariant phase space (CPS) formalism, one first considers the space of all possible allowed field configurations, the (field) configuration space $\cF$, or the field space in short. Then, the subspace of the fields that satisfy the equations of motion gives the corresponding covariant phase space $\cP$.\footnote{More precisely, this subspace may have a degenerate symplectic form, called the \emph{pre-symplectic form}. In literature, this subspace is referred to as the \emph{pre-phase space} \cite{Harlow:2019yfa} or the \emph{solution space} \cite{Margalef-Bentabol:2020teu}. Degeneracy of the (pre)symplectic form indicates the presence of gauge symmetries. Quotienting the corresponding gauge group from the pre-phase space yields the covariant phase space with a non-degenerate symplectic structure. In the following, we will be ignorant about the quotient process and consider the phase space same as the solution space (which is true in the absence of gauge symmetries).} The construction is as follows. One begins with the action for a general theory 
\begin{equation}
	S[\phi] = \int_{\cM} \bos{\sfL}(\phi),
\end{equation}
where, $\bos{\sfL}$ is the Lagrangian $d$-form for any $d$-dimensional theory, and $\phi$ denotes all the dynamical fields collectively. Variation of this action leads to
\begin{equation}
	\delta \bos{\sfL}[\phi] = \bos{E}_{\phi}  \delta \phi + \mathrm{d} \bth (\p; \d\p),
\end{equation}
where $\delta \phi$ represents a generic field perturbation and provides the basis of (tangent space of) the configuration space $\cF$. So, `$\mathrm{d}$' denotes the exterior derivative over the spacetime $\mathcal{M}$ and `$\delta$' represents the exterior derivative over the field space $\cF$.\footnote{We will consider the field independent diffeomorphisms throughout the work, i.e., $\d\xi=0$ and therefore $\d(\cL_\xi \p) = \cL_\xi \d\p$.} The equations of motion are given by the set of equations $\bos{E}_{\phi} = 0$. The quantity $\bth$ is called the \emph{symplectic potential}. It is a $(d-1)$-form in spacetime $\mathcal{M}$ and one form on the phase space, i.e., $\bth$ is a $(d-1;1)$-form.\footnote{In this notation, a $(p;q)$-form means it is a $p$-form in spacetime and $q$-form in field space.} One can now introduce the \emph{Noether current} $(d-1)$-form,
\begin{equation}\label{noethercurrent}
	\bos{J}_{\xi} \equiv \bth (\p; \cL_\xi\p) - \xi\i \bos{\sfL} 
\end{equation}
with $\xi^a$ being the diffeomorphism generator, $\mathcal{L}_\xi$ is the Lie derivative w.r.t it and `$\i$' denotes the interior product (i.e., contraction of $\xi^a$ with the first index of the form $\bos{\sfL}$). A simple check shows that this current is a closed form in spacetime on-shell, and locally, it can be made an exact form (by Poincaré's lemma), which leads to the definition of the \emph{Noether charge}, 
\begin{equation}\label{noethercharge}
	\bos{J}_{\xi} \approx \mathrm{d} \bos{ Q}_{\xi}.
\end{equation}
Here, `$\approx$' denotes the equality is on-shell. Apart from the Noether current, one can define another current $(d-1)$-form called \emph{(pre)symplectic current} as
\begin{equation}\label{sympcurrent}
	\bom(\p; \d_{1} \p, \d_{2} \p) = \delta_{1} \bth(\p; \d_2\p) - \delta_{2} \bth(\p; \d_1\p).
\end{equation}
Here, $\delta_{1} \phi$ and $\delta_{2} \phi$ are two arbitrary field perturbations independent of each other. It is evident that $\bom$ is a $2$-form in the field space, or more compactly a $(d-1;2)$-form. By construction, $\bom$ is closed in field space ($\d\bom=0$). It is also closed in spacetime, $\ed\bom=0$, when the background satisfies the equations of motion and the perturbations satisfy the linearised equations; by Poincar\'e's lemma it is then locally exact. The symplectic current is used to construct the symplectic form by integrating it over a codimension-one spacelike (Cauchy) surface $\Sigma$,
\begin{equation}
	\Omega (\p; \d_{1} \p, \d_{2} \p) = \int_{\Sigma} \bom(\p; \d_{1} \p, \d_{2} \p).
\end{equation}
After quotienting by its degeneracy directions, $\Omega$ is non-degenerate; it is independent of the choice of $\Sigma$ provided no symplectic flux leaves through the boundary of the region considered. If $\mathcal{F}$ is the set of all possible field configurations, i.e., the field configuration space, then the covariant phase space $\cP\subset \cF$ is the submanifold of $\cF$ consisting of fields that satisfy the equations of motion and has the symplectic structure $\Omega$.

 For diffeomorphism invariant theories, covariance tells that $\delta_\xi \phi = \mathcal{L}_\xi \phi$, where $\xi$ is a diffeomorphism generator. So the (pre)-symplectic current becomes,
\begin{equation}\label{symp_current_diffs}
	\bom(\p; \d \p, \cL_\xi \p) = \delta \bth(\p; \cL_\xi\p)- \mathcal{L}_\xi \bth (\p; \d\p).
\end{equation}
Now, the symplectic current can also be expressed in terms of the Noether charge as,
\beq
 \bom (\p; \d \p, \cL_\xi \p) = \ed\left( \delta \boq_\xi - \xi\i\bth (\p; \d\p)\right).
\eeq
The \emph{global charge} associated with the diffeomorphism generator $\xi^a$ can then be written as,
\begin{equation}\label{hamiltonianIW}
	\delta \cQ_\xi = \Omega(\p; \d \p, \cL_\xi  \p) = \int_{\pa\Sigma} \left( \delta \boq_\xi - \xi\i\bth(\p; \d\p) \right).
\end{equation}
In many earlier literatures, the global charge has been denoted $H_\xi$ and called the Hamiltonian. When the diffeomorphism generator $\xi^a$ is timelike, the above charge is indeed the Hamiltonian.
In general, the global charge $\cQ_\xi$ may not exist. For two independent perturbations $\delta_1$ and $\delta_2$, the integrability condition is given by
\beq
 \delta_1 \delta_2 \cQ_\xi - \delta_2 \delta_1 \cQ_\xi =0\,.
\eeq
But using the definition of global charges, one gets
\begin{align}\label{eq:global_charge_nonintegrability}
    \delta_1 \delta_2 \cQ_\xi - \delta_2 \delta_1 \cQ_\xi &= \int_\cC \delta_1 \left( \delta_2 \boq_\xi - \xi\i\bth(\p;\delta_2\p) \right)-\int_\cC \delta_2 \left( \delta_1 \boq_\xi - \xi\i\bth(\p;\delta_1\p) \right) = -\int_\cC \xi\i\bom(\p;\delta_1\p,\delta_2\p) \,.
\end{align}
In general, `$\xi\i\bom$' does not vanish, implying that the charge may not be integrable. The above relation holds for all $\delta_1 g,\delta_2 g\in \cT_g\cP$, with $\cC=\pa\Sigma$. 

 To perform the Wald-Zoupas (WZ) procedure \cite{Wald:1999wa}, one needs to take the pullback of the symplectic current to a hypersurface $\cN$,
 \beq
 \Pi^* \bom (\p; \delta_1 \p, \delta_2\p) = \bom \left(\, \underline{\p}; \underline{\delta_1 \p}, \underline{\delta_2 \p} \,\right) = \underline{\bom} \,, \qquad  \forall\, \p \in \cP \nonumber
 \eeq
where $\p$ is any dynamical field,  and $\underline{\p}$'s are the fields on $\cN$ induced by $\p\in\cP$. We will use the shorthand notation $\underline{\bom}$ to denote the pullback of the symplectic current to the null hypersurface. One can construct a (pre)-symplectic potential $\bTh(\underline{\p};\underline{\delta\p})$ on $\cN$ locally and covariantly\footnote{$\bTh$ depends on the dynamical fields $\p$, its variation $\delta\p$, the boundary $\cN$, and an universal structure on $\cN$ denoted by $\fs$ (we will discuss this later), $ \bTh = \bTh (\p, \delta\p, \fs, \cN)$. Consider the diffeomorphism of the spacetime: $ \psi: \cM \rightarrow \cM $. Locality and covariance require that for any such diffeomorphism, its pullback to $\cN$ satisfies the following
\beq
 \psi^* \bTh (\p, \delta\p, \fs, \cN) = \bTh (\psi^*\p, \psi^*\delta\p, \psi^*\fs, \psi^{-1}(\cN)). \nonumber
\eeq
As we will be focused on the case when the boundary and its universal structure are preserved, the above condition reduces to
\beq
 \psi^* \bTh (\p, \delta\p, \fs, \cN) = \bTh (\psi^*\p, \psi^*\delta\p,\fs, \cN). \nonumber
\eeq
} such that
\beq
 \underline{\bom} = \delta_1 \bTh (\underline{\p};\underline{\delta_2\p}) - \delta_2 \bTh (\underline{\p};\underline{\delta_1\p})\,.
\eeq
Now, one can add the term $\xi\i\bTh$ to the integrand on the
right hand side of \cref{hamiltonianIW} to define the variation of a
\emph{localised charge},
\beq\label{variation_local_charges}
 \delta \sq_\xi[\cC]
 =
 \int_\cC \delta\bos{\sq}_\xi
 =
 \int_\cC
 \left(
  \delta\boq_\xi
  -\xi\i\bth
  +\xi\i\bTh
 \right).
\eeq
Here, the bulk forms are understood to be pulled back to the
boundary. It is easy to check that
\begin{align*}
     \delta_1 \delta_2 \sq_\xi - \delta_2 \delta_1 \sq_\xi  = 0 \,. 
\end{align*}
 As the integrability condition is satisfied, it is possible to define the charge $\sq_\xi[\cC]$ up to integration constants on the covariant phase space $\cP$. One can fix the integration constant by demanding \cite{Chandrasekaran:2018aop, Wald:1999wa}
\beq
 \sq_\xi[\cC] \Big|_{\p=\p_0} = 0 \qquad \forall \,\xi^a \in \mathsf{G} \;\text{and}\; \forall \,\cC\,.
\eeq
Here $\mathsf{G}$ denotes the group of symmetries that preserve $\cN$, and $\p_0$ is some reference solution (e.g., solution in a stationary background). Still, the charge $\sq_\xi$ is not unique as the boundary symplectic potential is ambiguous up to
\beq
 \bTh (\underline{\p};\underline{\delta\p}) \rightarrow \bTh(\underline{\p};\underline{\delta\p}) + \delta \bos{Z}(\underline{\p})\,,
\eeq
where $\bos{Z}$ is a $(d-1)$-form on $\cN$ locally constructed from the dynamical fields and the boundary structure $\fs$. Therefore, additional condition(s) need to be imposed to uniquely fix $\bTh$ and, hence, the charge $\sq_\xi$. One such convenient condition is
\beq\label{stationary_boundary_symplectic_potential}
 \bTh (\underline{\p};\underline{\delta\p}) = 0 \qquad \forall\,\underline{\p}=\p_0  \in \cP\,.
\eeq
And this choice ensures that $\delta\bos{Z}$ vanishes for all $\delta\p\in\cP$ whenever the solution is stationary \cite{Wald:1999wa}.
If we choose a portion of $\cN$ bounded by the cross-sections
$\cC_1$ and $\cC_2$, with $\cC_2$ to the future of $\cC_1$,
we orient this portion, denoted by $\cN_{12}$, such that
\beq
 \partial\cN_{12}=\cC_2\cup(-\cC_1).
\eeq
The cross-sections carry the area orientation specified below.
Then, the difference between the charges of these cross-sections
is defined by
\beq
 \Delta\delta\sq_\xi[\cN_{12}] =
 \delta\sq_\xi[\cC_2]-\delta\sq_\xi[\cC_1].
\eeq
For fields and perturbations satisfying the equations of motion
and their linearisation, respectively, Stokes' theorem gives
\begin{align*}
 \Delta\delta\sq_\xi[\cN_{12}]
 &=
 \int_{\cN_{12}}
 \ed\left(
  \delta\boq_\xi-\xi\i\bth+\xi\i\bTh
 \right)
 \\
 &=
 \int_{\cN_{12}}
 \left[
  \ul{\bom}(\p;\delta\p,\cL_\xi\p)
  +\ed(\xi\i\bTh)
 \right].
\end{align*}
Using Cartan's formula, one gets $\ed(\xi\i\bTh)=\cL_\xi\bTh$, since the intrinsic exterior derivative of $\bTh$ vanishes on $\cN$. Substituting this in the above expression yields
\begin{align*}
 \Delta\delta\sq_\xi[\cN_{12}]
 &=
 \int_{\cN_{12}}
 \left[
  \ul{\bom}(\p;\delta\p,\cL_\xi\p)
  +\cL_\xi\bTh(\p;\delta\p)
 \right]
 \\
 &=
 \int_{\cN_{12}}
 \delta\!\left[\bTh(\p;\cL_\xi\p)\right].
\end{align*}
The last equality follows from the defining relation for the boundary symplectic potential, its covariance, and $\delta\xi^a=0$. The above quantity measures the variation of the flux through $\cN_{12}$. Choosing the boundary symplectic potential to vanish on stationary solutions ensures that the flux also vanishes on those solutions. Therefore, one can identify the flux of the WZ charges as
\beq\label{sympectic_flux}
 \sf_\xi[\cN_{12}]
 =
 \int_{\cN_{12}}\bos{\sf}_\xi,
 \qquad
 \bos{\sf}_\xi=\bTh(\p;\cL_\xi\p).
\eeq
Its variation satisfies
\beq
 \delta\sf_\xi[\cN_{12}]
 =
 \int_{\cN_{12}}\ul{\bom}(\p;\delta\p,\cL_\xi\p)
 +
 \int_{\partial\cN_{12}}\xi\i\bTh(\p;\delta\p).
\eeq
With a compatible choice of the reference charges, the
conservation (flux balance) law is
\beq
 \sf_\xi[\cN_{12}]
 =
 \Delta\sq_\xi[\cN_{12}]
 =
 \sq_\xi[\cC_2]-\sq_\xi[\cC_1].
\eeq

\subsection{Null surface geometry, symmetries and charges} \label{sec:null_surface_review}
Here, we will briefly discuss the structure of a null hypersurface, its symmetries and the corresponding conserved charges associated with the symmetries. A nice and detailed review of null hypersurfaces can be found in \cite{Gourgoulhon:2005ng}. We will mainly follow \cite{Chandrasekaran:2018aop, Chandrasekaran:2019ewn} but with a different convention. 

 We consider the spacetime manifold $\cM$ to be equipped with the metric $g_{ab}$ and a null hypersurface $\cN$ with topology $\cZ \times \mathbb{R}$. In this section, we will consider that spacetime has four dimensions. Generalisation to higher dimensions is straightforward. Let $\ell_a$ be a choice of future directed null normal to $\cN$. This 1-form is associated with the null vector\footnote{One can define a null surface by demanding that any vector field $\ell^a \in \cT(\cM)$ that is normal to $\cN$ to be null w.r.t the metric:
 \begin{equation*}
     g_{ab} \ell^a \ell^b = \ell^a \ell_a = 0\,.
 \end{equation*}
 These normal vectors are both tangent and orthogonal to the null surface.} $\ell^a$ via the standard metric duality
 \beq\label{normal_covector}
 \ell_a \nulleq g_{ab} \ell^a \,,
 \eeq
where `$\nulleq$' indicates that the equality holds only on the null hypersurface. The choice of $\ell_a$ is not unique, but can be rescaled as
 \beq\label{rescaling}
 \ell_a \rightarrow e^{\rho} \ell_a\, ,
 \eeq
where $\rho$ is some scalar on $\cN$. One can define non-affinity $\k$, a function on $\cN$ as follows:
 \beq\label{non-affinity}
 \ell^a \nabla_a \ell^b \nulleq \k \, \ell^b .
 \eeq
Under the rescaling \eqref{rescaling}, the non-affinity transforms as
\beq
 \k \rightarrow e^{\rho} \left( \k \,+ \cL_{\ell} \rho \right)\, . 
\eeq 
Next, we define the pullback as a map
\begin{align}
    \Pi : \cT^*(\cM) &\rightarrow \cT^*(\cN)  \\
     \eta_a  &\mapsto \Pi^{a}_i\, \eta_a \,,\nonumber
\end{align}
where $\eta_a$ is some one-form on the manifold. Notice that, by definition, the pullback of $\ell_a$ vanishes,
\beq
 \Pi^a_i \ell_a = 0\,.
\eeq
Using the pullback, one can define the induced metric on the null hypersurface as,
\beq
 q_{ij} = \Pi^a_i \,\Pi^b_j \,g_{ab}\,.
\eeq
The induced metric is degenerate with signature $(0,+,+)$. As $\ell_a$ is normal to $\cN$, we have
\beq
 q_{ij} \ell^j = 0\,.
\eeq
It implies that the induced metric can be regarded as a tensor on $\cZ$ (specifically, it lies in the subspace of the cotangent space $\cT_p^* (\cN)$). Therefore, the metric can be written as $q_{AB}$, and its inverse as $q^{AB}$. Similarly, one can define the second fundamental form (or, the extrinsic curvature) of $\cN$ as follows
\beq\label{2nd_fund_form}
 K_{ij} = \Pi^a_i \, \Pi^b_j \, \nabla_a \ell_b = \dfrac{1}{2} \cL_\ell q_{ij}\,.
\eeq
It has the property $K_{[ij]}=0$, and is normal to $\ell^i$: $\ell^i K_{ij} = 0$. This implies that the second fundamental form also resides on $\cZ$, and can be rewritten as $K_{AB}$. It can be decomposed in the standard way as
\beq
 K_{ij} = \dfrac{1}{2} \vartheta q_{ij} + \sigma_{ij} ~ \implies ~ K_{AB} = \dfrac{1}{2} \vartheta q_{AB} + \sigma_{AB},
\eeq
where $\vartheta$ is the expansion of the surface and $\sigma_{AB}$ is the shear which is traceless: $q^{AB}\sigma_{AB} = 0$. Now, we need to introduce another important quantity, called the \emph{Weingarten map}, which is defined as\footnote{Notice that we cannot write the Weingarten map as $\c_i^{\ j}= \nabla_i \ell^j$, since there is no notion of unique covariant derivative present on $\cN$ because of the degeneracy of the induced metric $q_{ij}$.}
\beq\label{eq:Weingarten_map}
 \c_i^{\ j} =  \Pi_i^a \nabla_a \ell^j\,.
\eeq
Sometimes it is also called the \emph{shape operator}. It is related to the second fundamental form as follows
\beq
 K_{ij} = \c_i^{\ k} q_{jk}\,.
\eeq
Another important relation of the Weingarten map is
\beq\label{weingarten_non-affinity}
\c_i^{\ j}\ell^i=\k\,\ell^j\,.
\eeq
As we will be working on the null hypersurface, its volume form
is also required to be defined. The volume form of the spacetime
$\epsilon_{abcd}$ determines the volume form of $\cN$.
We consider the $3$-forms $\bar{\epsilon}_{abc}$ satisfying
\beq\label{volume_form_relation}
 -4\bar{\epsilon}_{[abc}\ell_{d]}
 \nulleq \epsilon_{abcd}.
\eeq
Notice that these $3$-forms are not unique. However, their
pullback defines a unique volume form on $\cN$,
\beq
 \epsilon_{ijk}
 =
 \Pi_i^a\Pi_j^b\Pi_k^c\bar{\epsilon}_{abc}.
\eeq
For an auxiliary null vector $n^a$ satisfying
$n^a\ell_a=-1$,\footnote{In the orthonormal basis formalism, one chooses an auxiliary vector field $n^a$ on $\cN$ such that
\beq\label{n_a_conditions}
 n^a n_a =0; \qquad n_a \ell^a =-1\,. \nonumber
\eeq
The pullback of the $1$-form $n_a$ yields the $1$-form $n_i$ on $\cN$,
\beq
n_i = \Pi_i^a n_a \nonumber
\eeq
which is used in the text. As mentioned in the main text, we have freedom to choose $n_i$, but along with conditions \eqref{n_a_conditions} and $n_i \ell^i=-1$, the auxiliary normal $n_a$ is uniquely fixed.} this is equivalently
\beq
 \epsilon_{ijk}
 =
 -\Pi_i^a\Pi_j^b\Pi_k^c n^d\epsilon_{dabc}.
\eeq
We define the area form of a cross-section by
\beq
 \epsilon_{ij}=\epsilon_{ijk}\ell^k.
\eeq
Thus, the orientation of $\cN$ is compatible with the
future directed generator $\ell^a$ and the chosen area
orientation of its cross-sections.
For a vector $\cV^i$ tangent to $\cN$, restriction to a
cross-section gives
\beq
 \left.\left(\epsilon_{ijk}\cV^k\right)\right|_\cC
 = -\epsilon_{ij}n_k\cV^k,
\eeq
where the pullback $n_i$ is chosen to be normal to $\cC$
and satisfies $n_i\ell^i=-1$.
In the case of a stationary null hypersurface (for which the expansion and shear of the surface vanish), the extrinsic curvature vanishes: $K_{ij}=0$. Then one can define the \emph{rotation $1$-form} as
\beq\label{rotation_1-form}
 w_i = - \c_i^{\ j} n_j\,.
\eeq
 It is evident that the rotation $1$-form depends on the choice of $n_i$, except for stationary surfaces. From \cref{weingarten_non-affinity}, it follows that $ w_i \ell^i =\k$.

\subsubsection{Symmetries of \texorpdfstring{$\cN$}{TEXT}}
The symmetry of the null hypersurface can be interpreted as a \emph{diffeomorphism} which preserves the structure of $\cN$ as follows
\begin{align*}
    \tilde\psi: \cN &\rightarrow \cN \\
     \ell^i &\mapsto e^\beta \ell^i \\
     \k &\mapsto e^\beta (\k+\cL_\ell \beta)\,,
\end{align*}
where $\beta$ is some smooth function on $\cN$. If $\xi^a$ generates a one-parameter family of such diffeomorphisms, then to linear order in $\beta$ these conditions read
\beq
\begin{aligned}
     \cL_\xi \ell^i &= \beta \ell^i \,;\\
     \cL_\xi \k &= \beta \k+\cL_\ell \beta\,.
\end{aligned}
\eeq
\\
\paragraph{Boundary structure}
Here we will briefly discuss the geometric structure of the hypersurface $\cN$, which was called the \emph{boundary structure} in \cite{Chandrasekaran:2018aop}. One considers the triples $(\ell^a,\k,\hat{\ell}_a)$ of fields on the hypersurface, where $\ell^a$ is a smooth, nowhere vanishing vector field, $\k$ is a smooth function on $\cN$, and $\hat{\ell}_a$ is a choice of normal $1$-form (covector)
to $\cN$.\footnote{Earlier $\hat{\ell}_a$ was denoted as $\ell_a$. Here, a hat $(\,\widehat{.}\,)$ has been put to emphasise that the boundary structure need not have a metric. In that case, $\ell^a$ and $\hat{\ell}_a$ are independent of each other.} Also, the following condition needs to be satisfied
\beq
 \ell^a \hat{\ell}_a \nulleq 0\,.
\eeq
One can define two such triples $(\ell^a,\k,\hat{\ell}_a)$ and $(\ell'^a,\k',\hat{\ell}'_a)$ to be equivalent if they are related by rescaling as
\begin{align*}
    \ell'^a &\nulleq e^\gamma \ell^a \\
    \k' &\nulleq e^\gamma (\k + \cL_\ell \gamma)\\
    \hat{\ell}'_a &\nulleq e^\gamma \hat{\ell}_a \, 
\end{align*}
for some smooth function $\gamma$ on $\cN$. Then, the equivalence class associated with $(\ell^a,\k,\hat{\ell}_a)$ is denoted by 
\beq
 \fs = \left[ \ell^a,\k,\hat{\ell}_a \right] \,. \nonumber
\eeq
A choice of the equivalence class is the desired boundary structure on $\cN$. A metric $g_{ab}$ on $\cM$ for which the boundary (hypersurface) $\cN$ is null, determines a unique boundary structure $\fs$. To realise it, recall that for a choice of the normal covector $\hat{\ell}_a$, we can find the normal vector $\ell^a$ via the metric duality \eqref{normal_covector}, and the function $\k$ is defined by the \cref{non-affinity}. So, we end up with a triple $(\ell^a,\k,\hat{\ell}_a)$ from the equivalence class $\fs$.\\

\paragraph{Covariant phase space of boundary structure}
Let $\cN$ be a fixed hypersurface in $\cM$, equipped with the boundary structure $\fs=[\ell^a,\k,\hat{\ell}_a]$. One defines the associated field configuration space $\cF_\fs$ as the set of smooth Lorentzian metrics $g_{ab}$ on $\cM$ satisfying
\begin{equation}
 \hat{\ell}_a\nulleq g_{ab}\ell^b,
 \qquad
 \ell^a\nabla_a\ell^b\nulleq\k\ell^b,
 \label{eq:boundary-compatible-metrics}
\end{equation}
where $\nabla_a$ is the Levi-Civita connection of $g_{ab}$. These are kinematical conditions; the metrics in $\cF_\fs$ are not required to satisfy the equations of motion.

In this construction, the boundary structure is specified before the metric. In particular, $\ell^a$ and $\hat{\ell}_a$ are initially prescribed independently, and the first condition in \cref{eq:boundary-compatible-metrics} requires an admissible metric to relate them by metric duality. Since $\hat{\ell}_a$ is normal to $\cN$ and $\ell^a\hat{\ell}_a\nulleq0$, this condition gives
\begin{equation}
 g_{ab}\ell^a\ell^b \nulleq\ell^a\hat{\ell}_a \nulleq0.
\end{equation}
Thus $\cN$ is null with respect to every metric in $\cF_\fs$, and $\ell^a$ is both tangent and normal to it. The second condition fixes the non-affinity of its generators to the prescribed function $\k$.

Fixing $\fs$ does not fix the intrinsic metric $q_{AB}$ on the cross-sections or the expansion and shear of the null generators. These remain field-dependent quantities and may vary within $\cF_\fs$. The construction therefore, accommodates dynamical null boundaries while providing a common boundary structure against which their variations
and symmetries can be defined.

Let $\cF_0$ denote the larger space of smooth Lorentzian metrics for which $\cN$ is null, without prescribing its induced boundary structure. Then $\cF_\fs\subset\cF_0$.
Choosing $\fs$ restricts the allowed boundary data and fixes part of the diffeomorphism freedom. It should not be identified with quotienting by all diffeomorphisms.
In particular, diffeomorphism redundancy is not a property unique to $\cF_0$; residual diffeomorphisms may also act within $\cF_\fs$.
The corresponding solution space is obtained by imposing the gravitational field equations:
\begin{equation}
 \cP_\fs := \left\{  g_{ab}\in\cF_\fs: \boe^{ab}(g)=0 \right\}
 \subset\cF_\fs.
 \label{eq:boundary-covariant-phase-space}
\end{equation}
We refer to $\cP_\fs$ as the covariant phase space, with the understanding that it is the space of solutions before quotienting by
degeneracies of the presymplectic form.

A variation $h_{ab}=\delta g_{ab}$ is a tangent vector to $\cF_\fs$, rather than another element of $\cF_\fs$. It is obtained by differentiating a one-parameter family of admissible metrics
while holding $\cN$ and a representative of $\fs$ fixed. Consequently, it satisfies the linearised versions of \cref{eq:boundary-compatible-metrics}. If the family also lies in $\cP_\fs$, then $h_{ab}$
additionally satisfies the linearised field equations,
\begin{equation}
 \delta\boe^{ab}(g;h)=0.
\end{equation}

\subsubsection{Charges and fluxes at \texorpdfstring{$\cN$}{TEXT}}\label{sec:charges_N}

In this part, we will discuss the construction of (local) charges and their fluxes on a null hypersurface using the Wald-Zoupas (WZ) prescription \cite{Wald:1999wa, Chandrasekaran:2018aop}. To apply it, one needs to consider a linearised perturbation in the metric,
\beq\label{metric_perturbation}
 g_{ab} \rightarrow g_{ab} + \underbrace{h_{ab}}_{\delta g_{ab}}.
\eeq
As the WZ prescription uses the CPS formalism to construct charges, one needs to assume $g_{ab}\in \cF_\fs$ and $h_{ab}\in \cT_g\cF_\fs$ such that the symmetries corresponding to these charges preserve the null hypersurface. This in turn implies the relations \eqref{normal_covector} and \eqref{non-affinity} satisfied by the metric. The perturbation $h_{ab}$ also needs to satisfy the following conditions
\begin{subequations}\label{perturbation_conditions}
\beq
 h_{ab} \ell^b \nulleq 0,
\eeq
and
\beq
 \nabla_c (h_{ab}\ell^a \ell^b) \nulleq 0.
\eeq
\end{subequations}
These conditions can be derived starting from the relations mentioned for metric and varying them. Any linearised metric perturbation that satisfies the above two conditions works well for the WZ prescription. One such choice of the variation is
\beq
 h_{ab} = \cL_\xi g_{ab},
\eeq
where $\xi^a$ is a representative of the symmetry generators of $\cN$. The conditions \eqref{perturbation_conditions} imply that one can restrict $h^{ab}$ to the null surface as an intrinsic tensor $h^{ij}$ and construct
\beq
 h_i^{\ j}=q_{ik} h^{jk}\,;\qquad h_{ij}= q_{ik} q_{jm} h^{km}= \Pi_i^a \Pi_j^b h_{ab}\,.
\eeq
They satisfy $h_{ij}\ell^i \nulleq 0 \nulleq h_i^{\ j} \ell^i$. So, one can define a rank-$1$ tensor as $h_a := h_a^{\ b} \hat{\ell}_b$ which vanishes on $\cN$. It follows that there exists a $1$-form $\Gamma_a$ on the null surface satisfying
\beq
 \nabla_{[a}h_{b]} \nulleq \hat{\ell}_{[a}\Gamma_{b]}\,.
\eeq
Notice that $\Gamma_a$ is independent of the background metric, though it depends on $h_a^{\ b}$ and its derivatives. $\Gamma_a$ can be uniquely determined (up to a transformation along $\hat{\ell}_a$) by imposing $\nabla_{[a}\hat{\ell}_{b]} = 0$. Nevertheless, pullback of the 1-form is unique: $\Gamma_i = \Pi_i^a\Gamma_a$ and is invariant under the rescaling \eqref{rescaling}. The pullback also satisfies
\beq
 \Gamma_i \ell^i = 0\,.
\eeq
Having most of the ingredients required for the WZ prescription as discussed above, it is a good point to look at different charges on the null hypersurface.\\

\paragraph{Noether charge}
Here, we will review the Noether charge for general relativity at a null boundary. In four spacetime dimensions, we have the Lagrangian $4$-form \cite{Iyer:1994ys},
\begin{subequations}
\beq
 \sfL_{abcd} = \dfrac{1}{16\pi} \epsilon_{abcd} \cR,
\eeq
the symplectic potential $3$-form
\beq\label{symplectic_potential}
 \theta_{abc} = \dfrac{1}{16\pi} \epsilon_{abc}^{\ ~\ d} \left( g^{ef}\nabla_d h_{ef} - \nabla^e h_{de} \right),
\eeq
and the Noether charge $2$-form
\beq
Q_{\xi_{ab}} = -  \dfrac{1}{16\pi} \epsilon_{abcd} \nabla^c \xi^d ,
\eeq
\end{subequations}
 where $\cR$ is the Ricci scalar. Now, as we are interested in the charges at $\cN$, we need to restrict to $\cN$. Therefore, the pullback of the Noether charge form gives
\begin{align*}
 Q_{\xi_{ij}}
 &=  \dfrac{1}{16\pi}\Pi_i^a\Pi_j^b
 \left(\bar{\epsilon}_{abc}\hat{\ell}_d
  -\bar{\epsilon}_{dab}\hat{\ell}_c
  +\bar{\epsilon}_{cda}\hat{\ell}_b
  -\bar{\epsilon}_{bcd}\hat{\ell}_a \right)\nabla^c\xi^d
 \\
 &= \dfrac{1}{16\pi}\Pi_i^a\Pi_j^b \bar{\epsilon}_{abc}\bmq^c.
\end{align*}
where $\bmq^c = 2 \hat{\ell}_d \nabla^{[c}\xi^{d]}$. By construction $\bmq^a$ satisfies $\hat{\ell}_a \bmq^a = 0$, which implies that $\bmq^a$ resides on $\cN$. Therefore, one can rewrite the Noether charge 2-form on $\cN$ as 
\beq
 Q_{\xi_{ij}} = \dfrac{1}{16\pi}\epsilon_{ijk}\bmq^k.
\eeq
To find an explicit expression for $\bos{Q}_\xi$, we need to rewrite $\bmq$. One can rewrite this quantity as
\begin{align*}
 \bmq^c &= \hat{\ell}_d  \left( \nabla^c \xi^d -\nabla^d \xi^c \right)  \\
  &= 2\,\beta \ell^c - 2\,\xi^b \nabla_b \ell^c .
\end{align*}
In the above derivation, the linearised rescaling, the definition of the Lie derivative, and metric duality have been used. Using the expression in the Noether charge form yields
\beq
 Q_{\xi_{ij}} = \dfrac{1}{16\pi}\epsilon_{ijk}
 \left( 2\beta\ell^k-2\xi^m\nabla_m\ell^k \right).
\eeq
Recall that the Weingarten map was defined as $\c_i^{\ j} =  \Pi_i^a \nabla_a \ell^j$. Substituting it in the above expression, one gets
\beq
 Q_{\xi_{ij}} = -\dfrac{1}{8\pi}\epsilon_{ijk}
 \left[\xi^m\c_m^{\ k}-\beta(\xi^i)\ell^k \right].
\eeq
Integrating it over a cross-section $\cC$ of $\cN$ yields the Noether charge
\begin{align}\label{noether_charge}
    Q_\xi = \int_\cC \bos{Q}_\xi &= \dfrac{1}{8\pi} \int_\cC \epsilon_{ij} n_k \left[ \xi^m \c_m^{\ k} - \beta(\xi^i) \ell^k \right] \nonumber \\
    &= \dfrac{1}{8\pi} \int_\cC \epsilon_{ij} \left[ \xi^m n_k \c_m^{\ k} + \beta(\xi^i) \right]\,.
\end{align}
As we have the Noether charge, we can look into its variation under the perturbation \eqref{metric_perturbation}. Using the relations $h = g^{ab}h_{ab}=q^{AB} \delta q_{AB}$ and $\delta \epsilon_{ijk}=h\, \epsilon_{ijk}/2 $, one finds
\beq
 \delta Q_{\xi_{ij}} = -\dfrac{1}{16\pi}\epsilon_{ijk}
 \left(  h\,\xi^m\c_m^{\ k} -\beta h\,\ell^k+2\xi^m\delta\c_m^{\ k}
 \right).
\eeq
To compute the variation of the Weingarten map, recall the definition of it or the extrinsic curvature $K_a^{\ b}=\nabla_a \ell^b$. Variation of this yields
\beq
\delta K_a^{\ b} = \dfrac{1}{2} \left( \nabla_a h_c^{\ b} + \nabla_c h_a^{\ b} - \nabla^b h_{ac} \right) \ell^c . \nonumber
\eeq
Now, by taking the pullback of the above quantity and after some algebraic manipulations, one gets
\beq
 2\Pi_i^a\delta K_a^{\ b}
 =
 \Pi_i^a
 \left(
  \nabla_a h_c^{\ b}
  +\nabla_c h_a^{\ b}
  -\nabla^b h_{ac}
 \right)\ell^c.
\eeq
As the above quantity is orthogonal to $\hat{\ell}_a$, it is intrinsic to $\cN$. Using the definition \eqref{eq:Weingarten_map}, one finds the variation of the Weingarten map as follows
\beq\label{variation_weingarten_map}
 \delta \c_i^{\ j} = - \dfrac{1}{2}\Gamma_i\ell^j + \dfrac{1}{2}\cL_\ell h_i^{\ j} +  h_i^{\ k}\c_k^{\ j} -  h_k^{\ j} \c_i^{\ k} .
\eeq
Substituting it in the variation of the Noether charge yields
\beq\label{variation_noether-charge}
 \delta Q_{\xi_{ij}} = -\dfrac{1}{16\pi}\epsilon_{ijk}
 \Big( h\,\xi^m\c_m^{\ k} -\beta h\,\ell^k
 -\xi^m\Gamma_m\ell^k +\xi^m\cL_\ell h_m^{\ k}
 +2\xi^m h_m^{\ n}\c_n^{\ k} -2\xi^m h_n^{\ k}\c_m^{\ n} \Big).
\eeq
\\

\paragraph{Global charges}
As discussed in \cref{sec:cps_review}, the global charge is defined as \eqref{hamiltonianIW},
\begin{align*}
\delta \cQ_\xi = \int_{\pa \S} \left( \delta \bos{Q}_\xi - \xi\i\bth \right)\,,
\end{align*}
or, in terms of components as
\beq\label{variation_global_charge}
 \delta \cQ_{\xi_{ab}} = \delta Q_{\xi_{ab}} - \xi^c \theta_{cab}\,.
\eeq
Now, using the definition \eqref{symplectic_potential} and the Weingarten map, one can rewrite the symplectic potential as
\beq
 \theta_{ijk} = -\dfrac{1}{16\pi}\epsilon_{ijk}
\left( \cL_\ell h+h_m^{\ n}\c_n^{\ m} \right).
\eeq
Pullback of the global charge to $\cN$, and substituting the above expression in the variation of $\cQ$, yields
\beq
 \delta\cQ_{\xi_{ij}}
 = -\dfrac{1}{16\pi}\epsilon_{ijk}
 \Big( h\,\xi^m\c_m^{\ k} -\beta h\,\ell^k -\xi^m\Gamma_m\ell^k
 +\xi^m\cL_\ell h_m^{\ k} +2\xi^m h_m^{\ n}\c_n^{\ k}
 -2\xi^m h_n^{\ k}\c_m^{\ n} -\xi^k\cL_\ell h 
 -\xi^k h_m^{\ n}\c_n^{\ m} \Big).
\eeq
where the variation of the Noether charge \eqref{variation_noether-charge} has been used. Although we have a variation of the global charge, it does not guarantee that the charge itself is defined. We can define the charge when the above expression is a total variation, and in general, it is not true (see  \cref{eq:global_charge_nonintegrability}). For a detailed discussion on this, the reader may refer to \cite{Chandrasekaran:2018aop, Wald:1999wa}.\\

\paragraph{Local charges}
To define an integrable charge, one needs to apply the WZ
procedure \cite{Wald:1999wa} as discussed in
\cref{sec:cps_review}, which requires a boundary symplectic
potential satisfying
\eqref{stationary_boundary_symplectic_potential}.
One useful choice is \cite{Chandrasekaran:2018aop}
\beq
 \bTh=\ul{\bth}-\delta\bos{\vr},
\eeq
for a form $\bos{\vr}$ on $\cN$.
In components, this becomes
$\Theta_{ijk}=\theta_{ijk}-\delta\vr_{ijk}$.
Substituting this in the definition of the WZ charges yields
\begin{align}
 \delta\bos{\sq}_\xi
 &=
 \delta\left(\boq_\xi-\xi\i\bos{\vr}\right)
 \nonumber\\
 \implies\qquad
 \sq_{\xi_{ij}}
 &=
 Q_{\xi_{ij}}-\xi^k\vr_{ijk}.
\end{align}
With our volume-form convention, one can choose
\beq
 \vr_{ijk}
 =
 -\dfrac{1}{8\pi}\vartheta\epsilon_{ijk}.
\eeq
Its variation gives
\beq
 \delta\vr_{ijk}
 =
 -\dfrac{1}{16\pi}
 \left(
  \cL_\ell h+h\vartheta
 \right)\epsilon_{ijk},
\eeq
where the definition of the expansion
$\vartheta=\frac12 q^{AB}\cL_\ell q_{AB}$ has been used.
Thus, the boundary symplectic potential becomes
\beq\label{boundary_symplectic_potential}
 \Theta_{ijk}
 =
 -\dfrac{1}{16\pi}\epsilon_{ijk}
 \left(
  h_m^{\ n}\c_n^{\ m}-h\vartheta
 \right).
\eeq
Now, it is straightforward to find an expression for the
localised charge $2$-form,
\beq\label{WZ_charge_1}
 \sq_{\xi_{ij}}
 =
 -\dfrac{1}{8\pi}\epsilon_{ijk}
 \left(
  \xi^m\c_m^{\ k}
  -\vartheta\xi^k
  -\beta\ell^k
 \right).
\eeq
Using $\cL_\xi\ell^a=\beta\ell^a$, the above charge can be rewritten as
\beq
 \sq_{\xi_{ij}}
 =
 -\dfrac{1}{8\pi}\epsilon_{ijk}
 \left(
  \ell^m\nabla_m\xi^k-\vartheta\xi^k
 \right).
\eeq
Upon restriction to a cross-section, the relation between
the volume and area forms gives
\beq
 \left.\sq_{\xi_{ij}}\right|_\cC
 =
 \dfrac{1}{8\pi}\epsilon_{ij}n_k
 \left(
  \ell^m\nabla_m\xi^k-\vartheta\xi^k
 \right).
\eeq
Therefore, the localised charge form on a cross-section is
\beq\label{WZ_charge_2}
 \bos{\sq}_\xi
 =
 \dfrac{1}{8\pi}\boep_2
 \left(
  n_b\ell^a\nabla_a\xi^b
  -\vartheta\xi^a n_a
 \right),
\eeq
where $\boep_2$ is the area form of the cross-section.
For the explicit charge and boundary potential above, the
on-shell balance law associated with \cref{sympectic_flux} is
\beq\label{eq:coserv_law_form}
 \bos{\sf}_\xi
 =
 \ed\bos{\sq}_\xi
 =
 \bTh(g;\cL_\xi g).
\eeq
Using \eqref{boundary_symplectic_potential}, the flux is
\begin{align*}
 \sf_{\xi_{ijk}}
 &=
 \Theta_{ijk}(g;\cL_\xi g)
 \\
 &=
 -\dfrac{1}{16\pi}\epsilon_{ijk}
 \left.
 \left(
  h_m^{\ n}\c_n^{\ m}-h\vartheta
 \right)
 \right|_{h_{ab}=\cL_\xi g_{ab}}
 \\
 &=
 -\dfrac{1}{8\pi}\epsilon_{ijk}
 \nabla_a\xi_b
 \left[
  \nabla^{(a}\ell^{b)}-g^{ab}\vartheta
 \right].
\end{align*}

\vspace{.5cm}

\subparagraph{Stationary regions:} For stationary regions, the charges become simpler. By construction, the flux vanishes at a stationary region and the WZ charge becomes
\begin{align*}
    \sq_\xi [\cC] &= \dfrac{1}{8\pi} \int_\cC \epsilon_{ij}n_k \left( \xi^m \c_m^{\ k}  - \beta \ell^k \right) \\
    &= \dfrac{1}{8\pi} \int_\cC \epsilon_{ij}  \left( \xi^m \c_m^{\ k} n_k  + \beta  \right) .
\end{align*}
Using the definition of the rotation $1$-form as in \eqref{rotation_1-form} in the above yields the localised charges for a stationary cross-section
\beq
 \sq_\xi [\cC] = \dfrac{1}{8\pi} \int_\cC \epsilon_{ij} \left( \beta - w_k \xi^k \right).
\eeq

\subsection{Gaussian null coordinates} \label{sec:GNC}
While studying or dealing with null surfaces, it is convenient to use coordinates adapted to the null surface $\cN$. One such set of coordinates is called the Gaussian null coordinates (GNC). Here we will discuss the null surface structure, its symmetries and charges in the GNC. We will follow the convention of \cite{Chandrasekaran:2019ewn}. As done earlier, we will consider spacetime to be of $4$ dimensions, and the null hypersurface to have the topology $\cZ \times \mathbb{R}$. The base space $\cZ$ (here $\mathbb{S}^2$) is the space of null generators $\ell^a$. Then, the null surface is generated by the geodesics of $\ell^a$. We will follow the index convention as described in \cref{sec:null_surface_review}.

As we are interested in the coordinate system, which is well adapted to $\cN$, we will choose the null vectors $\ell^a$ to be affinely parametrised, i.e.,
\beq
 \ell^a \nabla_a \ell^b = 0 \,,
\eeq
with $v$ being the affine parameter. So, it satisfies
\beq
 \ell^a\nabla_a v \nulleq 1\,. \nonumber
\eeq 
Let $\cC (\cong\mathbb{S}^2)$ be a cross-section of $\cN$ such that $v|_\cC =0$. If the coordinates of $\cC$ are denoted by $x^A$, one can extend $x^A$ to all $\cN$ by parallel transport along the null generators 
\beq
 \ell^a\nabla_a x^A \nulleq 0\,. \nonumber
\eeq
Therefore, the coordinates $x^A$ along with the affine parameter $v$ define a coordinate system $(v, x^A)$ on $\cN$. To define a coordinate in a neighbourhood of $\cN$, consider a function $u$ such that the null hypersurface corresponds to $u=0$. Then, one can define a normal $\ell_a$ s.t. 
\beq
 \ell_a \equiv - du\,,
\eeq
and a transverse vector field $n^a$ s.t.
\beq
 n^a \equiv \pa_u
\eeq
w.r.t $\cN$. To fix the coordinates away from the null surface, one can choose $u$ s.t. $n^a$ is a null vector field and is affinely parametrised,
\beq
\begin{split}
    n^a n_a =0 \,;\\
    n^a\nabla_a n^b = 0\,.
\end{split}
\eeq
Now, one can extend the coordinates $(v,x^A)$ away from the null surface by parallel transport along $n^a$,
\beq
 n^a \nabla_a x^A =0\,;\qquad  n^a \nabla_a v =0 \,.~ (trivially) \nonumber
\eeq
Then $(u,v,x^A)$ defines a Gaussian null coordinate system in a neighbourhood of the null hypersurface $\cN$. To obtain an explicit expression of the metric in the GNC, one needs to use the above conditions so that coordinates are well adapted to $\cN$. From the definition, it follows that in the GNC, the metric satisfies \cite{Madler:2018bmu,Chandrasekaran:2019ewn}
\begin{subequations}
\beq
      g_{uu}=g_{uA}=0\,; \qquad  g_{vA} \nulleq g_{vv} \nulleq \pa_u g_{vv} \nulleq 0\,.
\eeq
and\footnote{$\ell^a n_a = -1$ implies $g_{vu} \ell^v = -1$.
One can choose the normalisation $\ell^v = 1$, which gives
$g_{vu}=g_{uv}=-1\,.$
One could have chosen a different normalisation like $\ell^v = - \frac{1}{\eta}$, then the metric component would have been $g_{uv}=\eta$. This kind of normalisation was considered in \cite{Adami:2021nnf}.}
\beq
g_{uv}=-1\,.
\eeq
\end{subequations}
Combining all the above conditions, one gets the metric
\beq
 ds^2 = -2\,du\, dv + g_{vv}\, dv^2 + 2g_{vA}\, dv\, dx^A + g_{AB}\, dx^A dx^B . \nonumber
\eeq
Furthermore, one can write the two-dimensional $g_{AB}$ on $\cZ$ as $q_{AB}$ (as we have seen in \cref{sec:null_surface_review}),
\begin{align*}
g_{AB} = q_{AB}\,; \qquad g_{vA} = -q_{AB}W^B .
\end{align*}
By using the ansatz $g_{vv}=-W + q_{AB} W^A W^B$, we obtain the metric adapted to $\cN$ in GNC as
\beq\label{metric_GNC}
 ds^2 = -W dv^2 - 2du dv + q_{AB} \left( dx^A - W^A dv \right) \left( dx^B - W^B dv \right).
\eeq
Here, $W$, $W^A$ and $q_{AB}$ are functions of $(u,v,x^A)$ and tensors on $\cZ$. As mentioned earlier, $q_{AB}$ is the (Riemannian) metric on $\cZ$. The conditions on the different components of the metric provide constraints on $W$ and $W^A$ as follows
\beq
\begin{split}
   W^A &\nulleq 0; \\
   W  \nulleq &\pa_u W\nulleq 0\,.
\end{split}
\eeq
We choose the spacetime orientation such that
\beq
 \boep_4
 =
 -\sqrt{\det q_{AB}}\,
 \ed u\wedge\ed v\wedge\ed x^1\wedge\ed x^2.
\eeq
On $\cN$, where $\ell^a=\partial_v$ and
$n^a=\partial_u$, our volume and area forms are therefore
\beq
 \boep_3=-n\i\boep_4=\ed v\wedge\boep_2,
 \qquad
 \boep_2=\ell\i\boep_3
 =\sqrt{\det q_{AB}}\,\ed x^1\wedge\ed x^2.
\eeq
Thus, a segment between $v=v_1$ and $v=v_2$, with
$v_2>v_1$, has oriented boundary
$\cC(v_2)\cup[-\cC(v_1)]$.
\\

\paragraph{Null surface generator}
Once we have the metric adapted to the null hypersurface, the next step is to find the null vector field $\ell^a$ which generates $\cN$. One can express the vector field in the GNC as
\beq
 \ell =\ell^a \pa_a  = \ell^u \pa_u +  \ell^v \pa_v + \ell^A \pa_A \,. \nonumber
\eeq
We have already chosen $\ell^v =1$. Also, recall that $\ell_a = -du$, which implies $\ell_A = 0$ and $\ell^A = W^A$. To find the component $\ell^u$, we can use the null nature of the vector field,
\begin{align*}
    \ell^a\ell_a = g_{ab}\ell^a \ell^b &= 2g_{uv} \ell^u \ell^v + g_{vv} \ell^v \ell^v +2 g_{vA} \ell^v \ell^A + g_{AB} \ell^A \ell^B \\
    &=-2 \ell^u + (-W+q_{AB}W^A W^B) + 2(-q_{AB}W^B) W^A + q_{AB} W^A W^B\\
    &= -2 \ell^u - W \,.
\end{align*}
Using $\ell^a\ell_a=0$, one finds
\beq
 \ell^u = - \dfrac{W}{2}\,. \nonumber
\eeq
Therefore, the generator of the null hypersurface is given by,
\beq
 \ell = \pa_v - \dfrac{1}{2} W \pa_u + W^A \pa_A \,.
\eeq
\\

\paragraph{Symmetry generator}
Let $\xi^a$ be a vector field that generates infinitesimal coordinate transformations corresponding to the symmetry algebra of $\cN$ that preserves the GNC metric \eqref{metric_GNC}. So, one can write the vector field $\xi^a$ as
\beq
 \xi^a = \xi^a (u,v,x^A). \nonumber
\eeq
As this vector field corresponds to infinitesimal transformations, it can be expanded in GNC around the null surface $u=0$ as
\beq
 \xi = (f_0 + u f_1 )\, \pa_v + (\beta_0 + u\beta_1)\, \pa_u + (X^A_0 + u X_1^A )\, \pa_A +\mathcal{O}(u^2)
\eeq
 up to first order in $u$.To preserve the location $u=0$ of $\cN$, $\xi^a$ should be tangent to $\cN$, which implies
\[ \xi^u = 0 \quad \Rightarrow \quad \beta_0 = 0 . \]
As a symmetry generator, $\xi$ should preserve the metric \eqref{metric_GNC}. For that we need to look at $\cL_\xi g_{ab}$ satisfying the following conditions
\begin{enumerate}[label=(\roman*)]
\item $\cL_\xi g_{uu}=0$ implies $f_1 = 0$.

\item $\cL_\xi g_{uv}=0$ fixes $\beta_1 = - \pa_v f_0$.

\item $\cL_\xi g_{uA}=0$ gives $q_{AB} X^B_1 =  \pa_A f_0$.

\item $\cL_\xi g_{vA}=\mathcal{O}(u)$ restricts $X^B_0= X^B_0(x^A) = X^A$.
This implies that $X^A$ is a vector on $\cZ$.
\item $\cL_\xi g_{vv}=\mathcal{O}(u^2)$ provides the relation $\pa_v \beta_1 =0= \pa_v^2 f_0$.
Therefore, $\beta_1$ is independent of $v$. Let $\beta_1= \beta (=-\pa_v f_0)$. So we can find the function $f_0$ as follows
\begin{align*}
f_0 = \vs - v \beta \,.
\end{align*}
Similarly, one gets
\beq
 X^A_1 = q^{AB}\pa_B (\vs - v\beta)\,. \nonumber
\eeq
\end{enumerate}
Combining all these, the symmetry generator is given by,
 \beq\label{eq:symm_generator_GNC}
\xi = ( \vs - v \beta)\, \pa_v + u\beta\, \pa_u + \left[ X^A + u q^{AB} \pa_B (\varsigma - v \beta) \right]\pa_A\,.
 \eeq
Here, $\vs,\beta$ are functions on $\cZ$: $\vs= \vs(x^A), \beta=\beta(x^A)$; and $X^A$ is a vector field on $\cZ$. On the null hypersurface $\cN$ ($u=0$), the above expression reduces to
\beq
 \xi =( \vs - v \beta)\, \pa_v + X^A \pa_A\,.
\eeq
Therefore, on $\cN$, the symmetry generator $\xi^a$ is characterised by $(\vs, \beta, X^A)$. Two such given generators $\xi_1^a=(\vs_1, \beta_1, X_1^A)$ and $\xi^a_2 = (\vs_2, \beta_2, X_2^A)$, satisfy the  following algebra
\beq
 \left[ (\vs_1, \beta_1, X_1^A), (\vs_2, \beta_2, X_2^A) \right] = (\vs, \beta, X^A),
\eeq
where $[^.,^.]$ denotes the Lie bracket, and the quantities on the right hand side are given by
\begin{align*}
   \vs &= - \vs_1 \beta_2 + \vs_2 \beta_1 + X^A_1\pa_A \vs_2 -X^A_2\pa_A \vs_1 \,;  \\
   \beta &= X^A_1\pa_A \beta_2 - X^A_2\pa_A \beta_1 \,;\\
   X^A &= [X_1 , X_2]^A = X^B_1 \pa_B X_2^A - X^B_2 \pa_B X_1^A.
\end{align*}
Now, we can look at the different symmetries of the symmetry algebra $\fg$.
\begin{itemize}
    \item $X^A = 0$ implies the symmetries of the form $(\vs,\beta,0): \xi=(\vs - v\beta)\pa_v$, which corresponds to translation along the $v$ direction but with angle $(x^A)$ dependency. These kinds of translations are called \emph{supertranslations}.

  \item $\vs=\beta=0$ corresponds to symmetries of the form $(0,0,X^A): \xi= X^A \pa_A$. These transformations are called \emph{superrotations} and form the Lie algebra $\mathfrak{diff}(\cZ)$.

    \item $\beta=X^A=0$ correspond to $(\vs,0,0): \xi = \vs \pa_v$ which are angle dependent translations along the affine parameter $v$, and called \emph{affine supertranslations} \cite{Chandrasekaran:2018aop}. These also form an abelian Lie ideal $\ft_0 \subset \fg$.

    \item The quotient $\fw \cong \ft/\ft_0$ forms the symmetry $(0,\beta,0): \xi = -v \beta\pa_v$. These transformations are referred to as \emph{Killing supertranslations} \cite{Chandrasekaran:2018aop} or \emph{Weyl rescalings} \cite{Freidel:2021fxf}.
    \end{itemize}
Therefore, the symmetry algebra $\fg$ of any null surface $\cN$ has the following structure
\beq
 \fg \cong \, \mathfrak{diff}(\cZ) \ltimes \big( \fw \ltimes \ft_0 \big).
\eeq
Similarly, the corresponding symmetry group can be constructed, consisting of supertranslations $\cT$ and diffeomorphisms $\mathsf{Diff}(\cZ)$. The supertranslations $\cT$ has two subgroups: the normal subgroup corresponding to affine supertranslations $\cT_0$ and the non-normal subgroup corresponding to killing supertranslations (Weyl rescaling) $\cW$. Therefore, the symmetry group structure of any null surface is given by
\beq
  \mathsf{G} \cong \, \mathsf{Diff}(\cZ) \ltimes \big( \cW \ltimes \cT_0 \big).
\eeq
\section{Dynamical entropy of causal diamond}\label{sec:dyn_entropy_cd}
In this section, first we will briefly discuss the structure of a causal diamond and the conservation laws associated with it, following \cite{Chandrasekaran:2019ewn}, and then extend this idea to obtain the notion of dynamical entropy for causal diamonds.
 We denote the spacetime manifold $\cM$ armed with a metric $g_{ab}$ as we have done earlier. One can start by considering all the causal paths in this spacetime from a point $p^-$ to a point $p^+$ in its causal future, i.e., inside or on the future light cone of $p^-$. The set of all such causal paths defines the \emph{causal diamond} (or double cone) \cite{Dappiaggi:2010iq, Witten:2019qhl} denoted by $\mathscr{D}_{p^-}^{p^+}$ and is given by
 \beq
 \mathscr{D}_{p^-}^{p^+} = I^+ (p^-) \cap I^- (p^+) \subset \cM,
 \eeq
where $I^\pm$ denotes the chronological future (past) of the point $p^-$( $p^+$). Following \cite{Chandrasekaran:2019ewn}, we consider the causal diamond which has smooth null boundaries $\cN^\pm$ emanating from $p^\pm$ that intersect at the smooth codimension-2 \emph{bifurcation surface (edge)} $\cB$. A schematic diagram of a causal diamond is shown in \cref{fig:causal_diamond1}.
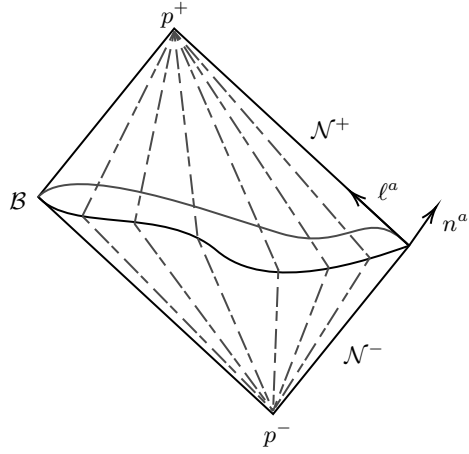
\begin{figure}[htbp]
\centering

\tikzset{every picture/.style={line width=0.75pt}} 
\begin{tikzpicture}[x=0.7pt,y=0.7pt,yscale=-.85,xscale=.85]

\draw   (210.11,139.87) -- (296.12,33.55) -- (444.7,170.69) -- (358.69,277.01) -- cycle ;
\draw    (210.11,139.88) .. controls (231,162.81) and (290,145.62) .. (323,173.75) .. controls (356,201.87) and (418,180) .. (444.7,170.68) ;
\draw [color={rgb, 255:red, 74; green, 74; blue, 74 }  ,draw opacity=1 ]   (444.7,170.68) .. controls (418,143.28) and (407,176.09) .. (363,162.81) .. controls (319,149.53) and (232,116.72) .. (210.11,139.88) ;
\draw [color={rgb, 255:red, 74; green, 74; blue, 74 }  ,draw opacity=1 ] [dash pattern={on 3.75pt off 3pt on 7.5pt off 1.5pt}]  (296.1,35.47) -- (420,178.53) ;
\draw [color={rgb, 255:red, 74; green, 74; blue, 74 }  ,draw opacity=1 ] [dash pattern={on 3.75pt off 3pt on 7.5pt off 1.5pt}]  (296.1,35.47) -- (394,184.7) ;
\draw [color={rgb, 255:red, 74; green, 74; blue, 74 }  ,draw opacity=1 ] [dash pattern={on 3.75pt off 3pt on 7.5pt off 1.5pt}]  (296.1,35.47) -- (362,186) ;
\draw [color={rgb, 255:red, 74; green, 74; blue, 74 }  ,draw opacity=1 ] [dash pattern={on 3.75pt off 3pt on 7.5pt off 1.5pt}]  (296.1,35.47) -- (311,166.19) ;
\draw [color={rgb, 255:red, 74; green, 74; blue, 74 }  ,draw opacity=1 ] [dash pattern={on 3.75pt off 3pt on 7.5pt off 1.5pt}]  (296.1,35.47) -- (271,156.93) ;
\draw [color={rgb, 255:red, 74; green, 74; blue, 74 }  ,draw opacity=1 ] [dash pattern={on 3.75pt off 3pt on 7.5pt off 1.5pt}]  (296.1,35.47) -- (238,152) ;
\draw [color={rgb, 255:red, 74; green, 74; blue, 74}  ,draw opacity=1 ] [dash pattern={on 3.75pt off 3pt on 7.5pt off 1.5pt}]  (420,178.1) -- (358.7,274.99) ;
\draw [color={rgb,255:red, 74; green, 74; blue, 74 }  ,draw opacity=1 ] [dash pattern={on 3.75pt off 3pt on 7.5pt off 1.5pt}]  (394,184.18) -- (358.7,274.99) ;
\draw [color={rgb, 255:red, 74; green, 74; blue, 74}  ,draw opacity=1 ] [dash pattern={on 3.75pt off 3pt on 7.5pt off 1.5pt}]  (362,185.45) -- (358.69,275) ;
\draw [color={rgb, 255:red, 74; green, 74; blue, 74 }  ,draw opacity=1 ] [dash pattern={on 3.75pt off 3pt on 7.5pt off 1.5pt}]  (311,165.96) -- (358.7,274.99) ;
\draw [color={rgb, 255:red, 74; green, 74; blue, 74 }  ,draw opacity=1 ] [dash pattern={on 3.75pt off 3pt on 7.5pt off 1.5pt}]  (271,156.85) -- (358.69,275) ;
\draw [color={rgb, 255:red, 74; green, 74; blue, 74}  ,draw opacity=1 ] [dash pattern={on 3.75pt off 3pt on 7.5pt off 1.5pt}]  (238,152) -- (358.69,275) ;
\draw    (444.7,170.69) -- (408.47,137.35) ;
\draw [shift={(407,136)}, rotate = 42.62] [color={rgb, 255:red, 0; green, 0; blue, 0 }  ][line width=0.75]    (10.93,-3.29) .. controls (6.95,-1.4) and (3.31,-0.3) .. (0,0) .. controls (3.31,0.3) and (6.95,1.4) .. (10.93,3.29)   ;
\draw    (444.7,170.69) -- (460.85,147.64) ;
\draw [shift={(462,146)}, rotate = 125.03] [color={rgb, 255:red, 0; green, 0; blue, 0 }  ][line width=0.75]    (10.93,-3.29) .. controls (6.95,-1.4) and (3.31,-0.3) .. (0,0) .. controls (3.31,0.3) and (6.95,1.4) .. (10.93,3.29)   ;

\draw (286,15) node [anchor=north west][inner sep=0.75pt]   [align=left] {$p^+$};
\draw (351,280) node [anchor=north west][inner sep=0.75pt]   [align=left] {$p^-$};
\draw (190,135) node [anchor=north west][inner sep=0.75pt]   [align=left] {$\cB$};
\draw (381,88) node [anchor=north west][inner sep=0.75pt]   [align=left] {$\cN^+$};
\draw (402,230) node [anchor=north west][inner sep=0.75pt]   [align=left] {$\cN^-$};
\draw (423,129) node [anchor=north west][inner sep=0.75pt]   [align=left] {$\ell^a$};
\draw (464,149) node [anchor=north west][inner sep=0.75pt]   [align=left] {$n^a$};

\end{tikzpicture}
\caption{\textit{\small Causal diamond $\mathscr{D}_{p^-}^{p^+}$ with null boundary $\cN=\cN^+ \cup \cN^-$ and bifurcation surface $\cB$. The affine null generators of $\cN^\pm$ are respectively $\ell^a$ and $n^a$.}}\label{fig:causal_diamond1}
\end{figure}
In GNC, the coordinates of different parts of the causal diamond are given by
\beq
\begin{split}
p^+ \equiv (u=0, v =\tilde V(x^A))\,; \qquad  p^- \equiv (u=\tilde U(x^A), v = 0)\,;\\
\cN^+ \equiv (u=0, v>0)\,; \quad \cN^- \equiv (u<0,v=0)\,; \quad \cB \equiv (u=0,v=0)\,.
\end{split}
\eeq
Unlike a single null hypersurface, the boundary of a causal diamond has a distinguished cut, the edge $\cB$ shared by both null boundaries $\cN^\pm$.
There exist unique future directed null vector fields $\ell^a$ and $n^a$ tangent to $\cN^+$ and $\cN^-$ respectively satisfying $\ell^a n_a |_\cB =-1$. One can extend these vector fields to $\cN^\pm$ such that they become the affine null generators on $\cN^\pm$, and satisfy $\ell^a n_a |_\cN =-1$. Since there are two null surfaces for a causal diamond, one obtains two sets of GNCs: $(u_+, v_+,x^A_+)$ associated to $\cN^+$ and $(u_-,v_-,x^A_-)$ associated to $\cN^-$. Both of these coincide at $\cB$, and we assume the spacetime metric $g_{ab}$ is smooth throughout the causal diamond. To fix the position of the bifurcation surface as above, one can choose
\beq
 \ell^a \big|_{\cN^+} \equiv \pa_v = \pa_{v_+}\big|_\cB = \pa_{v_-}\big|_\cB \,; \qquad 
 n^a|_{\cN^-} \equiv \pa_u = \pa_{u_+}\big|_\cB = \pa_{u_-}\big|_\cB\,. 
\eeq
We orient both null boundaries in the direction of increasing affine
parameter, with $\ell^a=\partial_v$ on $\cN^+$ and $n^a=\partial_u$
on $\cN^-$. Keeping the spacetime orientation fixed, we define the
volume $3$-forms by
\begin{subequations}
\beq
 \label{eq:volume_form}
 \epsilon^+_{abc}=-n^d\epsilon_{dabc},
 \qquad
 \epsilon^-_{abc}=+\ell^d\epsilon_{dabc}.
\eeq
The area $2$-forms on their cross-sections are
\beq
 \label{eq:area_form}
 \begin{aligned}
  \epsilon^+_{ab}
  &=\ell^c\epsilon^+_{cab}
   =-\ell^c n^d\epsilon_{dcab},\\
  \epsilon^-_{ab}
  &=n^c\epsilon^-_{cab}
   =n^c\ell^d\epsilon_{dcab}.
 \end{aligned}
\eeq
These area forms agree at the bifurcation surface:
\beq
 \epsilon^+_{ab}\big|_\cB
 =
 \epsilon^-_{ab}\big|_\cB.
\eeq
\end{subequations}
Thus, in coordinates transported along the null generators,
\beq
 \boep^+_3=\ed v\wedge\boep^+_2,
 \qquad
 \boep^-_3=\ed u\wedge\boep^-_2.
\eeq
The null volume forms have the opposite signs to those defined in
\cite{Chandrasekaran:2019ewn}, while the cut area forms are unchanged. For a portion of either boundary bounded by an earlier cross-section $\cC_1$ and a later cross-section $\cC_2$, our integration orientation gives
\beq
 \partial\cN^\pm_{12}=\cC_2\cup(-\cC_1).
\eeq
In particular,
\beq
 \partial\cN^+=p^+\cup(-\cB),
 \qquad
 \partial\cN^-=\cB\cup(-p^-),
\eeq
where the contributions at the corners $p^\pm$ are understood as limits. The conventions of \cref{sec:null_surface_review} apply
directly to $\cN^+$ through the identification
\beq
 (\ell^a,n^a,\boep_3,\boep_2)\to (\ell^a,n^a,\boep^+_3,\boep^+_2).
\eeq
On $\cN^-$, the generator and auxiliary normal are exchanged. To retain the common area orientation at $\cB$, the corresponding replacement in the formulas of \cref{sec:null_surface_review} is
\beq
(\ell^a,n^a,\boep_3,\boep_2)\to (n^a,\ell^a,-\boep^-_3,-\boep^-_2).
\eeq

Using the Wald-Zoupas procedure, one can define localised
charges on cross-sections of the boundary of a causal diamond.
For symmetry generators matched smoothly across $\cB$, the
charges evaluated from $\cN^+$ and $\cN^-$ agree at $\cB$.
Under the regularity assumptions of
\cite{Chandrasekaran:2019ewn}, the charges vanish in the limits
to the corners $p^\pm$. With our orientation conventions, the fluxes through the two null boundaries are
\beq
\begin{aligned}
 \sf_\xi[\cN^+]
 &=
 \sq_\xi[p^+]-\sq_\xi[\cB],
 \\
 \sf_\xi[\cN^-]
 &=
 \sq_\xi[\cB]-\sq_\xi[p^-].
\end{aligned}
\eeq
Here, the charges at $p^\pm$ denote limits of charges on cross-sections approaching the corners. Since these limits vanish, one obtains
\beq
 \sf_\xi[\cN^+]=-\sf_\xi[\cN^-].
\eeq
The relative minus sign follows from the orientations of the two null segments and the matching of their charges at $\cB$. This is the conservation law of \cite{Chandrasekaran:2019ewn}, expressed in our convention. There is one such relation for each matched boundary symmetry.

Now, we will implement the idea of localised charges (see \cref{sec:charges_N}) for an arbitrary slice of a causal diamond. The charge form is defined as in  \cref{WZ_charge_2},
\beq
\bos{\sq}_\xi = \dfrac{1}{8\pi} \boep_2  \left( n_b\ell^a \nabla_a \xi^b - \vartheta \xi^a n_a \right). \nonumber
\eeq
Integrating it on a cross-section of a null hypersurface yields the WZ localised charge. As discussed earlier, a causal diamond has two null boundaries $\cN^\pm$. We will focus on an arbitrary cross-section $\cC^+$ of the future null boundary $\cN^+$ as shown in the \cref{fig:cd_arb_slice}. We can rewrite the WZ charge in $\cN^+$ as
\beq
 \sq_\xi [\cC^+]= \dfrac{1}{8\pi} \int_{\cC^+} \boep^+_2  \left( n_b\ell^a \nabla_a \xi^b - \vartheta \xi^a n_a \right),
\eeq
where $\boep^+_2$ is the area $2$-form. For a cross-section $\cC^-$ of $\cN^-$, the corresponding expression is given by
\beq
 \sq_\xi[\cC^-]
 =
 -\dfrac{1}{8\pi}
 \int_{\cC^-}\boep^-_2
 \left(
  \ell_b n^a\nabla_a\xi^b
  -\vartheta_-\xi^a\ell_a
 \right),
\eeq
where $\vartheta_-$ is the expansion of the affine generator
$n^a$ of $\cN^-$. The relative sign follows from the
replacement rule relating the conventions of
\cref{sec:null_surface_review} to those on $\cN^-$.

 Now, consider the portion of the null boundary $\cN^+$ between the cross-section $\cC^+$ and
the bifurcation surface $\cB$. Denoting that region $\Delta\cN^+$, we can write the flux corresponding to the WZ charges as
\beq\label{eeq:conserv_arb_slice}
 \sf_\xi [\Delta\cN^+] = \sq_\xi [\cC^+] - \sq_\xi [\cB].
\eeq
\begin{figure}[htbp]
    \centering

\tikzset{every picture/.style={line width=0.75pt}} 

\begin{tikzpicture}[x=0.75pt,y=0.75pt,yscale=-1,xscale=1]

\draw   (204.56,99.5) -- (268.51,26.76) -- (379,120.57) -- (315.05,193.31) -- cycle ;
\draw    (204.56,99.5) .. controls (220.1,115.19) and (263.97,103.43) .. (288.51,122.67) .. controls (313.04,141.91) and (359.14,126.95) .. (379,120.58) ;
\draw [color={rgb, 255:red, 74; green, 74; blue, 74 }  ,draw opacity=1 ]   (379,120.58) .. controls (359.14,101.83) and (350.97,124.28) .. (318.25,115.19) .. controls (285.53,106.1) and (220.84,83.65) .. (204.56,99.5) ;
\draw [color={rgb, 255:red, 74; green, 74; blue, 74 }  ,draw opacity=1 ] [dash pattern={on 3.75pt off 3pt on 7.5pt off 1.5pt}]  (326.43,88.64) -- (360.63,125.95) ;
\draw [color={rgb, 255:red, 74; green, 74; blue, 74 }  ,draw opacity=1 ] [dash pattern={on 3.75pt off 3pt on 7.5pt off 1.5pt}]  (310.81,88.64) -- (341.3,130.17) ;
\draw [color={rgb, 255:red, 74; green, 74; blue, 74 }  ,draw opacity=1 ] [dash pattern={on 3.75pt off 3pt on 7.5pt off 1.5pt}]  (296.69,86.58) -- (319.74,132.42) ;
\draw [color={rgb, 255:red, 74; green, 74; blue, 74 }  ,draw opacity=1 ] [dash pattern={on 3.75pt off 3pt on 7.5pt off 1.5pt}]  (272.15,81.8) -- (266.2,111.9) ;
\draw [color={rgb, 255:red, 74; green, 74; blue, 74}  ,draw opacity=1 ] [dash pattern={on 3.75pt off 3pt on 7.5pt off 1.5pt}]  (257.28,79.06) -- (245.38,109.85) ;
\draw [color={rgb, 255:red, 74; green, 74; blue, 74 }  ,draw opacity=1 ] [dash pattern={on 3.75pt off 3pt on 7.5pt off 1.5pt}]  (242.41,77.69) -- (225.31,107.79) ;
\draw [color={rgb, 255:red, 74; green, 74; blue, 74 }  ,draw opacity=1 ] [dash pattern={on 3.75pt off 3pt on 7.5pt off 1.5pt}]  (234.97,74.39) -- (212.66,105.06) ;
\draw [color={rgb, 255:red, 74; green, 74; blue, 74 }  ,draw opacity=1 ] [dash pattern={on 3.75pt off 3pt on 7.5pt off 1.5pt}]  (284.05,83.16) -- (288.51,122.67) ;
\draw [color={rgb, 255:red, 74; green, 74; blue, 74 }  ,draw opacity=1 ] [dash pattern={on 3.75pt off 3pt on 7.5pt off 1.5pt}]  (333.12,86.48) -- (371.04,122.16) ;
\draw    (230.51,69.48) .. controls (236.46,82.48) and (253.56,74.27) .. (284.05,83.16) .. controls (314.53,92.06) and (332.38,88.64) .. (334.61,83.16) ;
\draw    (230.51,69.48) .. controls (248.36,62.64) and (288.51,71.53) .. (299.66,79.74) .. controls (310.81,87.95) and (324.94,77.69) .. (334.61,83.16) ;

\draw (259.1,11.58) node [anchor=north west][inner sep=0.75pt]  [font=\small] [align=left] {$p^+$};
\draw (307.69,192.88) node [anchor=north west][inner sep=0.75pt]  [font=\small] [align=left] {$p^-$};
\draw (188.16,93.78) node [anchor=north west][inner sep=0.75pt]  [font=\small] [align=left] {$\cB$};
\draw (360.66,88.2) node [anchor=north west][inner sep=0.75pt]  [font=\small] [align=left] {$\Delta\cN^+$};
\draw (345.36,158.67) node [anchor=north west][inner sep=0.75pt]  [font=\small] [align=left] {$\cN^-$};
\draw (210.47,56.73) node [anchor=north west][inner sep=0.75pt]  [font=\small] [align=left] {$\cC^+$};

\end{tikzpicture}
    \caption{\textit{\small An arbitrary slice of the causal diamond $\mathscr{D}_{p^-}^{p^+}$ that intersects the null boundary $\cN^+$ at $\cC^+$. The portion of $\cN^+$ bounded by $\cC^+$ and $\cB$ is $\Delta\cN^+$ shown with grey dashed lines.}}
    \label{fig:cd_arb_slice}
\end{figure}
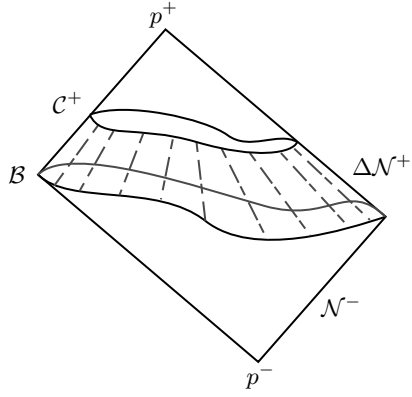
To find the explicit expressions of the charges and the
fluxes, we will use the GNC introduced in \cref{sec:GNC}.
In GNC, a symmetry generator of the future null boundary
of a causal diamond is given by
\beq
 \xi\nulleq -v\beta\,\pa_v+X^A\pa_A.
\eeq
For a causal diamond, the affine-supertranslation term
$\varsigma$ in \cref{eq:symm_generator_GNC} is absent because
the symmetry must preserve the bifurcation surface
\cite{Chandrasekaran:2019ewn}. For the thermodynamic interpretation, we restrict to the constant boost, with $X^A=0$ and $\beta$ constant, and choose
\beq
 \xi^a=-v\beta\ell^a
       =\kappa_\beta v\ell^a,
 \qquad
 \kappa_\beta=-\beta>0.
\eeq
The condition $\beta<0$ ensures that $\xi^a$ is
future directed on $\cN^+$, where $v>0$.
Since $\ell^a$ is affinely parametrised,
\beq
 \xi^a\nabla_a\xi^b
 =
 (-v\beta)\pa_v(-v\beta)\ell^b
 =
 -\beta\,\xi^b
 =
 \kappa_\beta\xi^b.
\eeq
The smoothly matched boost on $\cN^-$ is $\xi^a=\beta u n^a=-\kappa_\beta u n^a$. At the bifurcation surface, the charge is
\beq
 \sq_\beta[\cB]
 =
 \dfrac{\beta}{8\pi}\int_\cB\boep^+_2
 =
 -\dfrac{\kappa_\beta}{2\pi}\,
  \dfrac{A(\cB)}{4}.
\eeq
If one identifies $\kappa_\beta$ as the surface gravity, the charge can be interpreted as energy with the area term being the entropy of the causal diamond \cite{Chandrasekaran:2019ewn},
\beq
 \sq_\beta[\cB] = T_{(\beta)}\cS[\cB],
 \qquad
 T_{(\beta)}=-\dfrac{\kappa_\beta}{2\pi},
 \qquad
 \cS[\cB]=\dfrac{A(\cB)}{4}.
\eeq
The sign of $\beta$ is fixed by requiring $\xi^a$ to be future directed on $\cN^+$, as noted above, so that $\kappa_\beta>0$ and $T_{(\beta)}<0$. The negative sign is a property of the charge normalisation rather than of the diamond: the Rindler wedge of appendix \ref{app:wz-rindler} has the same relation $\sq_\xi=-\kappa_\beta\cS/(2\pi)$, yet its perturbed entropy increases under the null energy condition. What distinguishes a finite diamond is the nonzero background expansion of $\cN^+$, which enters the flux and decides the direction in which $\cS$ evolves. Thus, the temperature entering this relation is negative.
We extend the above relation to a cross-section $\cC^+$:
\begin{align}\label{eq:WZ_charge_dyn_entropy}
 \sq_\xi[\cC^+]
 &=
 \dfrac{1}{8\pi}
 \int_{\cC^+}\boep^+_2
 \left(
  n_b\ell^a\nabla_a\xi^b
  -\vartheta\xi^a n_a
 \right)
 \nonumber\\
 &=
 \dfrac{\beta}{8\pi}
 \int_{\cC^+}\boep^+_2(1-v\vartheta)
 \nonumber\\
 &=
 -\dfrac{\kappa_\beta}{2\pi}
 \left[
  \dfrac{A(\cC^+)}{4}
  -\dfrac14\int_{\cC^+}\boep^+_2\,v\vartheta
 \right].
\end{align}
The quantity in brackets has the same functional form as the dynamical black hole entropy in general relativity introduced in \cite{Hollands:2024vbe}. Borrowing this terminology, we identify the \emph{dynamical entropy of a causal diamond} by
\beq\label{eq:dyn_entropy_cd}
 \cS[\cC^+]
 =\dfrac14\int_{\cC^+}\boep^+_2(1-v\vartheta)
 = \dfrac{A(\cC^+)}{4} -\dfrac14\int_{\cC^+}\boep^+_2\,v\vartheta.
\eeq
For a related Noether charge construction on evolving null surfaces, see \cite{Ciambelli:2023mir, Rignon-Bret:2023fjq, Shajiee:2026coz}. The cut $\cC^+$ need not be a surface of constant $v$: since $\boep^+_2$ has no $\ed v$ component, its pullback to a cut $v=v_\cC(x^A)$ is $\sqrt{q}\big|_{v=v_\cC(x)}\,\ed^2x$. When $\cC^+$ coincides with $\cB$, the entropy reduces to $A(\cB)/4$, consistently with
\cite{Chandrasekaran:2019ewn,Jacobson:2018ahi}.
Now, we compute the flux through $\Delta\cN^+$. For a boundary symmetry $\xi^a=f\ell^a$, our volume-form convention gives
\beq
 \sf_\xi[\Delta\cN^+] = -\dfrac{1}{8\pi}
 \int_{\Delta\cN^+}\boep^+_3\,f
 \left( \sigma_{AB}\sigma^{AB}-\dfrac12\vartheta^2 \right).
\eeq
For the constant boost, $f=-v\beta=\kappa_\beta v$,
and hence
\beq
 \sf_\xi[\Delta\cN^+]
 = \dfrac{\kappa_\beta}{8\pi}
 \int_{\Delta\cN^+}\boep^+_3\,v
 \left(  \dfrac12\vartheta^2-\sigma_{AB}\sigma^{AB} \right).
\eeq
Substituting this and \cref{eq:WZ_charge_dyn_entropy} in the vacuum flux balance \cref{eeq:conserv_arb_slice}, we obtain
\beq\label{eq:cd_vacuum_entropy_balance}
\int_{\Delta\cN^+}\boep^+_3\,v
 \left( \dfrac12\vartheta^2-\sigma_{AB}\sigma^{AB} \right)
= A(\cB) - \left[ A(\cC^+)-\int_{\cC^+}\boep^+_2\,v\vartheta \right]
= 4\left(\cS[\cB]-\cS[\cC^+]\right).
\eeq
For a non-expanding, shear-free boundary, the flux vanishes and the entropy is independent of the cross-section. For a shear-free vacuum boundary, the integrand on the left hand side is non-negative, so the entropy decreases towards the future. In general, however, the expansion and shear contributions have opposite signs, and this balance law does not by itself determine the ordering of the entropies of two cross-sections. The general form of this balance, including matter, arbitrary cuts and arbitrary dimension, and the resulting local characterisation of entropy decrease, are given in \cref{sec:exact_balance}.

\section{Physical process first law \texorpdfstring{$\And$}{} second law}\label{sec:physical_process_1st_law}
The comparison first law of causal diamond has been well established in the literature, but the physical process version of the first law is still underdeveloped. There have been a few discussions of it in particular settings \cite{Jacobson:2018ahi}, in the context of light-cone thermodynamics \cite{DeLorenzo:2017tgx, Barbosa:2026bmv} and stretched light-cone thermodynamics \cite{Svesko:2020yxo, Parikh:2017aas}, but a concrete formula is still missing. Here, we will try to provide a local in time physical process version of the first law for causal diamonds inspired by the similar law for black holes \cite{Hollands:2024vbe, Visser:2024pwz}. Although we have discussed the causal diamonds in four-dimensional spacetimes throughout most of the paper, in the following, we will consider a general $d$-dimensional spacetime. In contrast to black hole thermodynamics, it is also possible to derive the first law of thermodynamics for a physical process in the first order perturbation, even when no matter is present. 

We begin with an exact, non-perturbative balance in Einstein gravity. It fixes the orientation used throughout and identifies which terms raise or lower the entropy. We then turn to the physical process law for a general theory.

\subsection{Exact entropy balance in Einstein gravity}
\label{sec:exact_balance}

Throughout this subsection we make the following assumptions.
\begin{enumerate}[label=(A\arabic*)]
 \item \label{ass:regular}
 $(\cM,g_{ab})$ is a $d$-dimensional spacetime, $d\ge3$, containing
 a causal diamond whose future boundary contains a smooth portion
 $\cN^+_{\rm reg}\subset\cN^+$. Its generators are affinely
 parametrised by $v$, with $v=0$ at $\cB$, and have no caustics
 in the region considered, so that the area density
 $\sfa=\sqrt{\det q_{AB}}$ is positive there.
 \item \label{ass:cuts}
 Cuts are written as $\cC=\{v=v_\cC(x^A)\}$, where $x^A$ label the
 generators. We write $\cC_1\preceq\cC_2$ if
 $v_{\cC_1}(x)\le v_{\cC_2}(x)$ for all $x$. The region
 $\cN^+_{12}$ between them is oriented by
 $\boep^+_{d-1}=\ed v\wedge\boep^+_{d-2}$, so that
 $\partial\cN^+_{12}=\cC_2\cup(-\cC_1)$, as shown in \cref{fig:cd_arb_slice}.
 \item \label{ass:einstein}
 Einstein's equation $G_{ab}+\Lambda g_{ab}=8\pi T_{ab}$ holds
 with an arbitrary stress tensor $T_{ab}$.
\end{enumerate}
Since $\boep^+_{d-2}=\sfa\,\ed^{d-2}x$ has no $\ed v$ component, the
entropy \cref{eq:dyn_entropy_cd} of an arbitrary, not necessarily
uniform cut is
$\cS[\cC]=\frac14\int\ed^{d-2}x\,s\big(v_\cC(x),x\big)$, with the
entropy density
\beq\label{eq:entropy_density}
 s:=\sfa\,(1-v\vartheta)=\sfa-v\,\partial_v\sfa ,
\eeq
with $\vartheta=\partial_v\ln\sfa$.
\begin{definition}[Signed flux density]
\label{def:signed_flux_density}
On $\cN^+$, define the scalar
\beq\label{eq:signed_flux_density}
 \varphi
 :=
 -\kappa_\beta v\left[
 T_{ab}\ell^a\ell^b
 +\frac{1}{8\pi}\left(\sigma_{AB}\sigma^{AB}-c_d\,\vartheta^2\right)
 \right],
 \qquad
 c_d:=\frac{d-3}{d-2}.
\eeq
The matter flux form $\bos e=-\xi^bT_b{}^c\epsilon_c$, with
$\epsilon_c\equiv\epsilon_{cb_2\cdots b_d}$, which is the volume form (suppressing the other indices) and the
gravitational WZ flux form then satisfy $\bos e+\bos{\sf}_\xi=\varphi\,\boep^+_{d-1}$. Accordingly,
\beq\label{eq:signed_total_flux}
 \Delta\cE+\sf_\xi[\cN^+_{12}]
 =\int_{\cN^+_{12}}\varphi\,\boep^+_{d-1},
 \qquad
 \Delta\cE:=\int_{\cN^+_{12}}\bos e .
\eeq
\end{definition}
Since $\varphi$ is a scalar density relative to the positive measure
$\ed v\wedge\boep^+_{d-2}$, statements such as $\varphi\ge0$ are free of
any further orientation convention.

\begin{lemma}[Kinematic identity]
\label{lem:kinematic_identity}
For any smooth positive $\sfa(v,x)$, the density
\cref{eq:entropy_density} satisfies
\beq
 \partial_v s=-v\,\partial_v^2\sfa .
\eeq
\end{lemma}
\begin{proof}
$\partial_v s=\partial_v(\sfa-v\partial_v\sfa)
=\partial_v\sfa-\partial_v\sfa-v\partial_v^2\sfa=-v\,\partial_v^2\sfa$.
\end{proof}
The expansion correction in $\cS$ is precisely what removes the
first derivative of the area. The entropy therefore responds only to
the acceleration of the area, which the Raychaudhuri equation relates
to the fluxes.

\begin{theorem}[Exact entropy balance on $\cN^+$]
\label{thm:exact_balance}
Under \ref{ass:regular}--\ref{ass:einstein}, for any cuts
$\cC_1\preceq\cC_2$ of $\cN^+_{\rm reg}$,
\beq\label{eq:exact_balance}
 -\frac{\kappa_\beta}{2\pi}\Big(\cS[\cC_2]-\cS[\cC_1]\Big)
 =\Delta\cE+\sf_\xi[\cN^+_{12}]
 =\int_{\cN^+_{12}}\varphi\,\boep^+_{d-1}.
\eeq
Equivalently, along every generator,
\beq\label{eq:exact_balance_local}
 \partial_v s=-\frac{8\pi}{\kappa_\beta}\,\sfa\,\varphi
 =v\,\sfa\left(\sigma_{AB}\sigma^{AB}-c_d\vartheta^2
 +8\pi T_{ab}\ell^a\ell^b\right).
\eeq
\end{theorem}
\begin{proof}
From $\partial_v\sfa=\sfa\vartheta$ and the Raychaudhuri equation
\cref{eq:ray-focusing},
\beq
 \partial_v^2\sfa
 =\sfa\left(\partial_v\vartheta+\vartheta^2\right)
 =\sfa\left(c_d\vartheta^2-\sigma_{AB}\sigma^{AB}
 -8\pi T_{ab}\ell^a\ell^b\right).
\eeq
\Cref{lem:kinematic_identity} then gives \cref{eq:exact_balance_local}.
Integrating along each generator from $v_{\cC_1}(x)$ to $v_{\cC_2}(x)$,
and then over the generators,
\beq
 \cS[\cC_2]-\cS[\cC_1]
 =\frac14\int\ed^{d-2}x\int_{v_{\cC_1}(x)}^{v_{\cC_2}(x)}
 \partial_v s\,\ed v
 =-\frac{2\pi}{\kappa_\beta}\int_{\cN^+_{12}}
 \varphi\,\sfa\,\ed v\,\ed^{d-2}x .
\eeq
Finally, $\sfa\,\ed v\,\ed^{d-2}x=\boep^+_{d-1}$ with the orientation
as defined in \ref{ass:cuts}.
\end{proof}
For uniform cuts, \cref{eq:exact_balance} reduces to
\cref{eq:ray-exact-balance}, and in four-dimensional vacuum it
reduces to \cref{eq:cd_vacuum_entropy_balance}.

\begin{corollary}[Second law on $\cN^+$]
\label{cor:exact_integrated}
For cuts $\cC_1\preceq\cC_2$ of $\cN^+_{\rm reg}$,
\beq
 \Delta\cE+\sf_\xi[\cN^+_{12}]\ge0
 \quad\Longleftrightarrow\quad
 \cS[\cC_2]\le\cS[\cC_1],
\eeq
with strict inequality on one side if and only if it holds on the
other.
\end{corollary}
\begin{proof}
Follows directly from  \cref{eq:exact_balance} as $\kappa_\beta>0$.
\end{proof}

\begin{corollary}[Local second law]
\label{cor:exact_local}
Let $\cU\subset\cN^+_{\rm reg}$ be an open region. The following are equivalent:
\begin{enumerate}[label=(\roman*)]
 \item $\varphi\ge0$ on $\cU$, that is,
 $8\pi T_{ab}\ell^a\ell^b+\sigma_{AB}\sigma^{AB}\le c_d\,\vartheta^2$
 at every point of $\cU$;
 \item $\cS[\cC']\le\cS[\cC]$ for every pair of cuts
 $\cC\preceq\cC'$ that coincide outside $\cU$.
\end{enumerate}
In particular, if $\varphi\ge0$ on $\cN^+_{12}$, then $\cS$ is
non-increasing along every monotone family of cuts inside
$\cN^+_{12}$, whether uniform or not, and not just between its two end cuts.
\end{corollary}
\begin{proof}
(i)$\Rightarrow$(ii):\\ Up to a set of measure zero, the region between
$\cC$ and $\cC'$ lies in $\cU$; applying \cref{thm:exact_balance} gives the reuslt.

(ii)$\Rightarrow$(i):\\
Suppose $p=(v_0,x_0)\in \cU$, where $x=(x^A)$ are coordinates on the space of generators and $B_\eta(x_0)=\{x:\,|x-x_0|<\eta\}$ is an open ball.
By continuity there are $\veps,\eta>0$ with $\varphi<0$ on the box
$[v_0,v_0+\veps]\times B_\eta(x_0)\subset \cU$. Let $\cC$ be a cut
with $v_\cC(x)=v_0$ for $x\in B_\eta(x_0)$, and let $\cC'$ be $\cC$
pushed to the future by a bump function $0\le b(x)\le\veps$, not
identically zero and supported in $B_\eta(x_0)$.
Then \cref{eq:exact_balance} gives
$\cS[\cC']-\cS[\cC]
=-\frac{2\pi}{\kappa_\beta}\int_{\rm bump}\varphi\,\boep^+_{d-1}>0$,
contradicting (ii).
\end{proof}

\begin{corollary}[Global decrease between the edge and the corner]
\label{cor:global_decrease}
Let $d\ge4$, and suppose that the generators of $\cN^+$ remain smooth
from $\cB$ up to a regular future corner $p^+$, with no caustics before
$p^+$. Then $\cS[\cC]\to0$ as $\cC\to p^+$, and
\beq\label{eq:global_decrease}
 \Delta\cE+\sf_\xi[\cN^+]
 =\frac{\kappa_\beta}{2\pi}\,\frac{A(\cB)}{4}>0,
 \qquad
 \cS[\cB]=\frac{A(\cB)}{4}>\cS[p^+]=0,
\eeq
irrespective of the matter content. The entropy therefore always nonincreasing along $\cN^+$. Monotonicity on intermediate cuts requires the
hypothesis of \cref{cor:exact_local}.
\end{corollary}
\begin{proof}
In Riemann normal coordinates at $p^+$, let
$\lambda=\tilde V(x)-v$ be the affine distance to the corner along each
generator. Then
$\sfa=c(x)\lambda^{d-2}\big(1+\mathcal O(\lambda^2)\big)$ with
$c(x)>0$, and
$\partial_v\sfa=-(d-2)c(x)\lambda^{d-3}\big(1+\mathcal O(\lambda^2)\big)$.
Hence $s=\sfa-v\partial_v\sfa=\mathcal O(\lambda^{d-3})\to0$ for
$d\ge4$. Applying \cref{eq:exact_balance} with $\cC_1=\cB$, where
$v=0$ and $s=\sfa$, and $\cC_2\to p^+$ gives the result, which is equivalent to the statement that the WZ charges vanish at the corners
\cite{Chandrasekaran:2019ewn}.
\end{proof}

On the past boundary $\cN^-$, we take the future directed affine generator $n^a=\partial_u$, where $u\in[\tilde U,0)$ and $u=0$ at $\cB$, together with the matched boost vector $\xi^a=-\kappa_\beta u\,n^a$. The charge of \cref{sec:dyn_entropy_cd}
gives $\sq_\xi[\cC^-]=-\frac{\kappa_\beta}{2\pi}\cS^-[\cC^-]$, with
$\cS^-[\cC]=\frac14\int_\cC\boep^-_{d-2}(1-u\vartheta_-)$. The same
computation, with the density $s^-=\sfa_-(1-u\vartheta_-)$, gives
\beq
 \partial_u s^-=-u\,\partial_u^2\sfa_-
 =|u|\,\sfa_-\left(c_d\vartheta_-^2-\sigma^-_{AB}\sigma_-^{AB}
 -8\pi T_{ab}n^an^b\right).
\eeq
Towards the future, that is towards $\cB$, the expansion now
increases $\cS^-$, while shear and matter obeying the null energy
condition decrease it. For $d\ge4$ and regular corners, the entropy rises
from zero at $p^-$ to $A(\cB)/4$ at $\cB$ and falls back to zero at
$p^+$, consistently with $\sf_\xi[\cN^+]=-\sf_\xi[\cN^-]$.

\subsection{Presence of perturbative matter}

Consider a diffeomorphism invariant theory of gravity whose
Lagrangian may contain an arbitrary number of derivatives.
We restrict ourselves to the case when the only dynamical
field in the gravitational Lagrangian is the metric.
Let $g_{ab}$ be a vacuum solution admitting a causal diamond,
and let $\xi^a=\kappa_\beta v\ell^a$ preserve the boundary
structure.
We assume that a covariant WZ boundary potential has been
chosen and that the associated localised charge admits
the entropy relation
\beq
 \bos{\sq}_\xi
 =
 -\dfrac{\kappa_\beta}{2\pi}\bots.
\eeq
These ingredients have been constructed above for Einstein
gravity. For a higher-derivative theory, their existence
and the corresponding entropy form are additional
assumptions of the following statement.

Now, consider perturbations $\dg$ sourced by a matter
stress-energy tensor $\delta T_{ab}$ and satisfying the
linearised equations of motion. The perturbed matter
energy flux form is
\beq\label{eq:cd_matter_flux_form}
 \delta\bos{e}
 =
 -\xi^b\delta T_b^{\ c}\epsilon_c
 =
 -\delta T_{ab}\xi^a\ell^b\,\boep^+_{d-1},
\eeq
where $\epsilon_c$ is the spacetime
volume form with its first index displayed, and pullback to
$\cN^+$ is understood.
To write this in the two-$\xi$ notation of
\cite{Hollands:2024vbe}, define, for $v>0$,
\beq\label{eq:cd_HWZ_volume_form}
 \tilde\boep
 =
 -\dfrac{1}{\kappa_\beta v}\boep^+_{d-1}.
\eeq
Then
\beq
 \delta\bos{e}
 =
 \delta T_{ab}\xi^a\xi^b\tilde\boep,
 \qquad
 \xi\i\tilde\boep=-\boep^+_{d-2}.
\eeq
Although $\tilde\boep$ is singular at $v=0$, the energy flux form is regular there, as is manifest in \cref{eq:cd_matter_flux_form}.

\begin{theorem}[Physical process version of the first law]
\label{thm:phys_process_1stlaw}
Let $\cC_1\preceq\cC_2$ be cuts of the future null boundary of a
causal diamond, for instance its intersections with two slices
$\Sigma(t_1)$ and $\Sigma(t_2)$, and let $\cN^+_{12}$ be oriented
by the background form $\boep^+_{d-1}=\ed v\wedge\boep^+_{d-2}$, so
that $\partial\cN^+_{12}=\cC_2\cup(-\cC_1)$. Let $\dg$ satisfy the
linearised equations sourced by $\delta T_{ab}$, with the boundary
structure, $\ell^a$, $v$ and $\kappa_\beta$ held fixed. Define the
perturbed entropy difference by
\beq
 \Delta\delta\cS
 =
 \delta\cS[\cC_2]-\delta\cS[\cC_1],
\eeq
and the signed perturbed matter energy flux by
\beq
 \Delta\delta\cE
 =
 \int_{\cN^+_{12}}\delta\bos{e}
 =
 \int_{\cN^+_{12}}
 \delta T_{ab}\xi^a\xi^b\tilde\boep.
\eeq
Then
\beq\label{eqn:pp1st}
 -\dfrac{\kappa_\beta}{2\pi}\Delta\delta\cS
 =
 \Delta\delta\cE
 +
 \delta\sf_\xi[\cN^+_{12}].
\eeq
\end{theorem}

\begin{proof}
We begin with the fundamental identity
\cite{Hollands:2012sf,Hollands:2024vbe},
\beq
 \bom(g;\dg,\cL_\xi g) = \xi\i(\boe\,\dg) +\xi^a\delta\boc_a
+ \ed\left[ \delta\boq_\xi-\xi\i\bth(g;\dg) \right].
\eeq
Here, $\boe$ denotes the equations of motion and $\boc_a$
is the constraint $(d-1)$-form.
The background is on shell, and the sourced linearised
equations imply
\beq
 \xi^a\delta\boc_a=-\delta\bos{e}.
\eeq
Pulling back to $\cN^+$ therefore gives
\beq
 \ul{\bom}(g;\dg,\cL_\xi g)
 =
 -\delta\bos{e}
 +
 \ed\left[
  \delta\boq_\xi-\xi\i\bth(g;\dg)
 \right].
\eeq
Using the WZ charge variation,
\beq
 \delta\bos{\sq}_\xi
 =
 \delta\boq_\xi-\xi\i\bth+\xi\i\bTh,
\eeq
we obtain
\begin{align*}
 \ed\delta\bos{\sq}_\xi
 &=
 \delta\bos{e}
 +\ul{\bom}(g;\dg,\cL_\xi g)
 +\ed\left[\xi\i\bTh(g;\dg)\right]
 \\
 &=
 \delta\bos{e}
 +\ul{\bom}(g;\dg,\cL_\xi g)
 +\cL_\xi\bTh(g;\dg)
 \\
 &=
 \delta\bos{e}
 +\delta\!\left[\bTh(g;\cL_\xi g)\right].
\end{align*}
The second equality uses Cartan's formula on $\cN^+$,
and the last equality follows from covariance of the
boundary symplectic potential and $\delta\xi^a=0$.
Using \cref{sympectic_flux} and the charge-entropy
relation, we find
\beq\label{pplocal}
 \ed\delta\bos{\sq}_\xi
 =
 -\dfrac{\kappa_\beta}{2\pi}\ed\delta\bots
 =
 \delta\bos{e}+\delta\bos{\sf}_\xi.
\eeq
Since
$\partial\cN^+_{12}=\cC_2\cup(-\cC_1)$,
integration gives
\beq
 \int_{\cN^+_{12}}\ed\delta\bots
 =
 \delta\cS[\cC_2]-\delta\cS[\cC_1].
\eeq
Thus, integration of \cref{pplocal} proves
\cref{eqn:pp1st}.
\end{proof}
In Einstein gravity about a vacuum background, \cref{eqn:pp1st} is
also the first variation of the exact balance
\cref{eq:exact_balance}, because the variation commutes with the
$v$-integration at fixed $v$ and fixed generator labels. This gives
an independent proof, carried out explicitly in appendix
\ref{sec:dyn_entropy_raychaudhuri}.
The signed flux $\Delta\delta\cE$ should be distinguished
from the positive outward matter energy. In our convention,
\beq
 \Delta\delta\cE_{\rm out}
 :=
 \int_{\cN^+_{12}}
 \delta T_{ab}\xi^a\ell^b\,\boep^+_{d-1}
 =
 -\Delta\delta\cE.
\eeq
Consequently, the null energy condition
$\delta T_{ab}\ell^a\ell^b\ge0$ implies
$\Delta\delta\cE_{\rm out}\ge0$ and
$\Delta\delta\cE\le0$.

\begin{corollary}[Linearised second law on $\cN^+$]
\label{cor:linearised_second_law}
Under the assumptions of \cref{thm:phys_process_1stlaw}, define the
regular scalar density $\varphi_\delta$ on $\cN^+$ by
\beq\label{eq:linearised_flux_density}
 \delta\bos{e}+\delta\bos{\sf}_\xi
 =\varphi_\delta\,\boep^+_{d-1}.
\eeq
Due to \cref{eq:cd_matter_flux_form}, this is equivalent to the two-$\xi$
expression $\delta T_{ab}\xi^a\xi^b\boep+\delta\bos{\sf}_\xi$, but
$\varphi_\delta$ remains regular at $v=0$. Then,for cuts $\cC_1\preceq\cC_2$,
 \beq
  \int_{\cN^+_{12}}\varphi_\delta\,\boep^+_{d-1}\ge0
  \quad\Longleftrightarrow\quad
  \Delta\delta\cS\le0.
 \eeq
\end{corollary}

\begin{proof}
Follows from \cref{eqn:pp1st}, since $\kappa_\beta>0$.
\end{proof}
For a general nonstationary background, the gravitational flux variation has no definite sign. Therefore, the balance law alone does not establish a general linearised second law from the null energy condition; the decrease of
$\delta\cS$ requires the perturbation to strengthen the focusing of the generators by more than the matter energy it carries.
\vspace{.3cm}
\begin{remark}[Explicit linearised criterion]
Consider Einstein gravity about a shear-free vacuum background,
$\sigma_{AB}=T_{ab}=0$, which includes all spherical diamonds in
maximally symmetric spacetimes. Then
$\delta(\sfa\,\sigma_{AB}\sigma^{AB})=0$ at first order, and the local
condition $\varphi_\d \ge 0$ reads
\beq\label{eq:linearised_criterion}
 8\pi\,\sfa\,\delta T_{ab}\ell^a\ell^b
 \le c_d\,\delta\!\left(\sfa\vartheta^2\right)
 =c_d\left(\vartheta^2\delta\sfa+2\sfa\vartheta\,\delta\vartheta\right),
 \qquad
 \delta\sfa=\tfrac12 h\,\sfa,
 \quad
 \delta\vartheta=\tfrac12\cL_\ell h .
\eeq
The perturbation must therefore increase the focusing term
$\sfa\vartheta^2$ by more than the matter energy it carries. A thin
null shell crossing the Minkowski diamond, worked out in
\cref{app:wz-shell}, shows that this condition fails across the shell
and holds after it.
\end{remark}

\subsection{Vacuum perturbation}

For perturbations of a stationary black hole, the first-order entropy difference vanishes in vacuum, and the leading physical process law occurs in the second order \cite{Hollands:2024vbe}. For a causal diamond with a nonstationary boundary, the first-order gravitational flux variation need not vanish. In the following, we consider both first- and second order vacuum perturbations.

\subsubsection{First order vacuum perturbation}

The first-order physical process law follows directly from
\cref{thm:phys_process_1stlaw}. In the absence of matter,
$\delta T_{ab}=0$, and \cref{pplocal} reduces to
\beq\label{eq:cd_vacuum_first_order_local}
 -\dfrac{\kappa_\beta}{2\pi}\ed\delta\bots
 =
 \delta\bos{\sf}_\xi.
\eeq
Integrating between the cross-sections gives
\beq\label{eq:cd_vacuum_first_order_integrated}
 -\dfrac{\kappa_\beta}{2\pi}
 \left(
  \delta\cS[\cC_2]-\delta\cS[\cC_1]
 \right)
 =
 \delta\sf_\xi[\cN^+_{12}].
\eeq
This is the variation of the vacuum conservation law in
\cref{sec:dyn_entropy_cd}. A non-negative integrated flux
variation is equivalent to $\Delta\delta\cS\le0$, and
\cref{cor:linearised_second_law} gives the corresponding local
statement with $\varphi_\delta$ the density of
$\delta\bos{\sf}_\xi$.
For Einstein gravity, the entropy form and gravitational flux form are
\begin{subequations}
\beq\label{eq:cd_vacuum_entropy_form}
 \bots =  \dfrac14(1-v\vartheta)\boep^+_{d-2},
\eeq
and
\beq\label{eq:cd_vacuum_flux_form}
 \bos{\sf}_\xi = -\dfrac{\kappa_\beta v}{8\pi}\boep^+_{d-1} \left( \sigma_{AB}\sigma^{AB}-\dfrac{d-3}{d-2}\vartheta^2\right).
\eeq
\end{subequations}
The coefficient $(d-3)/(d-2)$ is $1/2$ in four dimensions.
With the affine structure and cross-sections fixed,
\beq
 \delta\boep^+_{d-2}  = \dfrac{h}{2}\boep^+_{d-2},
 \qquad
 \delta\vartheta=\dfrac12\cL_\ell h.
\eeq
Variation of \cref{eq:cd_vacuum_entropy_form} gives
\beq\label{eq:cd_entropy_first_variation}
 \delta\cS[\cC] = \dfrac18
 \int_\cC\boep^+_{d-2}
 \left[ h(1-v\vartheta)-v\cL_\ell h \right].
\eeq
Similarly, variation of \cref{eq:cd_vacuum_flux_form} gives
\beq\label{eq:cd_flux_first_variation}
 \delta\sf_\xi[\cN^+_{12}] = -\dfrac{\kappa_\beta}{8\pi}
\int_{\cN^+_{12}}\boep^+_{d-1}\,v
\Bigg[\dfrac{h}{2}
\left( \sigma_{AB}\sigma^{AB}-\dfrac{d-3}{d-2}\vartheta^2\right)
 +\delta\sigma_{AB}\sigma^{AB} +\sigma_{AB}\delta\sigma^{AB}  -\dfrac{d-3}{d-2} \vartheta\cL_\ell h \Bigg].
\eeq
The variation of $\sigma^{AB}$ includes the variation
of the inverse metric used to raise its indices.\\

\subsubsection{Second order vacuum perturbation}

Consider a one-parameter family of metrics
\beq
 g_{ab}(\lambda)
 =
 g_{ab}
 +\lambda\,\ud g_{ab}
 +\dfrac{\lambda^2}{2}\ud^2g_{ab}
 +\mathcal{O}(\lambda^3).
\eeq
Here, $\ud$ denotes differentiation with respect to
$\lambda$, evaluated at $\lambda=0$, rather than the
field-space exterior derivative $\delta$.
The boundary identification, cross-sections, affine
structure, and the symmetry generator are fixed along
the family. We assume that the vacuum equations hold
through second order.
The fundamental identity along this family is
\beq
 \bom(g;\ud g,\cL_\xi g) = \xi\i(\boe\,\ud g) +\xi^a\ud\boc_a
+ \ed\left[ \ud\boq_\xi-\xi\i\bth(g;\ud g) \right].
\eeq
When differentiating this identity, the dependence of $\bom$ on its background metric must also be varied. For fixed tensor arguments $h_{ab}$ and $k_{ab}$, define
\beq
 \ud_g\bom(g;h,k) :=
  \left. \dfrac{\ed}{\ed\lambda} \bom(g(\lambda);h,k) \right|_{\lambda=0}.
\eeq
Thus,
\beq\label{eq:cd_second_variation_omega}
 \ud\!\left[\bom(g;\ud g,\cL_\xi g)\right]
 = \ud_g\bom(g;\ud g,\cL_\xi g) + \bom(g;\ud^2g,\cL_\xi g)
 + \bom(g;\ud g,\cL_\xi\ud g).
\eeq
All quantities are evaluated at the reference metric after differentiation.\footnote{At each metric $g$, the symplectic current defines a bilinear map
\[
 \bom_g:\cT_g\cF\times \cT_g\cF
 \longrightarrow \mathbb{\Omega}^{(d-1;2)}(\cM\times\cF),
\]
where $\cT_g \cF$ denotes the tangent space at point $g$ in the field space $\cF$ and $\mathbb{\Omega}^{(d-1;2)}(\cM\times\cF)$ denotes space of $(d-1;2)$-forms. Using metric-field coordinates to identify tangent vectors along $g(\lambda)$, the term $\ud_g\bom$ is the derivative of this bilinear map with its arguments held fixed. It accounts for the metric dependence of the presymplectic pairing. The resulting expression is evaluated at the reference metric $g=g(0)$.}
Varying the fundamental identity gives
\beq
 \ud\!\left[\bom(g;\ud g,\cL_\xi g)\right]
 = \xi\i\left( \ud\boe\,\ud g+\boe\,\ud^2g \right) + \xi^a\ud^2\boc_a 
 + \ed\left[ \ud^2\boq_\xi-\xi\i\ud\bth(g;\ud g) \right].
 \eeq
On-shell, the terms containing the equation of motions and constraints vanish. Pulling back to $\cN^+$ therefore yields
\beq
 \ud\!\left[ \ul{\bom}(g;\ud g,\cL_\xi g) \right]
 = \ed\left[ \ud^2\boq_\xi-\xi\i\ud\bth(g;\ud g) \right].
\eeq
Pullbacks of the remaining boundary expressions are understood below. Recall that the chosen WZ prescription satisfies
\beq
 \bTh=\ul{\bth}-\delta\bos{\vr},
 \qquad
 \bos{\sq}_\xi=\boq_\xi-\xi\i\bos{\vr}.
\eeq
Within the one-parameter family, these relations imply
\beq
\begin{aligned}
 \ud\bTh(g;\ud g)
 &=
 \ud\bth(g;\ud g)-\ud^2\bos{\vr},
 \\
 \ud^2\bos{\sq}_\xi
 &=
 \ud^2\boq_\xi-\xi\i\ud^2\bos{\vr}.
\end{aligned}
\eeq
Adding the boundary term to the varied fundamental identity gives
\beq
\ud\!\left[ \ul{\bom}(g;\ud g,\cL_\xi g) \right]
 + \ed\left[ \xi\i\ud\bTh(g;\ud g) \right]
 = \ed\ud^2\bos{\sq}_\xi.
\eeq
We denote the complete second order gravitational flux form on the left hand side by
\beq\label{eq:cd_second_order_flux_current}
 \bog
 := \ud\!\left[  \ul{\bom}(g;\ud g,\cL_\xi g) \right]
 + \ed\left[ \xi\i\ud\bTh(g;\ud g) \right].
\eeq
For a nonstationary background, this expression contains
all three terms in \cref{eq:cd_second_variation_omega}.
It need not be independent of $\ud^2g_{ab}$ or quadratic
only in $\ud g_{ab}$.
When $\cL_\xi g=0$, the first two terms on the right hand side of \cref{eq:cd_second_variation_omega} vanish. Then \cref{eq:cd_second_order_flux_current} takes the modified canonical-energy form
\cite{Hollands:2012sf,Hollands:2024vbe},
\beq
 \bog
 =
 \ul{\bom}(g;\ud g,\cL_\xi\ud g)
 + \ed\left[ \xi\i\ud\bth(g;\ud g) - \xi\i\ud^2\bos{\vr} \right].
\eeq
Recall that, for every metric along the family, the pullback of the
symplectic current can be written in terms of the boundary symplectic
potential,
\beq
 \ul{\bom}\left(g;\ud g,\cL_\xi g\right)
 = \ud\!\left[\bTh(g;\cL_\xi g)\right] - \cL_\xi \bTh(g;\ud g). \nonumber
\eeq
This relation follows from the covariance of $\bTh$ and $\delta\xi^a=0$
and does not use the field equations. Differentiating it once more along
the family, and using $\cL_\xi\bTh=\ed(\xi\i\bTh)$ on $\cN^+$, yields
\beq
 \bog = \ud^2\!\left[\bTh(g;\cL_\xi g)\right]
 = \ud^2\bos{\sf}_\xi.
\eeq
Taking the second order variation of it and applying Cartan's formula yields
\beq
 \bog = \ud^2\!\left[\bTh(g;\cL_\xi g)\right]
 = \ud^2\bos{\sf}_\xi.
\eeq
Hence, the second order vacuum balance law is
\beq\label{eq:cd_second_order_local_balance}
 \bog = \ud^2\bos{\sf}_\xi = \ed\ud^2\bos{\sq}_\xi
 = -\dfrac{\kappa_\beta}{2\pi}\ed\ud^2\bots.
\eeq
Here, the charge-entropy relation is assumed to hold through second order; this is the case for the Einstein gravity functional used above.
Integration yields 
\beq\label{eq:cd_second_order_integrated_balance}
 -\dfrac{\kappa_\beta}{2\pi}
 \left(
  \ud^2\cS[\cC_2]-\ud^2\cS[\cC_1]
 \right)
 =
 \int_{\cN^+_{12}}\bog
 =
 \ud^2\sf_\xi[\cN^+_{12}].
\eeq
Consequently, with inequalities defined relative to
$\boep^+_{d-1}$,
\beq
 \bog\ge0
 \quad\Longrightarrow\quad \Delta\ud^2\cS \le 0.
\eeq
Conversely, since \cref{eq:cd_second_order_integrated_balance} holds
for arbitrary cuts, the localisation argument in the proof of
\cref{cor:exact_local} shows that $\bog\ge0$ on an open set $U$ if
and only if $\ud^2\cS[\cC']\le\ud^2\cS[\cC]$ for every pair of cuts
$\cC\preceq\cC'$ that coincide outside $U$.
For the explicit Einstein gravity variations, define
along the family
\beq
 h_{AB}(\lambda)=\ud q_{AB}(\lambda),
 \qquad
 h(\lambda)=q^{AB}(\lambda)h_{AB}(\lambda).
\eeq
In particular,
\beq
 \ud h=q^{AB}\ud^2q_{AB}-h_{AB}h^{AB}.
\eeq
The required variations of the area form and expansion are
\beq
\begin{aligned}
 \ud^2\boep^+_{d-2}
 &=
 \left(
  \dfrac{\ud h}{2}+\dfrac{h^2}{4}
 \right)\boep^+_{d-2},
 \\
 \ud\vartheta
 &=
 \dfrac12\cL_\ell h,
 \qquad
 \ud^2\vartheta
 =
 \dfrac12\cL_\ell\ud h.
\end{aligned}
\eeq
Thus, varying \cref{eq:cd_vacuum_entropy_form} twice gives
\beq\label{eq:cd_entropy_second_variation}
 \ud^2\cS[\cC] = \dfrac18\int_\cC\boep^+_{d-2}
 \left[ \left(\ud h+\dfrac{h^2}{2}\right)(1-v\vartheta)
  -v h\cL_\ell h  -v\cL_\ell\ud h \right].
\eeq
Similarly, varying \cref{eq:cd_vacuum_flux_form} twice gives
\beq\label{eq:cd_flux_second_variation}
\begin{aligned}
 \ud^2\sf_\xi[\cN^+_{12}]
 &= -\dfrac{\kappa_\beta}{8\pi} \int_{\cN^+_{12}}\boep^+_{d-1}\,v
 \Bigg[ \left( \dfrac{\ud h}{2}+\dfrac{h^2}{4} \right)
 \left( \sigma_{AB}\sigma^{AB} -\dfrac{d-3}{d-2}\vartheta^2 \right)
 + h\left( \ud\sigma_{AB}\sigma^{AB} +\sigma_{AB}\ud\sigma^{AB}
  -\dfrac{d-3}{d-2}\vartheta\cL_\ell h \right) \\
 &\qquad
 + \ud^2\sigma_{AB}\sigma^{AB} +\sigma_{AB}\ud^2\sigma^{AB}
 +2\ud\sigma_{AB}\ud\sigma^{AB} -\dfrac{d-3}{2(d-2)}(\cL_\ell h)^2
 -\dfrac{d-3}{d-2}\vartheta\cL_\ell\ud h \Bigg].
\end{aligned}
\eeq
All raised shear tensors are formed using
$q^{AB}(\lambda)$ before differentiation.

As a check, on a non-expanding, shear-free background,
and for vacuum perturbations obeying the horizon
boundary condition $\ud\vartheta=0$, the second order
flux reduces to
\beq
 \bog
 =
 -\dfrac{\kappa_\beta v}{4\pi}
 \boep^+_{d-1}\,
 \ud\sigma_{AB}\ud\sigma^{AB}
 \le0.
\eeq
Therefore, $\Delta\ud^2\cS\ge0$ in this restricted setting,
consistent with the stationary-horizon result expressed
in our orientation convention.

\section{Semiclassical entropy balance \texorpdfstring{$\And$}{} improved quantum null energy condition}
\label{sec:generalise_second_law_cd}

The physical process first law obtained in the previous section
relates the change in dynamical entropy to the matter and
gravitational fluxes through the null boundary. With our
conventions, positivity of the combined flux implies a decrease
of the perturbed dynamical entropy towards the future. In
semiclassical gravity, however, the expectation value of the
stress-energy tensor of quantum matter need not obey the
classical null energy condition. We therefore consider the
corresponding balance involving the dynamical gravitational
entropy and a contribution associated with the quantum matter
fields, and identify the condition under which their combined
first-order change is non-positive.

The generalised second law (GSL) was originally proposed in the
context of black hole thermodynamics
\cite{Bekenstein:1973ur,Bekenstein:1974ax} and has subsequently
been established under suitable assumptions for semiclassical
causal horizons \cite{Wall:2011gsl}. For first-order
perturbations of a stationary black hole, the quantum null
energy condition (QNEC) can be combined with the local physical
process first law to obtain a GSL involving the dynamical
entropy and a modified von Neumann entropy
\cite{Hollands:2024vbe}.\footnote{Recently a dual GSL has been proposed in \cite{Rignon-Bret:2026ttb}.} Here we examine the corresponding
construction for causal diamonds, retaining the gravitational
Wald-Zoupas flux.

We work in four spacetime dimensions, within semiclassical
Einstein gravity. Since the null boundary of a causal diamond need not have vanishing classical expansion, we consider the improved quantum null energy condition (INEC) proposed in \cite{Ben-Dayan:2023inz}. Its derivation from the quantum focusing conjecture (QFC) uses vanishing quantum expansion and a shear-free congruence. These conditions are not satisfied by a generic causal diamond. The field theoretic argument in \cite{BenDayan:2026inec} also requires an additional assumption. Accordingly, the argument below is conditional on the stated quantum inequality holding throughout the interval of cuts under consideration.

Let $\cC(v)$ denote cuts of the future null boundary $\cN^+$ generated by a uniform affine displacement along $\ell^a$. We consider a smooth portion of the boundary before its generators reach a caustic or the future corner, with $v\ge0$ measured from $\cB$. We write
\beq
 \sfa(v,y)=\sqrt{q(v,y)},
 \qquad
 \boep^+_2=\sfa\,d^2y,
 \qquad
 \vartheta=\frac{1}{\sfa}\frac{d\sfa}{dv}.
\eeq
The orientation conventions are
\beq
 \boep^+_3=dv\wedge\boep^+_2,
 \qquad
 \xi^a=\kappa_\beta v\,\ell^a,
 \qquad
 \kappa_\beta>0.
\eeq

Let $\cS_{vN}[v]$ be the entropy of the quantum fields on a consistently chosen exterior side of the cut $v=v(y)$. The regions associated with future deformations are chosen to be nested, with the remaining boundary data held fixed. Entropy differences and shape derivatives are understood in a common renormalisation prescription in which the quantum inequalities below are well defined. We denote the restriction to uniform cuts by
$\cS_{vN}(v)=\cS_{vN}[v(y)=v]$.
The quantum expansion is
\beq
\vartheta_\fq = \vartheta + \frac{4}{\sfa}
 \frac{\delta\cS_{vN}[v]}{\delta v(y)}.
\eeq
The four-dimensional INEC takes the form
\beq
 \sfa\,\langle T_{ab}\rangle\ell^a\ell^b
 \ge
 \frac{1}{2\pi}
 \left[
 \left(
 \frac{\delta^2\cS_{vN}[v]}{\delta v(y)^2}
 \right)_{\rm diag}
 -
 \frac{1}{2}\vartheta
 \frac{\delta\cS_{vN}[v]}{\delta v(y)}
 \right].
 \label{eq:INEC_cd}
\eeq
Here, the subscript denotes the diagonal second shape derivative. The additional term is proportional to the expansion distinguishes this expression from the usual QNEC. For an infinitesimal deformation of a narrow congruence of generators, the differential combination can be written as
\beq
 \frac{d^2\cS_{vN}}{dv^2}
 -
 \frac{1}{2}\vartheta\frac{d\cS_{vN}}{dv}
 =
 \sfa^{1/2}\frac{d}{dv}
 \left(
 \sfa^{-1/2}\frac{d\cS_{vN}}{dv}
 \right),
 \label{eq:INEC_derivative_identity}
\eeq
where the derivatives in this identity refer to that deformation.

For a uniform displacement, the full second derivative also
contains mixed shape derivatives on distinct generators.
Strong subadditivity makes their contribution non-positive
for the regions specified above \cite{Bousso:2015mna}, so that
\beq
 \frac{d^2\cS_{vN}(v)}{dv^2}
 \le
 \int_{\cC(v)}d^2y\,
 \left.
 \left(
 \frac{\delta^2\cS_{vN}[v]}{\delta v(y)^2}
 \right)_{\rm diag}
 \right|_{v(y)=v}.
 \label{eq:INEC_uniform_deformation}
\eeq
It is therefore useful to define
\beq
 \mathcal{I}_{\rm INEC}(v)
 \equiv
 \frac{d^2\cS_{vN}(v)}{dv^2}
 -
 \frac{1}{2}
 \int_{\cC(v)}d^2y\,
 \vartheta(v,y)
 \left.
 \frac{\delta\cS_{vN}[v]}{\delta v(y)}
 \right|_{v(y)=v}.
 \label{eq:I_INEC}
\eeq
Integrating \cref{eq:INEC_cd} and using
\cref{eq:INEC_uniform_deformation} yields
\beq
 \int_{\cC(v)}\boep^+_2\,
 \langle T_{ab}\rangle\ell^a\ell^b
 \ge
 \frac{1}{2\pi}\mathcal{I}_{\rm INEC}(v).
 \label{eq:integrated_INEC}
\eeq

We now apply this bound perturbatively. The quantum
inequality is evaluated on the fixed reference geometry,
and $\d$ acting on matter quantities denotes a variation
of the quantum state. We assume that the perturbation
satisfies
\beq
 \int_{\cC(v)}\boep^+_2\,
 \d\langle T_{ab}\rangle\ell^a\ell^b
 \ge
 \frac{1}{2\pi}\d\mathcal{I}_{\rm INEC}(v).
 \label{eq:pert_INEC}
\eeq
A sufficient condition is a differentiable one-sided family
of states, parametrised by $\lambda\ge0$, for which
\cref{eq:integrated_INEC} holds and is saturated at
$\lambda=0$ at every cut considered. The right derivative
of its non-negative gap then gives \cref{eq:pert_INEC}.
For a differentiable family admissible for both signs of
$\lambda$, saturation instead forces equality at first
order. These conditions must be checked when applying
the construction to a particular state.

As in the preceding section, the change
$\d\langle T_{ab}\rangle$ is treated as an external
first-order source for the gravitational perturbation.
The matter entropy derivatives above are evaluated before
including this backreaction, at leading semiclassical order.
The resulting change in gravitational entropy is determined
by the physical process first law.

Let $\cN^+_{v_0v}$ be the portion of $\cN^+$ between a fixed
reference cut $\cC(v_0)$ and $\cC(v)$. With our orientation,
the matter flux form is
\beq
 \d\bos{e}
 =
 -\kappa_\beta v\,
 \d\langle T_{ab}\rangle\ell^a\ell^b\,
 \boep^+_3.
 \label{eq:semiclassical_matter_flux}
\eeq
Consequently, \cref{eqn:pp1st} gives
\beq
 -\frac{\kappa_\beta}{2\pi}
 \frac{d}{dv}\d\cS
 =
 -\kappa_\beta v
 \int_{\cC(v)}\boep^+_2\,
 \d\langle T_{ab}\rangle\ell^a\ell^b
 +
 \frac{d}{dv}\d\sf_\xi[\cN^+_{v_0v}].
 \label{eq:local_pp_gsl}
\eeq
Combining this equality with \cref{eq:pert_INEC}, and using
$v\ge0$, we obtain
\begin{align}
 \frac{d}{dv}\d\cS
 +
 \frac{2\pi}{\kappa_\beta}
 \frac{d}{dv}\d\sf_\xi[\cN^+_{v_0v}]
 &=
 2\pi v
 \int_{\cC(v)}\boep^+_2\,
 \d\langle T_{ab}\rangle\ell^a\ell^b
 \nonumber\\
 &\ge
 v\,\d\mathcal{I}_{\rm INEC}(v).
 \label{eq:cd_INEC_pp}
\end{align}
We define a modified matter entropy along the chosen family of cuts, up to a $v$-independent constant, by
\beq
 \frac{d}{dv}\d\cS_q(v)
 \equiv
 -v\,\d\mathcal{I}_{\rm INEC}(v).
 \label{eq:Sq_definition}
\eeq
Equivalently,
\beq
 \d\cS_q(v) = \d\cS_q(v_0) - \int_{v_0}^{v}d\bar v\,
 \bar v\,\d\mathcal{I}_{\rm INEC}(\bar v).
 \label{eq:Sq_integrated}
\eeq
To retain the distinction between the quantum bound and
its saturation, we introduce its non-negative first-order
remainder,
\beq
 \mathsf{R}(v)
 \equiv
 2\pi
 \int_{\cC(v)}\boep^+_2\,
 \d\langle T_{ab}\rangle\ell^a\ell^b
 -
 \d\mathcal{I}_{\rm INEC}(v)
 \ge0.
 \label{eq:cd_INEC_remainder}
\eeq
Using \cref{eq:local_pp_gsl,eq:Sq_definition}, the combined
entropy then satisfies the exact first-order balance
\beq
 -\frac{\kappa_\beta}{2\pi}
 \frac{d}{dv}\left(\d\cS+\d\cS_q\right)
 =
 \frac{d}{dv}\d\sf_\xi[\cN^+_{v_0v}]
 -
 \frac{\kappa_\beta}{2\pi}\,
 v\,\mathsf{R}(v).
 \label{eq:cd_combined_entropy_balance}
\eeq
Integrating between two cuts with $v_2>v_1$ gives
\beq
 -\frac{\kappa_\beta}{2\pi}
 \left(\Delta\d\cS+\Delta\d\cS_q\right)
 =
 \d\sf_\xi[\cN^+_{v_1v_2}]
 -
 \frac{\kappa_\beta}{2\pi}
 \int_{v_1}^{v_2}dv\,v\,\mathsf{R}(v).
 \label{eq:integrated_cd_entropy_balance}
\eeq
Here $\Delta$ denotes the difference between the later and earlier cuts.

These identities give the semiclassical counterpart of the flux condition in \cref{sec:physical_process_1st_law}. The quantity whose positivity controls the decrease of the combined entropy is the gravitational flux minus the weighted remainder in the quantum energy bound.

\begin{corollary}[Conditional decrease of the combined entropy]
\label{cor:cd_combined_entropy_decrease}
Under the assumptions above, if
\beq
 \frac{d}{dv}\d\sf_\xi[\cN^+_{v_0v}]
 - \frac{\kappa_\beta}{2\pi}\, v\,\mathsf{R}(v) \ge0,
 \label{eq:cd_GSL_flux_condition}
\eeq
throughout the interval of cuts, then
\beq
 \frac{d}{dv}\left(\d\cS+\d\cS_q\right)\le 0.
 \label{eq:cd_entropy_decrease}
\eeq
For a specified pair of cuts, the necessary and sufficient integrated condition is
\beq
 \d\sf_\xi[\cN^+_{v_1v_2}] - \frac{\kappa_\beta}{2\pi}
 \int_{v_1}^{v_2}dv\,v\,\mathsf{R}(v) \ge 0,
 \label{eq:cd_integrated_flux_condition}
\eeq
which is equivalent to
\beq
 \Delta \d \cS + \Delta \d \cS_q \le 0.
\eeq
Strict positivity in \cref{eq:cd_integrated_flux_condition}
gives a strict decrease.
\end{corollary}

\begin{proof}
The statements follow directly from
\cref{eq:cd_combined_entropy_balance,eq:integrated_cd_entropy_balance}, since $\kappa_\beta>0$.
\end{proof}
In particular, suppose that the perturbative INEC is saturated throughout the interval, so that $\mathsf{R}(v)=0$. This includes the differentiable two-sided state families described above. Then
\cref{eq:integrated_cd_entropy_balance} reduces to
\beq
 -\frac{\kappa_\beta}{2\pi} \left(\Delta\d\cS+\Delta\d\cS_q\right)
 = \d\sf_\xi[\cN^+_{v_1v_2}],
\eeq
and hence
\beq
 \d \sf_\xi [\cN^+_{v_1v_2}] \ge 0
 \quad \Longleftrightarrow \quad
\Delta \d \cS + \Delta \d \cS_q \le 0.
 \label{eq:cd_GSL_saturated}
\eeq
Thus, when the quantum bound is saturated, positive gravitational flux implies a decrease of the combined perturbed entropy, directly paralleling the flux balance in the previous section. When the bound is not saturated, the gravitational flux must exceed the non-negative weighted remainder appearing in
\cref{eq:cd_integrated_flux_condition}.

The INEC alone also implies
\beq
 \frac{d}{dv}\left(\d\cS+\d\cS_q\right) \ge
 -\frac{2\pi}{\kappa_\beta} \frac{d}{dv}\d\sf_\xi[\cN^+_{v_0v}],
 \label{eq:generalise_second_law_cd}
\eeq
but this lower bound does not by itself determine the sign of the combined entropy change when the gravitational flux is positive. The decrease asserted in \cref{cor:cd_combined_entropy_decrease} therefore requires the stated condition on the full right-hand side of the balance law. On a nonstationary background, the
gravitational flux can contribute already at first order; its sign is not fixed by the second order vacuum analysis.

For comparison, in the stationary Killing limit, $\vartheta=0$, the quantum bound reduces to QNEC and \cref{eq:I_INEC} becomes
\beq
 \mathcal{I}_{\rm INEC}(v) = \frac{d^2\cS_{vN}(v)}{dv^2}.
\eeq
It follows from \cref{eq:Sq_definition} that
\beq
\d\cS_q = \d\cS_{vN} - v\,\frac{d}{dv}\d\cS_{vN} + \text{constant}.
 \label{eq:Sq_stationary}
\eeq
This agrees with the first variation of the modified von Neumann entropy introduced in \cite{Hollands:2024vbe}. The first-order gravitational flux vanishes for perturbations of a stationary Killing horizon in Einstein gravity, and \cref{eq:cd_combined_entropy_balance} then gives
\beq
 \frac{d}{dv} \left( \d \cS + \d \cS_q \right)
 = v\,\mathsf{R}(v) \ge 0.
 \label{eq:cd_GSL_zero_flux}
\eeq
The negative temperature $T_{(\beta)}=-\kappa_\beta/(2\pi)$ is consistent with both statements: the sign of the entropy change is determined by the full balance law, including the oriented matter flux and the remainder in the quantum inequality.

For nonzero expansion, \cref{eq:Sq_integrated} defines $\d\cS_q$ along the specified family of cuts. An independent statistical interpretation, and independence from a more general choice of deformation path, would require additional input. The monotonicity established here concerns the first-order change of this modified entropy sum under the stated flux condition. It does not determine the evolution of its nonstationary background value or establish the usual GSL for $A/4+\cS_{vN}$ on an arbitrary causal diamond. The black hole comparison is discussed in appendix \ref{sec:bh_gsl}.

\section{Spherically symmetric causal diamonds}
\label{sec:CD_MSS}

We now consider spherical causal diamonds in maximally
symmetric backgrounds and compare the results of
\cref{sec:dyn_entropy_cd,sec:physical_process_1st_law}
with the equilibrium relations reviewed in appendix
\ref{sec:cd_mss}.

Let $\S_t$ be the intersection of the diamond with a
constant-$t$ surface in the de Sitter static patch coordinates, 
\beq \label{eq:cd_mss_static_metric}
 ds^2 = -\left(1-\frac{r^2}{L^2}\right)dt^2
 + \frac{dr^2}{\left(1-\dfrac{r^2}{L^2}\right)} + r^2d\tom_{d-2}^2 
\eeq
with cross sections $\cC_t=\partial\S_t$ and de Sitter radius $L$. We focus on $t\ge0$, for which $\cC_t\subset\cN^+$. The maximal slice is $\S_0$ and
$\cC_0=\cB$. For the metric \eqref{eq:cd_mss_static_metric}, the radius of the
round cut is
\beq
 r_t=L\tanh\frac{R_*-t}{L},
 \qquad
 R=L\tanh\frac{R_*}{L}.
 \label{eq:VI_cut_radius}
\eeq
Similar expressions for flat and AdS spacetimes follow by the limits
specified in \cref{sec:cd_mss}. Throughout this
section, $K$ denotes the outward mean curvature of
the original edge $\cB$ within $\S_0$ (see appendix \ref{sec:cd_mss}).
The conformal Killing vector (CKV) $\z^a$ is given by \cite{Jacobson:2018ahi},
\beq\label{eq:cd_mss_ckv_static}
 \z = \frac{\kappa L^2}{R} \Bigg[
 \left( 1- \frac{\sqrt{1-R^2/L^2}}{\sqrt{1-r^2/L^2}}\cosh\frac{t}{L} \right)\partial_t
 -
 \frac{r}{L} \sqrt{\left(1-\frac{R^2}{L^2}\right)
 \left(1-\frac{r^2}{L^2}\right)} \sinh\frac{t}{L}\,\partial_r
 \Bigg].
\eeq
Its conformal factor and normal derivative on $\S_t$ are
\begin{align}
 \alpha
 &=
 -\frac{\kappa L}{R}
 \sqrt{\left(1-\frac{R^2}{L^2}\right)
       \left(1-\frac{r^2}{L^2}\right)}
 \sinh\frac{t}{L},
 \label{eq:VI_alpha}
 \\
 \dot{\alpha}
 \equiv u^a\nabla_a\alpha
 &=
 -\frac{\kappa}{R}\sqrt{1-\frac{R^2}{L^2}}
 \cosh\frac{t}{L}
 =
 -\frac{\kappa K}{d-2}\cosh\frac{t}{L},
 \label{eq:VI_alpha_dot}
\end{align}
where $u^a$ is the future directed unit normal to
$\S_t$, normalised as $u_au^a=-1$ \cite{Jacobson:2018ahi}.
In particular, $\dot{\alpha}$ is constant over each
of these slices, although $\alpha$ generally is not.
We write
\beq
 \begin{split}
 A_t&=\int_{\cC_t}\boep^+_{d-2},
 \qquad
 V_t=\int_{\S_t}dV,
 \qquad dV=u\i\boep,
 \\
 V_\z(t)&=\int_{\S_t}\z\i\boep
       =-\int_{\S_t}u_a\z^a\,dV.
 \end{split}
 \label{eq:VI_volumes}
\eeq
On a generic slice, the weight in $V_\z(t)$ is the
lapse $-u_a\z^a$. It equals $|\z|$ when $\z^a$ is
normal to the slice, as on $\S_0$.

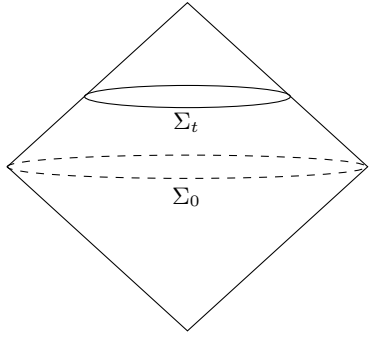
\begin{figure}[htbp]
 \centering
 \begin{tikzpicture}[x=0.75pt,y=0.75pt,yscale=-.6,xscale=.6]
  \draw (320.75,12) -- (471.5,149) --
        (320.75,286) -- (170,149) -- cycle;
  \draw[dashed]
   (170,149)
   .. controls (170,143.63) and (237.21,139.28)
   .. (320.13,139.28)
   .. controls (403.04,139.28) and (470.26,143.63)
   .. (470.26,149)
   .. controls (470.26,154.37) and (403.04,158.72)
   .. (320.13,158.72)
   .. controls (237.21,158.72) and (170,154.37)
   .. (170,149) -- cycle;
  \draw
   (234.65,90.37)
   .. controls (234.65,85.26) and (273.2,81.12)
   .. (320.75,81.12)
   .. controls (368.3,81.12) and (406.85,85.26)
   .. (406.85,90.37)
   .. controls (406.85,95.49) and (368.3,99.63)
   .. (320.75,99.63)
   .. controls (273.2,99.63) and (234.65,95.49)
   .. (234.65,90.37) -- cycle;
  \node at (320,174) {$\S_0$};
  \node at (320,112) {$\S_t$};
 \end{tikzpicture}
 \caption{The maximal slice $\S_0$ of a spherical
 causal diamond, shown dashed, and a later spacelike
 slice $\S_t$, whose boundary lies on $\cN^+$.}
 \label{causal_diamond}
\end{figure}

There are two relevant generators. On $\cN^+$, we choose the affine parameter $v=R-r$, increasing along $\cN^+$, so that
\beq
 \ell^a\partial_a
 =
 \frac{1}{1-r^2/L^2}\partial_t-\partial_r,
 \qquad
 \ell^a\nabla_a v=1.
 \label{eq:VI_affine_generator}
\eeq
Restricting the CKV to this boundary gives
\beq
 \z^a\big|_{\cN^+}
 =
 \kappa v\left(1-\frac{v}{R}\right)\ell^a,
 \qquad
 \alpha\big|_{\cN^+}=-\frac{\kappa v}{R}.
 \label{eq:VI_CKV_on_boundary}
\eeq
The boost generator used in
\cref{eq:WZ_charge_dyn_entropy}, with the matching
normalisation $\kappa_\beta=\kappa$, is instead
\beq
 \xi^a\big|_{\cN^+}=\kappa v\ell^a.
 \label{eq:VI_boost_generator}
\eeq
For the fixed affine boundary structure, a generator
$f\ell^a$ must satisfy $\partial_v^2f=0$
\cite{Chandrasekaran:2018aop,Chandrasekaran:2019ewn}.
Thus $\xi^a$ is an allowed boundary symmetry, whereas
the CKV generally does not preserve this structure.
The generators have the same boost normalisation at
the edge, but differ on later cuts. We use $\z^a$ for
the conformal comparison law and $\xi^a$ for the WZ
entropy and physical process law.

\subsection{First law}

\subsubsection{Comparison first law}

Consider a metric perturbation $h_{ab}=\d g_{ab}$, with $h=g^{ab}h_{ab}$ and
$h^{ab}=g^{ac}g^{bd}h_{cd}$. The generators and the embeddings of the slices and cuts are held fixed under the perturbation. The perturbed geometry need not admit the background CKV.
Substituting $\cL_\z g_{ab}=2\alpha g_{ab}$ into the
Einstein symplectic current gives
\cite{Burnett:1990tww,Jacobson:2018ahi}
\begin{subequations}
 \label{symp_current_t}
 \beq
 \bom(g;h,\cL_\z g)
 =
 \frac{d-2}{16\pi}\epsilon_a
 \left[
 \alpha\left(\nabla^a h-\nabla_bh^{ab}\right)
 +
 \left(h^{ab}-hg^{ab}\right)\nabla_b\alpha
 \right].
 \label{eq:VI_CKV_current}
 \eeq
 Here $\epsilon_a$ is the spacetime volume form with its first index displayed,
so that $V^a\epsilon_a=V\i\boep$ for any vector $V^a$.
On $\S_t$ it restricts to
$\epsilon_a|_{\S_t}=-u_a\,dV$.
Let $D_a\alpha=(\delta_a{}^b+u_a u^b)\nabla_b\alpha$.
Using
$\epsilon_a|_{\S_t}=-u_a\,dV$ and
$\d(dV)=\frac12(g^{ab}+u^a u^b)h_{ab}\,dV$,
its pullback becomes
 \beq
 \bom(g;h,\cL_\z g)\big|_{\S_t} = \frac{d-2}{8\pi}\dot{\alpha}\,\d(dV)
 -(d-2)\alpha\,\bth(g;h)\big|_{\S_t}
 - \frac{d-2}{16\pi} u_a h^{ab}D_b\alpha\,dV.
 \label{eq:VI_CKV_current_pullback}
 \eeq
\end{subequations}
The last term accounts for the spatial dependence of $\alpha$. It vanishes, for example, in a gauge where the mixed normal-tangential components of $h_{ab}$ vanish on the slice. The second and third terms on the right hand side vanish on $\S_0$, where $\alpha=D_a\alpha=0$. The symplectic form is given by,
\beq
 \Omega_{\S_t}(g;h,\cL_\z g) = \int_{\S_t}\bom(g;h,\cL_\z g).
 \label{eq:VI_symplectic_pairing}
\eeq
For the constant-$t$ slices,
\cref{eq:VI_alpha_dot,eq:VI_CKV_current_pullback}
give
\beq
 \Omega_{\S_t}(g;h,\cL_\z g) =
 -\frac{\kappa K}{8\pi} \cosh\frac{t}{L}\,\d V_t
 -(d-2)\int_{\S_t}\alpha\,\bth(g;h)
- \frac{d-2}{16\pi} \int_{\S_t}u_a h^{ab}D_b\alpha\,dV.
 \label{eq:VI_gravitational_pairing}
\eeq
 Here, $K$ is the trace of the extrinsic curvature of the edge, not that of $\cC_t$.
Let $\tilde m$ denote perturbative matter other than the cosmological constant, with vanishing background stress tensor, and define
\beq
 \d H_\eta^{\tilde m}[\S_t]
 =
 \int_{\S_t}
 \d(T_a{}^b)^{\tilde m}\eta^a u_b\,dV,
 \qquad
 V_\eta(t)=\int_{\S_t}\eta\i\boep,
 \label{eq:VI_matter_variation}
\eeq
for a fixed vector field $\eta^a$. Integrating the fundamental CPS identity with the linearised Einstein constraints imposed, gives
\beq
 \int_{\cC_t} \left[\d\boq_\z-\z\i\bth(g;h)\right] = 
 \Omega_{\S_t}(g;h,\cL_\z g) +\d H_\z^{\tilde m}[\S_t] + \frac{V_\z(t)}{8\pi}\,\d\Lambda.
 \label{fund_identity}
\eeq
For fixed $\Lambda$ and vacuum perturbations, only the symplectic pairing remains on the right. This is the comparison identity on $\S_t$. Its interpretation as the variation of an integrated Hamiltonian requires the integrability condition discussed in \cref{eq:global_charge_nonintegrability}.
At the edge, $\z^a=0$ and $\int_{\cB}\d\boq_\z=-\kappa\d A_0/(8\pi)$. Consequently, \cref{fund_identity} reduces to
\beq
 \d H_\z^{\tilde m}[\S_0] + \frac{V_\z(0)}{8\pi}\,\d\Lambda
 = \frac{\kappa K}{8\pi}\,\d V_0 - \frac{\kappa}{8\pi}\,\d A_0,
 \label{eq:VI_maximal_comparison}
\eeq
in agreement with \cref{eq:cd_equilibrium_firstlaw}. Away from the edge, the surface integral in \cref{fund_identity} includes both the variation of the Noether charge and the symplectic potential term. It should therefore be retained in this form.

For comparison with the entropy defined in \cref{eq:dyn_entropy_cd}, use the admissible generator $\xi^a$ and the WZ identity
\beq
 \int_{\cC_t}
 \left[\d\boq_\xi-\xi\i\bth(g;h)\right] =
 -\frac{\kappa}{2\pi}\d\cS[\cC_t] - \int_{\cC_t}\xi\i\bTh(g;h).
 \label{eq:VI_WZ_comparison_boundary}
\eeq
Integrating the CPS identity yields
\beq
 -\frac{\kappa}{2\pi}\d\cS[\cC_t] = 
 \Omega_{\S_t}(g;h,\cL_\xi g) +\d H_\xi^{\tilde m}[\S_t]
 +\frac{V_\xi(t)}{8\pi}\,\d\Lambda + \int_{\cC_t}\xi\i\bTh(g;h).
 \label{eq:VI_WZ_comparison}
\eeq
A smooth extension of $\xi^a$ into the interior, compatible with the fixed boundary structure, is understood. The complete right-hand side is fixed by the boundary data. Since $\xi^a$ is generally not the CKV, its symplectic pairing is not given by \cref{eq:VI_gravitational_pairing}.
At $\cB$, the boundary-potential term vanishes and $\cS[\cB]=A_0/4$, so the entropy term agrees with that of the equilibrium comparison law.

\subsubsection{Physical process first law and entropy balance}

For evolution along the future boundary, we keep $\Lambda$ fixed and use $\xi^a=\kappa v\ell^a$ with the orientation $\boep^+_{d-1}=\ed v\wedge\boep^+_{d-2}$. For two cuts with $v_2>v_1$, the result of
\cref{sec:physical_process_1st_law} gives
\beq
 -\frac{\kappa}{2\pi}\Delta\d\cS
 =
 \Delta\d\cE
 +
 \d\sf_\xi[\cN^+_{12}],
 \label{eq:VI_physical_process}
\eeq
where
\beq
 \Delta\d\cE
 =
 -\int_{\cN^+_{12}}
 \d T_{ab}\xi^a\ell^b\,\boep^+_{d-1}.
 \label{eq:VI_signed_matter_flux}
\eeq
Thus, the signed matter term is negative for positive outgoing null energy in this orientation. The corresponding entropy criterion
(\cref{cor:linearised_second_law}) is
\beq
 \Delta\d\cE+\d\sf_\xi[\cN^+_{12}]\ge0
 \quad\Longleftrightarrow\quad
 \Delta\d\cS\le0.
 \label{eq:VI_conditional_second_law}
\eeq
In vacuum, this condition reduces to $\d\sf_\xi[\cN^+_{12}]\ge0$. Positivity of this variation is an additional condition; it does not
follow merely from positivity of the background flux.

The unperturbed spherical diamond provides an explicit check. With $r=R-v$, its null geometry is
\beq
 A(v)=\tom_{d-2}(R-v)^{d-2},
 \qquad
 \vartheta=-\frac{d-2}{R-v},
 \qquad
 \sigma_{AB}=0.
 \label{eq:VI_spherical_null_geometry}
\eeq
The dynamical entropy therefore is
\beq
 \cS(v)
 =
 \frac{\tom_{d-2}}{4}
 (R-v)^{d-3}\big[R+(d-3)v\big],
 \label{eq:VI_spherical_entropy}
\eeq
and hence
\beq
 \frac{d\cS}{dv}
 =
 -\frac{\tom_{d-2}}{4}
 (d-2)(d-3)v(R-v)^{d-4}.
 \label{eq:VI_background_entropy_derivative}
\eeq
For $d>3$, this entropy decreases from the edge towards the corner $p^+$ and vanishes there, as shown in \cref{cor:global_decrease}. For $d=3$, it is constant. In four dimensions, $\cS(v)=\pi(R^2-v^2)$. The associated vacuum flux is
\beq
 \frac{d}{dv}\sf_\xi[\cN^+_{0v}]
 =
 \frac{\kappa v}{8\pi}
 \frac{d-3}{d-2}
 \int_{\cC(v)}\vartheta^2\boep^+_{d-2}
 =
 -\frac{\kappa}{2\pi}\frac{d\cS}{dv}
 \ge0.
 \label{eq:VI_background_flux}
\eeq
This exact background relation is consistent with the sign of \cref{eq:VI_physical_process}. The linearised law concerns the variation of this balance, and its monotonicity requires the condition in \cref{eq:VI_conditional_second_law}.

\subsection{Smarr relation}

We first evaluate the Noether identity for the CKV on the maximally symmetric background. For Einstein gravity with cosmological constant,
\beq
 \boj_\z
 =
 -\frac{d-1}{8\pi}\epsilon_a\nabla^a\alpha
 -
 \frac{\Lambda}{4\pi(d-2)}\,\z\i\boep,
 \qquad
 \boj_\z=\ed\boq_\z.
 \label{eq:VI_CKV_Noether_current}
\eeq
Its integral over $\S_t$ is
\beq
 \int_{\S_t}\boj_\z
 =
 \frac{d-1}{8\pi}
 \int_{\S_t}\dot{\alpha}\,dV
 -
 \frac{\Lambda}{4\pi(d-2)}V_\z(t).
 \label{eq:VI_bulk_Noether_integral}
\eeq
On a round cut of $\cN^+$, the Noether charge is given by
\beq
 \int_{\cC_t}\boq_\z
 =
 -\frac{\kappa+\alpha_t}{8\pi}A_t,
 \qquad
 \alpha_t=\alpha|_{\cC_t}
 =-\frac{\kappa v_t}{R},
 \qquad
 v_t=R-r_t.
 \label{eq:VI_CKV_cut_charge}
\eeq
The factor $\kappa+\alpha_t$ follows by contracting the antisymmetric derivative of $\z^a$ with the cut binormal, using $\nabla_a (\zeta^2) =-2\kappa\zeta_a$ and $\nabla_{(a}\zeta_{b)}=\alpha g_{ab}$. The conformal surface gravity $\kappa$ is distinct from the inaffinity:
$\zeta^a\nabla_a\zeta^b=(\kappa+2\alpha)\zeta^b$ \cite{Jacobson:2018ahi}.
Equating \cref{eq:VI_bulk_Noether_integral} and
\cref{eq:VI_CKV_cut_charge} gives
\beq
 (d-2)(\kappa+\alpha_t)A_t
 =
 2\Lambda V_\z(t)
 -
 (d-1)(d-2)
 \int_{\S_t}\dot{\alpha}\,dV.
 \label{eq:VI_Smarr_general_slice}
\eeq
For the constant-$t$ slices this becomes
\beq
 (d-2)\kappa\frac{r_t}{R}A_t
 =
 2\Lambda V_\z(t)
 +
 (d-1)\kappa K
 \cosh\frac{t}{L}\,V_t.
 \label{eq:VI_Smarr_static_slices}
\eeq
At $t=0$, $r_t=R$ and $\alpha_t=0$, recovering \cref{eq:cd_smarr}. This reduction uses the vanishing of $\z^a$ at $\cB$, together with
$\alpha|_{\S_0}=0$.
To relate this identity to the dynamical entropy, let $\bob_{\cN}$ denote the boundary form entering $\bTh=\underline{\bth}-\d\bob_{\cN}$. The improved charge of the admissible generator is
\beq
 \bos{\sq}_\xi
 =
 \boq_\xi-\xi\i\bob_{\cN},
 \qquad
 \int_{\cC_t}\bos{\sq}_\xi
 =
 -\frac{\kappa}{2\pi}\cS[\cC_t].
 \label{eq:VI_improved_charge}
\eeq
The corresponding current identity on the vacuum null boundary is
\beq
 \ed\bos{\sq}_\xi
 =
 \underline{\boj_\xi}
 -
 \ed(\xi\i\bob_{\cN}).
 \label{improved_smarr_identity}
\eeq
Thus, improving the charge also changes its current.
Keeping both the change of generator and the boundary correction gives
\beq
 \int_{\cC_t}\bos{\sq}_\xi = \int_{\S_t}\boj_\z +
 \int_{\cC_t} \left[ \boq_\xi-\boq_\z-\xi\i\bob_{\cN} \right].
 \label{eq:VI_Noether_WZ_relation}
\eeq
Consequently, the Smarr identity expressed through the dynamical entropy is
\beq
 4(d-2)\kappa\,\cS[\cC_t] =
 2\Lambda V_\z(t) -(d-1)(d-2)\int_{\S_t}\dot{\alpha}\,dV
- 8\pi(d-2) \int_{\cC_t} \left[ \boq_\xi-\boq_\z-\xi\i\bob_{\cN} \right].
 \label{eq:VI_dynamical_Smarr}
\eeq
For spherically symmetric cuts, the last integral is explicitly
\beq
 \int_{\cC_t}
 \left[ \boq_\xi-\boq_\z-\xi\i\bob_{\cN} \right]
 =
 -\frac{\kappa A_t}{8\pi}
 \left[
 \frac{v_t}{R}
 +
 \frac{(d-2)v_t}{R-v_t}
 \right].
 \label{eq:VI_Smarr_boundary_correction}
\eeq
It vanishes at the edge, where the two generators have the same boost charge and $\xi^a=0$. Using $\cS[\cB]=A_0/4$, the dynamical entropy
expression then reduces to the equilibrium Smarr formula derived in \cite{Jacobson:2018ahi}.








\section{Einstein-scalar gravity}\label{sec:EH-scalar_gravity}
Now, we shift our focus to Einstein-scalar (ES) gravity. In this section, we will explore the symmetry of a null surface in the ES gravity, the corresponding charges and fluxes on it, and, finally, apply it to causal diamonds.
As in the rest of the paper, we consider a $4$-dimensional spacetime, with a Lagrangian $4$-form given by \cite{Faulkner:2010fh, Anabalon:2023efw}
\beq\label{eq:lagrangian_ES_gravity}
  \sfL_{abcd} = \dfrac{1}{16\pi} \epsilon_{abcd}\left( \cR - (\nabla \P)^2 - 2 U (\P) \right)  ,
\eeq
where $\P$ is the scalar field coupled to gravity.

\subsection{Structure of a null surface and the configuration space} \label{sec:ES-null-configuration}
Let $\cN$ be a fixed smooth hypersurface in the four-dimensional spacetime $\cM$. The dynamical fields of the theory are the pair
$\p=(g_{ab},\P)$, with the scalar potential $U(\P)$ held fixed. Since the scalar in \cref{eq:lagrangian_ES_gravity} is minimally coupled, the kinematical definition of the null surface boundary structure is the same as in \cref{sec:null_surface_review}. Following \cite{Chandrasekaran:2018aop}, we fix the equivalence class
\begin{equation}
 \fs=[\ell^a,\k,\hat{\ell}_a],
 \label{eq:ES-boundary-structure}
\end{equation}
where $\ell^a$ is a nowhere vanishing vector tangent to the generators of $\cN$, $\hat{\ell}_a$ is a nowhere vanishing normal
covector, and $\hat{\ell}_a\ell^a\nulleq0$. Representatives
are related by
\begin{equation}
 (\ell^a,\k,\hat{\ell}_a)
 \sim
 \left(
 e^\gamma\ell^a,\,
 e^\gamma(\k+\cL_\ell\gamma),\,
 e^\gamma\hat{\ell}_a
 \right),
 \label{eq:ES-boundary-rescaling}
\end{equation}
for a smooth function $\gamma$ on $\cN$.
An admissible metric realises this structure
through
\begin{equation}
 g_{ab}\ell^b\nulleq\hat{\ell}_a,
 \qquad
 \ell^a\nabla_a\ell^b\nulleq\k\ell^b.
 \label{eq:ES-metric-compatibility}
\end{equation}
The hypersurface and the chosen representative of $\fs$ are held fixed under field-space variations. The additional intrinsic field is the scalar restriction
\begin{equation}
 \P_{\cN}:=\P|_{\cN}.
\end{equation}
Its derivatives along the generators and along a cross-section are, respectively,
\begin{equation}
 \cL_\ell\P_{\cN},
 \qquad
 D_A\P_{\cN}=\pa_A\P_{\cN},
\end{equation}
where $D_A$ is the Levi-Civita derivative of $q_{AB}$. The scalar has zero weight under \cref{eq:ES-boundary-rescaling}, whereas $\cL_\ell\P$ transforms as
\begin{equation}
 \P\longmapsto\P,
 \qquad
 \cL_\ell\P\longmapsto
 e^\gamma\cL_\ell\P.
 \label{eq:ES-scalar-rescaling}
\end{equation}
The transverse derivative $n^a\nabla_a\P$, for an auxiliary null vector satisfying $n^a\ell_a=-1$, depends on the bulk field and is not determined by $\P_{\cN}$ alone.
The unrestricted configuration space associated with the fixed geometric boundary structure is 
\begin{equation}
 \cF_{\fs}^{\rm ES} := \left\{
 (g_{ab},\P):
 \begin{aligned}
 &g_{ab}\ \text{is smooth and Lorentzian},\\
 &\P\in C^\infty(\cM),\\
 &g_{ab}\ell^b\nulleq\hat{\ell}_a,\quad
   \ell^a\nabla_a\ell^b\nulleq\k\ell^b
 \end{aligned}
 \right\}.
 \label{eq:ES-configuration-space}
\end{equation}
Here $q_{AB}$, $\vartheta$, $\sigma_{AB}$, and $\P_{\cN}$ are allowed to vary. In particular, fixing $\fs$ does not fix the
scalar profile or require $\cL_\ell\P=0$. Off shell, and in the absence
of additional scalar boundary conditions, $\cF_{\fs}^{\rm ES}$ is the product of the metric configuration space of \cref{sec:null_surface_review} with the space of smooth scalar fields.
A tangent vector to $\cF_{\fs}^{\rm ES}$ is a pair
\begin{equation}
 \delta\p=(h_{ab},\delta\P),
 \qquad h_{ab}=\delta g_{ab}.
\end{equation}
The metric perturbations satisfy the conditions given in \cref{perturbation_conditions},
\begin{equation}
 h_{ab}\ell^b\nulleq0,
 \qquad
 \nabla_c(h_{ab}\ell^a\ell^b)\nulleq0.
 \label{eq:ES-tangent-conditions}
\end{equation}
No additional restriction on $\delta\P$ follows from preservation of the geometric boundary structure. Any restriction on the scalar perturbation must instead come from a separately specified scalar boundary condition.
The covariant phase space is the subset
\begin{equation}
 \cP_{\fs}^{\rm ES}
 :=
 \left\{
 (g_{ab},\P)\in\cF_{\fs}^{\rm ES}: E_{ab}=0,\quad E_\P=0 \right\},
 \label{eq:ES-solution-space}
\end{equation}
where, with the normalisation of
\cref{eq:lagrangian_ES_gravity},
\begin{align}
 E_{ab}  &:=  G_{ab}-\nabla_a\P\nabla_b\P + \frac12g_{ab}(\nabla\P)^2
 +g_{ab}U(\P),\\
 E_\P &:=  \nabla^a\nabla_a\P-U'(\P).
 \label{eq:ES-equations-configuration}
\end{align}
A tangent vector to $\cP_{\fs}^{\rm ES}$ satisfies \cref{eq:ES-tangent-conditions}, together with the coupled linearised equations
\begin{equation}
 \d E_{ab}=0, \qquad \d E_\P=0.
\end{equation}
Thus the metric and scalar perturbations are independent kinematical variables but are coupled on-shell.
The scalar also modifies the null constraint equations through its backreaction. In particular, contraction of the Einstein equation with $\ell^a\ell^b$ yields
\begin{equation}
 \cR_{ab}\ell^a\ell^b
 =(\cL_\ell\P)^2.
 \label{eq:ES-null-einstein-constraint}
\end{equation}
Consequently, the twist-free Raychaudhuri equation takes the form
\begin{equation}
 \cL_\ell\vartheta
 =\k\vartheta-\frac12\vartheta^2
  -\sigma_{AB}\sigma^{AB}
  -(\cL_\ell\P)^2.
 \label{eq:ES-null-focusing}
\end{equation}
The potential does not contribute directly to this contraction because $\ell^a$ is null. The boundary geometry therefore has the same
kinematical definitions as in general relativity, while its evolution is constrained by the scalar field. In particular, an
on-shell boundary with $\vartheta=\sigma_{AB}=0$ throughout a
portion of $\cN$ must also satisfy $\cL_\ell\P=0$ there. For completeness, we discuss the symmetry of a null surface in ES gravity in \cref{sec:ES-null-symmetry}.

\subsection{Charges and fluxes in Einstein-scalar gravity}
\label{sec:ES-charges-fluxes}

We consider the configuration space \cref{eq:ES-configuration-space}, allowing the scalar boundary data to vary. The generators are held fixed under variations and need not preserve a particular
scalar profile. The additional restrictions of \cref{sec:ES-null-symmetry} apply when one chooses the corresponding scalar stabiliser.  For the ES action as given in
\cref{eq:lagrangian_ES_gravity},
\beq
  \sfL_{abcd} = \dfrac{1}{16\pi} \epsilon_{abcd}\left( \cR - (\nabla \P)^2 - 2 U (\P) \right)  , \nonumber
\eeq
the physical scalar stress tensor is
\beq
 T^{\P}_{ab}
 =
 \frac{1}{8\pi}
 \left[
 \nabla_a\P\nabla_b\P
 -\frac12g_{ab}(\nabla\P)^2
 -g_{ab}U(\P)
 \right],
 \qquad
 G_{ab}=8\pi T^{\P}_{ab}.
 \label{eq:ES-stress-normalisation}
\eeq
This convention agrees with the matter stress tensor
used in \cref{sec:physical_process_1st_law}.
In particular,
$8\pi T^{\P}_{ab}\ell^a\ell^b=(\cL_\ell\P)^2$,
as in \cref{eq:ES-null-einstein-constraint}.
Writing $h_{ab}=\d g_{ab}$ and $h=g^{ab}h_{ab}$,
variation of the Lagrangian gives
\beq
 \d\bos{\sfL}
 =
 -\frac{1}{16\pi}E^{ab}h_{ab}\boep
 +\frac{1}{8\pi}E_\P\,\d\P\,\boep
 +\ed\bth_{\rm ES}(\p;\d\p),
 \label{eq:ES-Lagrangian-variation}
\eeq
where $E_{ab}$ and $E_\P$ are the equations defined in \cref{eq:ES-equations-configuration}. The symplectic potential is
\beq
 \bth_{\rm ES}(\p;\d\p)
 =
 \frac{1}{16\pi}\epsilon_a
 \left[
 \nabla_bh^{ab}-\nabla^a h
 -2\nabla^a\P\,\d\P
 \right].
 \label{symp_potential_scalar-gravity}
\eeq
Here $\epsilon_a$ denotes the spacetime volume form with its first index displayed, as in \cref{sec:CD_MSS}, and the indices on $h_{ab}$ are raised with the background metric. The relative signs follow directly from varying \cref{eq:lagrangian_ES_gravity} in the Iyer-Wald convention \cite{Iyer:1994ys}.
The symplectic current contains both contributions,
\beq
 \bom_{\rm ES}=\bom_{\rm g}+\bom_\P,
 \label{eq:ES-symplectic-current}
\eeq
where
\beq
 \bom_\P(\p;\d_1\p,\d_2\p)
 =
 -\frac{1}{8\pi}
 \left[
 \d_1(\epsilon_a\nabla^a\P)\,\d_2\P
 -
 \d_2(\epsilon_a\nabla^a\P)\,\d_1\P
 \right].
 \label{eq:ES-scalar-symplectic-current}
\eeq
The variations in parentheses act on the metric,
the volume form, and the scalar field.

The Noether current satisfies
\beq
 \boj_\xi
 =
 \ed\boq_\xi+\xi^a\boc_a,
 \qquad
 \boc_a
 =
 \left(
 \frac{1}{8\pi}G_a{}^b-T^{\P}_a{}^b
 \right)\epsilon_b
 =
 \frac{1}{8\pi}E_a{}^b\epsilon_b.
 \label{eq:ES-Noether-constraint}
\eeq
Thus, the full constraint vanishes on solutions of
the coupled theory. The Noether charge two-form is
\beq
 (Q_\xi)_{ab}
 =
 -\frac{1}{16\pi}\epsilon_{abcd}\nabla^c\xi^d.
 \label{eq:ES-Noether-charge}
\eeq
Minimal coupling leaves this expression unchanged. The scalar contributes to the symplectic current through the symplectic potential,
\beq
 \bom_{\rm ES}(\p;\d\p,\cL_\xi\p) \approx \ed\left(\d\boq_\xi-\xi\i\bth_{\rm ES}(\p;\d\p)\right),
 \label{eq:ES-surface-charge-variation}
\eeq
whose integrability is subject to the condition discussed in \cref{sec:cps_review}. Using the same geometric boundary term as in the preceding sections, written with this orientation,
\beq
 \bob_{\cN}\equiv\bos{\vr} = -\frac{\vartheta}{8\pi}\boep_3,
 \qquad
 \bTh_{\rm ES} = \underline{\bth_{\rm ES}}-\d\bob_{\cN},
 \label{eq:ES-boundary-counterterm}
\eeq
we obtain
\beq
 \bTh_{\rm ES}(\p;\d\p) = -\frac{\boep_3}{16\pi}
 \left[ \left(\sigma^{AB}-\frac12\vartheta q^{AB}\right)
 h_{AB} + 2(\cL_\ell\P)\,\d\P \right].
 \label{eq:ES-boundary-potential}
\eeq
Here, $h_{AB}$ is the pullback of $h_{ab}$ to the cuts. The gravitational part is the potential of \cite{Chandrasekaran:2018aop}; the scalar part follows from \cref{symp_potential_scalar-gravity}. On a stationary on-shell portion of the boundary, $\vartheta=\sigma_{AB}=0$ implies $\cL_\ell\P=0$ by \cref{eq:ES-null-focusing}. Consequently, $\bTh_{\rm ES}$ vanishes there for every allowed perturbation, as required by the WZ stationarity condition.

For a field-independent generator, the WZ variation becomes
\begin{align}
 \d\bos{\sq}_\xi
 &=
 \d\boq_\xi-\xi\i\underline{\bth_{\rm ES}}
 +\xi\i\bTh_{\rm ES}
 \nonumber\\
 &=
 \d\left(\boq_\xi-\xi\i\bob_{\cN}\right).
 \label{eq:ES-WZ-variation}
\end{align}
With the charge normalisation used in
\cref{sec:dyn_entropy_cd}, this gives
\beq
 \bos{\sq}_\xi
 =
 \boq_\xi-\xi\i\bob_{\cN}
 =
 \frac{\boep_2}{8\pi}
 \left(
 n_b\ell^a\nabla_a\xi^b
 -\vartheta\xi^a n_a
 \right).
 \label{eq:ES-WZ-charge}
\eeq
The charge therefore retains the geometric form
of \cref{WZ_charge_2}. Its scalar dependence through
the metric is determined by the coupled equations.

The flux is obtained by evaluating the boundary
potential on the symmetry variation:
\beq
 \begin{split}
 \bos{\sf}^{\rm ES}_\xi
 &=
 \bTh_{\rm ES}(\p;\cL_\xi\p)
 \\
 &=
 \bos{\sf}^{\rm g}_\xi
 -
 \frac{\boep_3}{8\pi}
 (\cL_\ell\P)(\cL_\xi\P).
 \end{split}
 \label{eq:ES-total-flux}
\eeq
In particular, for $\xi^a=f\ell^a$,
\beq
 \bos{\sf}^{\rm ES}_\xi
 =
 -\frac{f}{8\pi}\boep_3
 \left[
 \sigma_{AB}\sigma^{AB}
 -\frac12\vartheta^2
 +(\cL_\ell\P)^2
 \right].
 \label{eq:ES-longitudinal-flux}
\eeq
On shell, the corresponding balance law is
\beq
 \ed\bos{\sq}_\xi=\bos{\sf}^{\rm ES}_\xi,
 \qquad
 \sq_\xi[\cC_2]-\sq_\xi[\cC_1]
 =
 \sf_\xi^{\rm ES}[\cN_{12}].
 \label{eq:ES-charge-balance}
\eeq

The absence of an explicit scalar term in the
localised charge is consistent with its presence
in the flux. The scalar terms in the two potentials
cancel in \cref{eq:ES-WZ-variation}, while the
boundary potential itself retains the scalar
contribution. The evolution of the geometric
charge includes the scalar backreaction through
\cref{eq:ES-null-focusing}. If the generator belongs
to a scalar stabiliser, $\cL_\xi\P=0$, its explicit
scalar flux in \cref{eq:ES-total-flux} vanishes.

\subsection{Causal diamonds in Einstein-scalar gravity}
\label{sec:ES-causal-diamonds}

We consider a causal diamond in a smooth solution
$(g_{ab},\P)$ of ES gravity. Its geometric construction
is the same as in \cref{sec:dyn_entropy_cd}: the null
boundaries $\cN^\pm$ intersect at the spacelike
codimension-two surface $\cB$, and we retain the same
smoothness assumptions, affine parametrisation, and
normalisation of their null generators. Minimal
coupling introduces no additional geometric boundary
structure, although the scalar stress tensor affects
the boundary geometry through the coupled equations.

The scalar restrictions on $\cN^\pm$ have matching
values at $\cB$, inherited from the smooth bulk field.
We allow these data to fluctuate and transform as
$\d_\xi\P=\cL_\xi\P$. The geometric boundary symmetry
algebra and the matching conditions for its generators
therefore remain those of \cref{sec:dyn_entropy_cd}
\cite{Chandrasekaran:2019ewn}.

Preservation of the edge removes the affine
supertranslation, so that $\varsigma=0$. On $\cN^+$,
choose the constant boost
\beq
 \xi^a=\kappa_\beta v\ell^a,
 \qquad
 \kappa_\beta=-\beta>0,
 \qquad
 v|_{\cB}=0.
 \label{eq:ES-CD-boost}
\eeq
This choice does not impose scalar-profile
preservation. Within a scalar stabiliser it would
additionally require
\beq
 \cL_\xi\P
 =
 \kappa_\beta v\,\cL_\ell\P=0.
 \label{eq:ES-CD-stabiliser-condition}
\eeq
Thus, a nontrivial pure boost is excluded from that
stabiliser on a region where $\cL_\ell\P\ne0$.
The following entropy and flux relations use the
configuration space in which the scalar data may
vary.

Substituting \cref{eq:ES-CD-boost} into
\cref{eq:ES-WZ-charge} gives
\beq
 \sq_\xi[\cC]
 =
 -\frac{\kappa_\beta}{8\pi}
 \int_{\cC}(1-v\vartheta)\boep_2^+
 =
 -\frac{\kappa_\beta}{2\pi}\cS[\cC],
 \label{eq:ES-CD-entropy-charge}
\eeq
where
\beq
 \cS[\cC]
 =
 \frac14\int_{\cC}(1-v\vartheta)\boep_2^+.
 \label{eq:ES-CD-entropy}
\eeq
The dynamical entropy has the same functional form
as in \cref{eq:dyn_entropy_cd}, and
$\cS[\cB]=A(\cB)/4$.
For two cuts on the future boundary, the exact balance is
\beq
 -\frac{\kappa_\beta}{2\pi}
 \left(\cS[\cC_2]-\cS[\cC_1]\right)
 =
 \sf_\xi^{\rm ES}[\cN^+_{12}],
 \label{eq:ES-CD-exact-balance}
\eeq
with
\beq
 \sf_\xi^{\rm ES}[\cN^+_{12}]
 =
 -\frac{\kappa_\beta}{8\pi}
 \int_{\cN^+_{12}}v\,\boep_3^+
 \left[
 \sigma_{AB}\sigma^{AB}
 -\frac12\vartheta^2
 +(\cL_\ell\P)^2
 \right].
 \label{eq:ES-CD-total-flux}
\eeq
A direct check follows from
\cref{eq:ES-null-focusing}: for the uniform cuts
$\cC(v)$,
\beq
 \frac{d\cS}{dv}
 =
 \frac{v}{4}\int_{\cC(v)}\boep_2^+
 \left[
 \sigma_{AB}\sigma^{AB}
 +(\cL_\ell\P)^2
 -\frac12\vartheta^2
 \right].
 \label{eq:ES-CD-entropy-evolution}
\eeq
This reproduces both the scalar coefficient and the
sign in \cref{eq:ES-CD-exact-balance}.

For perturbations satisfying the coupled linearised
equations, with the boundary structure and generator
held fixed, variation gives the physical process law
\beq
 -\frac{\kappa_\beta}{2\pi}\Delta\d\cS
 =
 \d\sf_\xi^{\rm ES}[\cN^+_{12}].
 \label{eq:ES-CD-linear-balance}
\eeq
The scalar is already included in the total flux.
Equivalently, one may write
\begin{align}
 \sf_\xi^\P[\cN^+_{12}]
 &=
 -\int_{\cN^+_{12}}
 T^\P_{ab}\xi^a\ell^b\,\boep_3^+,
 \label{eq:ES-CD-scalar-flux}
 \\
 \Delta\d\cE_\P
 &\equiv
 \d\sf_\xi^\P[\cN^+_{12}],
 \nonumber\\
 -\frac{\kappa_\beta}{2\pi}\Delta\d\cS
 &=
 \Delta\d\cE_\P
 +
 \d\sf_\xi^{\rm g}[\cN^+_{12}].
 \label{eq:ES-CD-split-balance}
\end{align}
This is the split used in
\cref{sec:physical_process_1st_law}. When the
background scalar flux is nonzero, the variation
in $\Delta\d\cE_\P$ acts on the full integral,
including its volume form.

The conditional linearised entropy statement is
\beq
 \d\sf_\xi^{\rm ES}[\cN^+_{12}]\ge0
 \quad\Longleftrightarrow\quad
 \Delta\d\cS\le0.
 \label{eq:ES-CD-conditional-law}
\eeq
Equivalently, the condition is
$\Delta\d\cE_\P+\d\sf_\xi^{\rm g}\ge0$.
Although $T^\P_{ab}\ell^a\ell^b\ge0$, this does not
establish the required sign of the total flux or
its first variation. The expansion, shear, and
scalar terms must be combined with the signs in
\cref{eq:ES-CD-total-flux}. Indeed, since
$8\pi T^\P_{ab}\ell^a\ell^b=(\cL_\ell\P)^2$,
\cref{thm:exact_balance} and its corollaries as in \cref{sec:physical_process_1st_law} apply verbatim with
$T_{ab}=T^\P_{ab}$. The scalar raises the entropy along $\cN^+$, as shear does, and only the focusing term lowers it. In particular, \cref{cor:global_decrease} continues to hold for a regular corner.

For comparison with \cref{sec:CD_MSS}, the
spacelike comparison identity uses the full
Einstein-scalar symplectic current:
\beq
 -\frac{\kappa_\beta}{2\pi}\d\cS[\cC]
 =
 \int_{\S}\bom_{\rm ES}
       (\p;\d\p,\cL_\xi\p)
 +\int_{\cC}\xi\i\bTh_{\rm ES}(\p;\d\p).
 \label{eq:ES-CD-comparison}
\eeq
Here, $\xi^a$ is extended smoothly into the interior,
compatible with its boundary data, and the potential
$U$ is fixed. This identity follows from the coupled
linearised constraints and the WZ charge variation.
The maximally symmetric limit relevant to
\cref{sec:cd_mss} is obtained when
\beq
 \P=\P_0=\text{constant},
 \qquad
 U'(\P_0)=0,
 \qquad
 \Lambda_{\rm eff}=U(\P_0),
 \label{eq:ES-MSS-limit}
\eeq
so that $G_{ab}+\Lambda_{\rm eff}g_{ab}=0$.
The geometric comparison and the Smarr relations then
reduce to those of the preceding section and \cref{sec:cd_mss}, with the same distinction between the
CKV $\z^a$ and the boundary boost $\xi^a$.
Since $U$ is fixed,
$\d\Lambda_{\rm eff}=U'(\P_0)\d\P_0=0$ at first
order; varying a parameter of the potential would
require an extended comparison law.

More generally, on any background with
$\cL_\ell\P=0$, the scalar contribution to the
first-order null flux vanishes. For an ordinary
one-parameter family of scalar perturbations, its
leading possible contribution is quadratic:
\beq
 \left.\d^2\bos{\sf}^{\P}_\xi
 \right|_{\cL_\ell\P=0}
 =
 -\frac{\kappa_\beta v}{4\pi}
 \left(\cL_\ell\d\P\right)^2\boep_3^+.
 \label{eq:ES-scalar-second order-flux}
\eeq
It must be combined with the gravitational
contribution in the second order balance.

\section{Conclusion}
\label{sec:conclusion}

We have studied the entropy of causal diamonds through the
Wald-Zoupas charges and fluxes on their null boundaries
\cite{Wald:1999wa,Chandrasekaran:2018aop,Chandrasekaran:2019ewn}. For
the chosen boundary symplectic potential, to every cut of the future boundary is assigned an entropy consisting of the area term and an expansion correction. This functional has the same form as
the dynamical black hole entropy \cite{Hollands:2024vbe,Visser:2024pwz}. It reduces to $A(\cB)/4$ at the edge, in agreement with the equilibrium analysis of \cite{Jacobson:2018ahi}, and its evolution along the boundary is governed by a balance law.

In Einstein gravity, this balance is exact. The expansion correction
removes the first derivative of the area, so the entropy responds only
to the focusing of the generators. The Raychaudhuri equation then gives
the entropy difference between any two cuts, as the
integral of a signed flux density. The entropy therefore decreases, both
between two cuts and under every local deformation of a cut, if and only
if this flux is non-negative. In four or
higher dimensions, and the entropy
falls from $A(\cB)/4$ at the edge to zero at a regular corner $p^+$, irrespective of the matter content. The same balance explains why the entropy of a finite spherical diamond decreases while the perturbed entropy of a
Rindler or de Sitter horizon increases \cite{Zhao:2025zny}. These two
differ in the background expansion and the resulting gravitational flux,
not in the sign of the temperature.

At linear order, the covariant phase space identity gives the physical
process first law for any solution of general relativity with the charge-entropy
relation, and the same equivalence between non-negative flux and entropy
decrease holds. Unlike the stationary case
\cite{Hollands:2024vbe,Visser:2024pwz}, the boundary boost is generally
not a Killing field, so the gravitational flux contributes already at
first order in vacuum. At second order, the balance depends on the
second order metric perturbation and on the variation of the symplectic
current as a background dependent bilinear map. We showed with an
explicit Minkowski family that the latter term cannot be dropped, since
the boost is not Killing even in flat space. The modified canonical
energy of \cite{Hollands:2012sf,Hollands:2024vbe} is recovered only
when the boundary generator is Killing.

At the semiclassical level, the assumed improved quantum null energy
condition \cite{Ben-Dayan:2023inz,BenDayan:2026inec} leads to a balance
for the sum of the gravitational entropy and an expansion dependent
matter entropy. This parallels the generalised second law obtained from
the quantum null energy condition for stationary black holes
\cite{Wall:2011gsl,Hollands:2024vbe,Bousso:2015wca,Balakrishnan:2017bjg}.
The combined entropy decreases when the gravitational flux exceeds the
weighted remainder of the quantum inequality; when the inequality is
saturated, positivity of the gravitational flux alone suffices. The
result remains conditional on the applicability of the INEC to the
diamond boundary and on its perturbative implementation.

For spherical diamonds in maximally symmetric spacetimes, the boundary
boost and the diamond-preserving conformal Killing field differ away
from the edge. Keeping both, we obtained comparison and Smarr identities
on generic slices, which reduce to the equilibrium relations of Jacobson
and Visser \cite{Jacobson:2018ahi} on the maximal slice. In minimally
coupled Einstein-scalar gravity \cite{Faulkner:2010fh,Anabalon:2023efw},
the entropy keeps its geometric form, the scalar enters the symplectic
flux, and its energy is counted once in the balance law.

Several avenues are still to be explored. First, one can aim to describe the nonstationary diamonds and the classes of perturbations for which the needed sign of the flux is guaranteed by geometric or dynamical criteria. Second, it is important to understand how the entropy depends on the chosen boundary symplectic potential and boundary structure, an issue already apparent in the ambiguities of Noether charge entropy \cite{Iyer:1994ys}, in anomalies in null boundary charges \cite{Chandrasekaran:2020wwn, Odak:2022ndm, Odak:2023pga}, and in related entropy balance laws on null hypersurfaces \cite{Rignon-Bret:2023fjq,Shajiee:2026coz}. Generalising the construction to higher derivative theories and to non-minimal couplings would link it to the linearised and effective field theory versions of the second law \cite{Wall:2015raa, Hollands:2022fkn}. From the quantum perspective, providing a field theoretic account of the expansion dependent matter entropy, and relating the dynamical entropy on later cuts to entanglement entropy via entanglement equilibrium \cite{Jacobson:2015hqa} or holographic entanglement \cite{Ryu:2006bv,Hubeny:2007xt,Faulkner:2013ica}, would help determine when the conditional balance becomes a generalised second law for causal diamonds.


\begin{acknowledgments}
AC and PKD gratefully acknowledge financial support from the DST-INSPIRE Fellowship (IF200258, IF200253) of the Department of Science and Technology, Government of India. MM’s work is funded through DST-FIST Grants No. SR/FST/LS-I/2020/621 and No. SR/FST/PSI-225/2016.
\end{acknowledgments}

\appendix

\section{Thermodynamics of spherical causal diamonds in maximally symmetric spacetimes}
\label{sec:cd_mss}

In this appendix, we briefly review the geometry and
equilibrium thermodynamics of spherical causal diamonds
in Einstein gravity with a cosmological constant,
following
\cite{Jacobson:2015hqa,Jacobson:2018ahi,JacobsonVisser:2019negative}.
We use the coordinates and conformal Killing normalisation of \cite{Jacobson:2018ahi}, writing $A$, $K$, and $V$ for the edge area, its outward mean curvature, and the proper volume of the maximal slice.
Our $K$ is the quantity denoted by $k$ in that reference. We display $G$ and $\hbar$ explicitly; setting $G=\hbar=1$ recovers the units used in the main text.

For positive cosmological constant, the static coordinates of the maximally symmetric background are given by \cref{eq:cd_mss_static_metric}
\beq
 ds^2 = -\left(1-\frac{r^2}{L^2}\right)dt^2
 + \frac{dr^2}{\left(1-\frac{r^2}{L^2}\right)}
 + r^2d\tom_{d-2}^2,
 \qquad
 \Lambda=\frac{(d-1)(d-2)}{2L^2}. \nonumber
\eeq
The expressions for AdS follow by $L\rightarrow iL$, while the flat limit is $L\rightarrow\infty$. Let the diamond be the domain of dependence of the
ball $\S=\{t=0,\,0\le r\le R\}$, with $R<L$ in de Sitter space. This is the maximal slice, denoted by $\S_0$ in the main text, and its edge is
$\cB=\partial\S$. The corners lie at $r=0$, $t=\pm R_*$, where $R=L\tanh(R_*/L)$.

The diamond-preserving conformal Killing vector is (\cref{eq:cd_mss_ckv_static})
\beq
 \z
 =
 \frac{\kappa L^2}{R}
 \Bigg[  \left(
 1- \frac{\sqrt{1-R^2/L^2}}{\sqrt{1-r^2/L^2}} \cosh\frac{t}{L} \right)\partial_t
 - \frac{r}{L}
 \sqrt{ \left(1-\frac{R^2}{L^2}\right) \left(1-\frac{r^2}{L^2}\right)}
 \sinh\frac{t}{L}\,\partial_r \Bigg]. \nonumber
\eeq
It is future directed and timelike in the interior,
tangent to the null boundaries, and vanishes at the
edge and the corners. The conformal Killing equation is
\beq
 \cL_\z g_{ab}=2\alpha g_{ab},
 \qquad
 \alpha=\frac{1}{d}\nabla_a\z^a.
 \label{eq:cd_ckv}
\eeq
Following \cite{Jacobson:2018ahi}, surface gravity is
defined by the gradient of the squared norm on the
conformal Killing horizons:
\beq
 \nabla_a(\z^b\z_b)\big|_{\cN^\pm}
 =
 \mp2\kappa\z_a,
 \qquad
 \kappa>0.
 \label{eq:cd_mss_surface_gravity}
\eeq
The explicit vector in \cite{Jacobson:2018ahi} has
$\kappa=1$; the factor $\kappa$ in
\cref{eq:cd_mss_ckv_static} allows a constant rescaling.

Let $u^a$ be the future directed unit normal to the maximal slice. Here, for simplicity, we refer to the maximal slice $\S_0$ in the main text as $\S$. 
On this slice,
\beq
 u^a\partial_a
 =
 \left(1-\frac{r^2}{L^2}\right)^{-1/2}\partial_t,
 \qquad
 \alpha\big|_{\S}=0.
\eeq
Thus $\z^a$ is Killing on $\pa\S$,
although it does not vanish in the interior of
this slice. The edge has outward mean curvature
\beq
 K
 =
 \frac{d-2}{R}\sqrt{1-\frac{R^2}{L^2}}.
 \label{eq:cd_mss_edge_curvature}
\eeq
Here $K$ is the trace of the extrinsic curvature of
$\partial\S$ within $\S$. A useful consequence of
the conformal Killing geometry is
\beq
 \dot{\alpha}\big|_{\S}
 \equiv
 u^a\nabla_a\alpha\big|_{\S}
 =
 -\frac{\kappa}{R}\sqrt{1-\frac{R^2}{L^2}}
 =
 -\frac{\kappa K}{d-2},
 \label{alpha_dot}
\eeq
which is constant over $\S$
\cite[section~2]{Jacobson:2018ahi}.
The area and proper volume are
\beq
 \begin{split}
 A&=\int_{\partial\S}dA
   =\tom_{d-2}R^{d-2},
 \\
 V&=\int_{\S}dV
   =\tom_{d-2}\int_0^R
   \frac{r^{d-2}\,dr}{\sqrt{1-r^2/L^2}},
 \end{split}
 \label{eq:cd_area_volume}
\eeq
where $\tom_{d-2}$ is the area of the unit
$(d-2)$-sphere. The thermodynamic volume is the
conformal Killing weighted volume
\beq
 V_\z
 \equiv
 \int_{\S}|\z|\,dV,
 \qquad
 |\z|=\sqrt{-\z^a\z_a}.
 \label{eq:cd_thermodynamic_volume}
\eeq
Since $\z^a=|\z|u^a$ on $\S$, this also equals
$\int_{\S}\z\i\boldsymbol{\epsilon}$, with
$dV=u\i\boldsymbol{\epsilon}$ chosen positive.
Explicitly,
\beq
 V_\z
 =
 \frac{\kappa L^2}{R}
 \left[
 \frac{\tom_{d-2}R^{d-1}}{d-1}
 -
 \sqrt{1-\frac{R^2}{L^2}}\,V
 \right].
 \label{eq:cd_mss_thermodynamic_volume}
\eeq
The quantities $V$ and $V_\z$ are distinct; the
latter depends on the normalisation of $\z$
\cite[section~3.1]{Jacobson:2018ahi}.

\subsection*{Smarr relation and comparison first law}
\label{sec:cdthermodynamics}

The integrated on-shell Noether identity is
\beq
 \int_{\S}\boj_\z
 =
 \int_{\partial\S}\boq_\z.
 \label{smarr_identity}
\eeq
With the edge oriented outward from the ball, its two sides evaluate to
\cite[section~3.1]{Jacobson:2018ahi}
\begin{align}
 \int_{\partial\S}\boq_\z
 &=
 -\frac{\kappa A}{8\pi G},
 \label{eq:cd_noether_edge}
 \\
 \int_{\S}\boj_\z
 &=
 -\frac{(d-1)\kappa K V+2\Lambda V_\z}
 {(d-2)8\pi G}.
 \label{eq:cd_bulk_noether}
\end{align}
Consequently, the Smarr relation is
\beq
 (d-2)\kappa A
 =
 (d-1)\kappa K V+2\Lambda V_\z.
 \label{eq:cd_smarr}
\eeq
Using \cref{alpha_dot}, it can equivalently be written as
\beq
 (d-2)\kappa A
 =
 2\Lambda V_\z
 -
 (d-1)(d-2)\dot{\alpha}\,V,
 \label{eq:cd_smarr_alphadot}
\eeq
where $\dot{\alpha}$ is evaluated on $\S$. At the edge, any regular boundary correction proportional to $\z^a$ vanishes, so the Noether charge expression agrees there with the equilibrium limit of the entropy construction in the main text.

The comparison first law concerns linearised variations away from this maximally symmetric background, with $\z^a$ and the hypersurface $\S$ held fixed. The perturbations satisfy the linearised gravitational constraints; the perturbed geometry need not remain maximally symmetric \cite[section~3.2]{Jacobson:2018ahi}.
Following that reference, $\tilde m$ denotes matter other than the cosmological constant. Its conformal Hamiltonian variation is
\beq
 \d H_\z^{\tilde m} =
 \int_{\S} \d\left(T_a{}^b\right)^{\tilde m} \z^a u_b\,dV.
 \label{eq:cd_matter_hamiltonian}
\eeq
The reference stress tensor of this matter vanishes. The cosmological and gravitational contributions are
\begin{align}
 \d H_\z^\Lambda
 &=
 \frac{V_\z}{8\pi G}\,\d\Lambda,
 \label{eq:cd_lambda_hamiltonian}
 \\
 \d H_\z^{\rm g}
 &=
 -\frac{\kappa K}{8\pi G}\,\d V.
 \label{eq:cd_grav_hamiltonian}
\end{align}
The gravitational term is present because $\z^a$ is conformal Killing rather than Killing. All coefficients in these first-order expressions are evaluated on the reference geometry.
Combining the three contributions gives
\beq
 \d H_\z^{\tilde m} + \frac{V_\z}{8\pi G}\,\d\Lambda
 = \frac{\kappa K}{8\pi G}\,\d V - \frac{\kappa}{8\pi G}\,\d A.
 \label{eq:cd_equilibrium_firstlaw}
\eeq
This is the comparison first law given in \cite{Jacobson:2018ahi}. In terms of the total conformal Hamiltonian variation,
\beq
 \d H_\z
 \equiv \d H_\z^{\tilde m} + \d H_\z^\Lambda + \d H_\z^{\rm g}
 = -\frac{\kappa}{8\pi G}\,\d A.
 \label{eq:cd_total_hamiltonian_firstlaw}
\eeq
The thermodynamic identification is
\beq
 S_{\rm BH}=\frac{A}{4\hbar G},
 \qquad
 T=-\frac{\hbar\kappa}{2\pi},
 \label{eq:cd_equilibrium_TS}
\eeq
so that $\d H_\z=T\,\d S_{\rm BH}$
\cite{Jacobson:2018ahi,JacobsonVisser:2019negative}.
Keeping the matter Hamiltonian separate, the first
law instead reads
\beq
 \d H_\z^{\tilde m}
 =
 T\,\d S_{\rm BH}
 +
 \frac{\kappa K}{8\pi G}\,\d V
 -
 \frac{V_\z}{8\pi G}\,\d\Lambda.
 \label{eq:cd_thermodynamic_firstlaw}
\eeq
The volume-work term must therefore be retained
unless the variation fixes $V$. In particular,
\beq
 \left.\d H_\z^{\tilde m}\right|_{V,\Lambda}
 =
 -\frac{\kappa}{8\pi G}
 \left.\d A\right|_{V,\Lambda}.
 \label{eq:cd_fixed_volume_firstlaw}
\eeq
Positive matter conformal energy thus decreases
the edge area at fixed volume and cosmological
constant. This is the comparison-law basis for
the negative-temperature interpretation. It
compares nearby configurations and does not
determine the ordering of entropy on successive
future null cuts.

The cosmological term may also be expressed through
$p=-\Lambda/(8\pi G)$:
\beq
 \d H_\z^{\rm g}
 +
 \d H_\z^{\tilde m}
 =
 T\,\d S_{\rm BH}
 +
 V_\z\,\d p.
 \label{eq:cd_pressure_firstlaw}
\eeq
This is the enthalpy form discussed in
\cite[section~4.1]{Jacobson:2018ahi}; the volume
conjugate to $p$ is $V_\z$. Two limits illustrate these relations.
In flat spacetime,
\beq
 \begin{split}
 \z_{\rm flat}
 &=
 \frac{\kappa}{2R}
 \left[
 (R^2-t^2-r^2)\partial_t
 -
 2tr\,\partial_r
 \right],
 \\
 K&=\frac{d-2}{R},
 \qquad
 V_\z=\frac{\kappa R}{d+1}V.
 \end{split}
 \label{eq:cd_mss_flat_limit}
\eeq
The Smarr formula reduces to
$(d-2)A=(d-1)KV$.
For the full de Sitter static patch,
$R\rightarrow L$, one instead has
$\z\rightarrow\kappa L\partial_t$ and
$K\rightarrow0$. The generator becomes Killing,
the gravitational volume term vanishes, and
the first law is
\beq
 \d H_\z^{\tilde m}
 =
 -\frac{\kappa}{8\pi G}\,\d A
 -
 \frac{V_\z}{8\pi G}\,\d\Lambda.
 \label{eq:cd_mss_static_patch_firstlaw}
\eeq
With the usual static-time normalisation
$\z=\partial_t$, the surface gravity is
$\kappa=1/L$.

Finally, the semiclassical interpretation connects
these relations to entanglement equilibrium
\cite{Jacobson:2015hqa,Jacobson:2018ahi}.
For conformal matter, variations of the vacuum
state on the fixed background obey
\beq
 \d\langle H_\z^{\tilde m}\rangle
 =
 T_{\rm H}\,\d S_{\tilde m},
 \qquad
 T_{\rm H}=\frac{\hbar\kappa}{2\pi}=-T>0.
 \label{eq:cd_matter_entanglement_firstlaw}
\eeq
Substitution into the gravitational first law yields
\beq
 -\frac{\kappa K}{8\pi G}\,\d V
 +
 \frac{V_\z}{8\pi G}\,\d\Lambda
 =
 T\,\d S_{\rm gen},
 \qquad
 S_{\rm gen}=S_{\rm BH}+S_{\tilde m}.
 \label{eq:cd_equilibrium_generalised_entropy}
\eeq
Thus
\beq
 \left.\d S_{\rm gen}\right|_{V,\Lambda}=0.
 \label{eq:cd_entanglement_equilibrium}
\eeq
For nonconformal matter, the corresponding
small-diamond argument requires the additional
entropy assumptions described in
\cite{Jacobson:2015hqa,Jacobson:2018ahi}.
The positive matter temperature $T_{\rm H}$
and the negative gravitational temperature $T$
enter different first laws and are compatible
with this relation.

Equilibrium may also be expressed through the
free conformal energy
\beq
 F=H_\z-T S_{\rm BH},
 \qquad
 \left.\d F\right|_T=0,
 \label{eq:cd_free_conformal_energy}
\eeq
with the matter Hamiltonian replaced by its
expectation value in the semiclassical setting
\cite{Jacobson:2018ahi,JacobsonVisser:2019negative}.
These are first-order stationarity statements
about variations away from the reference diamond.
Establishing an entropy maximum requires further
information beyond first order. The physical
process balance between null cuts is addressed
separately in the main text.

\section{Dynamical entropy from Raychaudhuri equation}
\label{sec:dyn_entropy_raychaudhuri}

This appendix contains an independent check of the flux balance in Einstein
gravity using the null Raychaudhuri equation, following the approach of
\cite{Visser:2024pwz} for black holes. We keep the cosmological
constant $\Lambda$ fixed and use
\beq
 G_{ab}+\Lambda g_{ab}=8\pi T_{ab}.
\eeq
The null generators are affinely parametrised,
$\ell^a\nabla_a\ell^b=0$, and hypersurface orthogonal.
We use the curvature convention
$[\nabla_a,\nabla_b]V^c=\cR^c{}_{dab}V^d$.

On a smooth portion of $\cN^+$, let $v$ increase towards the
future, with $v=0$ at the edge, and choose coordinates on the
cross-sections transported along $\ell^a=(\partial_v)^a$.
The volume forms satisfy
\beq
 \boep^+_{d-1}=\ed v\wedge\boep^+_{d-2},
 \qquad
 \ell\i\boep^+_{d-1}=\boep^+_{d-2},
 \qquad
 \underline{\epsilon_a}=\ell_a\boep^+_{d-1},
\eeq
where $\boep^+_{d-2}$ is the positive area form on each
cross-section. Thus
$\partial\cN^+_{12}=\cC_2-\cC_1$ for $v_2>v_1$.
The boundary generator is
$\xi^a=\kappa_\beta v\ell^a$, with constant $\kappa_\beta>0$.

Let
\begin{align}
 q_{ab}&=g_{ab}+\ell_a n_b+n_a\ell_b,
 \qquad \ell^a n_a=-1,
 \nonumber\\
 K_{AB}&=\frac12\cL_\ell q_{AB}
       =\frac{\vartheta}{d-2}q_{AB}+\sigma_{AB},
 \qquad q^{AB}\sigma_{AB}=0.
 \label{eq:ray-optical-decomposition}
\end{align}
The twist vanishes because $\ell^a$ is normal to a null hypersurface. Using an affine extension of the generators, the Ricci identity gives
\begin{align}
 \cL_\ell\vartheta
 &=\nabla_a(\ell^b\nabla_b\ell^a)
   -(\nabla_a\ell^b)(\nabla_b\ell^a)
   -\cR_{ab}\ell^a\ell^b
 \nonumber\\
 &=-K_{AB}K^{AB}-\cR_{ab}\ell^a\ell^b
 \nonumber\\
 &=-\frac{\vartheta^2}{d-2}
   -\sigma_{AB}\sigma^{AB}
   -8\pi T_{ab}\ell^a\ell^b.
 \label{eq:ray-focusing}
\end{align}
The cosmological term drops out because $\ell^a$ is null.
Writing $\boep^+_{d-2}=\sfa(v,x)\ed^{d-2}x$, we have
\beq
 \partial_v\sfa=\sfa\vartheta.
\eeq
Combining this relation with \cref{eq:ray-focusing} gives
\beq
 \partial_v^2\sfa
 =\sfa\left[
    \frac{d-3}{d-2}\vartheta^2
    -\sigma_{AB}\sigma^{AB}
    -8\pi T_{ab}\ell^a\ell^b
   \right].
 \label{eq:ray-area-acceleration}
\eeq
For the constant boost normalisation used in the main text,
the charge entropy is
\begin{align}
 \cS(v)
 &=\frac14\int_{\cC(v)}\boep^+_{d-2}(1-v\vartheta)
 \nonumber\\
 &=\frac14\int\ed^{d-2}x\,
       \left(\sfa-v\partial_v\sfa\right).
 \label{eq:ray-entropy}
\end{align}
Its derivative therefore satisfies
\begin{align}
 \frac{\ed\cS}{\ed v}
 &=-\frac{v}{4}\int\ed^{d-2}x\,
                         \partial_v^2\sfa
 \nonumber\\
 &=\frac{v}{4}\int_{\cC(v)}\boep^+_{d-2}
 \left[
  \sigma_{AB}\sigma^{AB}
  -\frac{d-3}{d-2}\vartheta^2
  +8\pi T_{ab}\ell^a\ell^b
 \right].
 \label{eq:ray-entropy-evolution}
\end{align}
In four dimensions the expansion coefficient is $1/2$,
as in the Wald-Zoupas flux in the main text.
With the above orientation, the signed matter flux form is
\beq
 \bos{e}
 =-\xi^a T_a{}^b\underline{\epsilon_b}
 =-\kappa_\beta v\,T_{ab}\ell^a\ell^b\boep^+_{d-1}.
\eeq
The matter contribution and the gravitational flux on a strip
are consequently
\begin{align}
 \Delta\cE
 &=-\kappa_\beta
   \int_{v_1}^{v_2}\ed v\,v
   \int_{\cC(v)}\boep^+_{d-2}\,
                           T_{ab}\ell^a\ell^b,
 \label{eq:ray-signed-matter}\\
 \sf_\xi[\cN^+_{12}]
 &=-\frac{\kappa_\beta}{8\pi}
   \int_{v_1}^{v_2}\ed v\,v
   \int_{\cC(v)}\boep^+_{d-2}
   \left[
    \sigma_{AB}\sigma^{AB}
    -\frac{d-3}{d-2}\vartheta^2
   \right].
 \label{eq:ray-gravitational-flux}
\end{align}
In particular, $T_{ab}\ell^a\ell^b\geq0$ gives
$\Delta\cE\leq0$ for $v_2>v_1\geq0$.
Multiplying \cref{eq:ray-entropy-evolution} by
$-\kappa_\beta/(2\pi)$ and integrating yields
\beq
 -\frac{\kappa_\beta}{2\pi}
       \left(\cS[\cC_2]-\cS[\cC_1]\right)
 =\Delta\cE+\sf_\xi[\cN^+_{12}].
 \label{eq:ray-exact-balance}
\eeq
This verifies both the coefficient of the expansion term
and the relative signs in the charge balance. It is the
uniform-cut case of \cref{thm:exact_balance}. Varying
\cref{eq:ray-exact-balance} gives
\beq
 -\frac{\kappa_\beta}{2\pi}\Delta\d\cS
 =\Delta\d\cE+\d\sf_\xi[\cN^+_{12}],
 \label{eq:ray-linear-balance}
\eeq
which is in agreement with \cref{eqn:pp1st}. Here, the perturbed matter energy flux is given by
\beq
 \Delta\d\cE
 =-\kappa_\beta\int_{v_1}^{v_2}\ed v\,v
   \int_{\cC(v)}\boep^+_{d-2}\,
                    \d T_{ab}\ell^a\ell^b,
\eeq
with the background area measure. The variation of
\cref{eq:ray-gravitational-flux} acts on both the area measure
and the full metric contractions in its integrand.
For a non-vacuum background, the same balance follows by
varying the full matter integral in \cref{eq:ray-signed-matter},
including its area measure.
Writing $h_{AB}=\d q_{AB}$ and $h=q^{AB}h_{AB}$, we have
\beq
 \d\boep^+_{d-2}=\frac12h\boep^+_{d-2},
 \qquad
 \d\vartheta=\frac12\cL_\ell h.
\eeq
Consequently, the perturbed dynamical entropy is
\beq
 \d\cS
 =\frac18\int_{\cC(v)}\boep^+_{d-2}
       \left[h(1-v\vartheta)-v\cL_\ell h\right].
 \label{eq:ray-linear-entropy}
\eeq
The linearised entropy criterion is therefore
\beq
 \Delta\d\cE+\d\sf_\xi[\cN^+_{12}]\geq0
 \quad\Longleftrightarrow\quad
 \Delta\d\cS\leq0,
 \label{eq:ray-conditional-law}
\eeq
in agreement with \cref{cor:linearised_second_law}.
The sign condition concerns the combined signed flux; it
does not follow from the null energy condition alone on a
nonstationary background.

For comparison, if the background boundary is stationary,
$\vartheta=\sigma_{AB}=0$, then $\d\sf_\xi=0$ and
\beq
 \frac{\ed}{\ed v}\d\cS
 =2\pi v\int_{\cC(v)}\boep^+_{d-2}\,
                  \d T_{ab}\ell^a\ell^b\geq0
 \label{eq:ray-stationary-limit}
\eeq
under the linearised null energy condition. This agrees with
the stationary background result of
\cite{Hollands:2024vbe,Visser:2024pwz}. In our convention the matter contribution is then non-positive, so this result is consistent with \cref{eq:ray-linear-balance}. On a nonstationary background, the variation of the gravitational flux must also be retained.

Finally, with the normalisation used for the scalar field in \cref{sec:EH-scalar_gravity},
\beq
 8\pi T^{\Phi}_{ab}\ell^a\ell^b=(\cL_\ell\Phi)^2.
\eeq
Substitution into \cref{eq:ray-entropy-evolution} reproduces
the scalar contribution to the entropy evolution. The sum of the gravitational flux and this scalar matter flux is the full ES symplectic flux, so the scalar contribution is counted only once in the balance law.

\section{Examples}
\label{app:wz-examples}
In this appendix, we compute the WZ charge and flux for some specific examples in four dimensional spacetime. Recall, on $\cN^+$,
\begin{equation}
 \ell^a=(\partial_v)^a,\qquad
 \ell^a n_a=-1,\qquad
 \xi^a=\kappa_\beta v\ell^a,
\end{equation}
with constant $\kappa_\beta>0$. The localised charge form is
\begin{align}
 \bos{\sq}_\xi
 &=\frac{\boep^+_2}{8\pi}
   \left(n_b\ell^a\nabla_a\xi^b
         -\vartheta\xi^a n_a\right)
 \nonumber\\
 &=-\frac{\kappa_\beta}{8\pi}
     \boep^+_2(1-v\vartheta).
 \label{eq:examples-charge}
\end{align}
Here we have used affine parametrisation to obtain
$n_b\ell^a\nabla_a\xi^b=-\kappa_\beta$.
The boundary symplectic potential is, with
$\boep^+_3=\ed v\wedge\boep^+_2$,
\beq
 \bTh(g;h)
 =-\frac{\boep^+_3}{16\pi}
       \left(h_i{}^j\chi_j{}^i-h\vartheta\right).
\eeq
The corresponding flux is
\beq
 \bos{\sf}_\xi
 =\bTh(g;\cL_\xi g)
 =-\frac{\kappa_\beta v}{8\pi}\boep^+_3
       \left(\sigma_{AB}\sigma^{AB}
             -\frac12\vartheta^2\right).
 \label{eq:examples-flux}
\eeq
These are the only geometrical ingredients needed below.

\subsection{A finite spherical diamond in Minkowski spacetime}
\label{app:wz-minkowski}
For the first example, we consider a finite causal diamond in Minkowski spacetime. Consider the diamond whose maximal spatial section is
$t=0$, $r\leq R$, in
\begin{equation}
 \ed s^2=-\ed t^2+\ed r^2+r^2\ed\tom_2^2.
\end{equation}
Introduce double-null coordinates $(u,v,x^A)$ adapted to the future boundary, leaving the angular coordinates unchanged:
\begin{equation}
 t=v+\frac{u}{2},\qquad
 r=R-v+\frac{u}{2},
\end{equation}
so that
\begin{equation}
 \ed s^2=-2\ed u\ed v+
       \left(R-v+\frac{u}{2}\right)^2\ed\tom_2^2.
\end{equation}
Equivalently, $u=t+r-R$ and $v=(t-r+R)/2$. The future boundary $\cN^+$ is given by $u=0$, where $v=t=R-r$ is an affine parameter along its generators. The (bifurcation) edge $\cB$ lies at $u=v=0$, while the future corner is
approached as $v\to R$. On the null boundaries,
\beq
 \ell^a=(\partial_v)^a,\qquad
 n^a=(\partial_u)^a,\qquad
 q_{AB}=(R-v)^2\gamma_{AB},
\eeq
where $\gamma_{AB}$ is the metric on the unit two-sphere.
The coordinate $v$ is affine, and direct differentiation
of the induced metric gives
\beq
 K_{AB}=-(R-v)\gamma_{AB},\qquad
 \vartheta=-\frac{2}{R-v},\qquad
 \sigma_{AB}=0,\qquad
 A(v)=4\pi(R-v)^2.
\eeq
The boundary boost can be written as
\begin{equation}
 \xi=\kappa_\beta
       \left(v\partial_v-u\partial_u\right).
\end{equation}
It preserves the required boundary structure, although
it is not a Killing field of the finite diamond metric.
Substitution into \cref{eq:examples-charge} yields
\begin{align}
 \sq_\xi[\cC(v)]
 &=-\frac{\kappa_\beta}{8\pi}
     4\pi(R-v)^2
        \left(1+\frac{2v}{R-v}\right)
 \nonumber\\
 &=-\frac{\kappa_\beta}{2}(R^2-v^2).
\end{align}
The associated entropy is therefore
\begin{equation}
 \cS(v)=\pi(R^2-v^2).
 \label{eq:examples-flat-entropy}
\end{equation}
In particular, the entropy of the maximal slice is $\cS_0 = A (\cB)/4 =\pi R^2$. The gravitational flux follows independently from \cref{eq:examples-flux}. Using
$\boep^+_3=\ed v\wedge\boep^+_2$,
\begin{align}
 \sf_\xi[\cN^+_{12}]
 &=\frac{\kappa_\beta}{16\pi}
     \int_{v_1}^{v_2}\ed v\,v
     \int_{\cC(v)}\boep^+_2\,\vartheta^2
 \nonumber\\
 &=\kappa_\beta\int_{v_1}^{v_2}v\,\ed v
  =\frac{\kappa_\beta}{2}(v_2^2-v_1^2).
 \label{eq:examples-flat-flux}
\end{align}
Thus, we have
\begin{equation}
 \sq_\xi[\cC_2]-\sq_\xi[\cC_1]
 =\sf_\xi[\cN^+_{12}]
 =-\frac{\kappa_\beta}{2\pi}
       \left(\cS(v_2)-\cS(v_1)\right).
\end{equation}
This is a direct verification of the vacuum charge balance. The flux is nonzero in flat spacetime because the boundary has nonzero expansion; it is not a flux of gravitational radiation. Although the expansion diverges towards the corner $p^+$, which is a regular point in the spacetime, the charge and integrated flux have finite limits\footnote{It can be shown by using the Riemann normal coordinates in a small neighbourhood of the corner as discussed in section $3.2$ of \cite{Chandrasekaran:2019ewn}.}
\begin{equation}
 \lim_{v\to R}\cS(v)=0,\qquad
 \lim_{v\to R}\sf_\xi[\cN^+_{0v}]
       =\frac{\kappa_\beta R^2}{2}.
\end{equation}
These limits are the four-dimensional case of
\cref{cor:global_decrease}.

\subsection{A thin null shell through the Minkowski diamond}
\label{app:wz-shell}
To illustrate \cref{cor:linearised_second_law} and the criterion
\cref{eq:linearised_criterion}, perturb the diamond considered in the preceeding example by a spherical shell of null dust that
crosses $\cN^+$ at $v=v_p\in(0,R)$, with
\beq
 \delta T_{ab}\ell^a\ell^b=\mu\,\delta(v-v_p),
 \qquad \mu>0,
\eeq
so that the null energy condition holds. Let
$\sfa=(R-v)^2=r^2$ denote the area density per unit solid angle, and
write $r_p=R-v_p$. We fix the data at the edge,
$\delta\sfa=\partial_v\delta\sfa=0$ at $v=0$. Linearising
\cref{eq:ray-area-acceleration} gives
\beq
 \partial_v^2\delta\sfa
 -\frac{\partial_v\sfa}{\sfa}\,\partial_v\delta\sfa
 +\frac{(\partial_v\sfa)^2}{2\sfa^2}\,\delta\sfa
 =-8\pi\,\sfa\,\delta T_{ab}\ell^a\ell^b .
\eeq
The homogeneous solutions are $r$ and $r^2$. The shell produces a jump
$-8\pi\mu r_p^2$ in $\partial_v\delta\sfa$, so that for $v>v_p$
\beq
 \delta\sfa=8\pi\mu\,r_p\,r\,(r-r_p)\le0 ;
\eeq
the shell focuses the generators. With
$\delta\cS=\pi(\delta\sfa-v\,\partial_v\delta\sfa)$, and with the signed
fluxes measured from $\cB$, one finds
\begin{align}
 \delta\cS(v)&=8\pi^2\mu\,r_p\left(Rv_p-v^2\right)\Theta(v-v_p),
 \label{eq:shell_entropy}\\
 \Delta\delta\cE\big|_0^v&=-4\pi\kappa_\beta\,\mu\,v_p r_p^2\,\Theta(v-v_p)
 \;(\le0),
 \qquad
 \delta\sf_\xi\big|_0^v=4\pi\kappa_\beta\,\mu\,r_p\left(v^2-v_p^2\right)
 \Theta(v-v_p)\;(\ge0),
 \label{eq:shell_fluxes}\\
 \Delta\delta\cE+\delta\sf_\xi\big|_0^v
 &=4\pi\kappa_\beta\,\mu\,r_p\left(v^2-Rv_p\right)\Theta(v-v_p).
\end{align}
These satisfy \cref{eqn:pp1st} identically. As the matter obeying the
null energy condition crosses $\cN^+$, the entropy jumps up by
$8\pi^2\mu v_p r_p^2$.
Afterwards, the additional focusing drives it down. The combined
flux from $\cB$ is non-negative, and the net entropy change
non-positive, exactly when $v\ge\sqrt{Rv_p}$.

The hypothesis of \cref{cor:linearised_second_law} is therefore
non-trivial. It fails across the shell, where the sphere-integrated
density $-4\pi\kappa_\beta v_p r_p^2\mu\,\delta(v-v_p)$ is negative, and
holds after it, where the density $8\pi\kappa_\beta\mu r_p v$ is
positive. At the tip, the linearised value
$\delta\cS(R)=-8\pi^2\mu Rr_p^2$ does not vanish. This reflects the
first-order displacement of the tip, near which perturbation theory is
not uniform, rather than a failure of \cref{cor:global_decrease}.
\subsection{The Rindler wedge} \label{app:wz-rindler}
In Minkowski spacetime with inertial coordinates $(t,x,y,z)$, define the null coordinates
\begin{equation}
 u=\frac{t-x}{\sqrt{2}},\qquad
 v=\frac{t+x}{\sqrt{2}}.
\end{equation}
 On the future horizon $u=0$, $v>0$ is an affine parameter along the
generators, with $v=0$ at the bifurcation surface. In terms of $(u,v,y,z)$, the line element takes the form
\begin{equation}
 \ed s^2=-2\,\ed u\,\ed v+\ed y^2+\ed z^2.
\end{equation}
The right wedge has $u<0$, $v>0$, and its future horizon is $u=0$. The boost Killing field is
\begin{equation}
 \xi=\kappa_\beta \left(x\partial_t + t\partial_x\right)
    =\kappa_\beta \left(v\partial_v-u\partial_u\right).
\end{equation}
On the horizon choose $\ell=\partial_v$ and
$n=\partial_u$. Since both expansion and shear vanish,
\begin{equation}
 q_{AB}\ed x^A\ed x^B=\ed y^2+\ed z^2,
 \qquad K_{AB}=0.
\end{equation}
Let $A$ denote a regulated transverse area. Then, the charge and flux are given by
\begin{equation}
 \sq_\xi=-\frac{\kappa_\beta}{8\pi} A,
 \qquad
 \cS_0=\frac{A}{4},
 \qquad
 \sf_\xi=0.
\end{equation}
The identity \eqref{eqn:pp1st} therefore gives
\begin{align}
 -\frac{\kappa_\beta}{2\pi}\Delta\delta\cS
 &=\Delta\delta\cE=-\delta \cE_{out},\\
 \frac{\kappa_\beta}{2\pi}\Delta\delta\cS
 &=\int_{v_1}^{v_2}\ed v
       \int_{\cC(v)}\boep^+_2\,
                  \delta T_{ab}\xi^a\ell^b.
 \label{eq:examples-rindler-ppfl}
\end{align}
Here, $\cE_{out}$ denotes the matter flux going outward from the Rindler wedge. This relation holds for arbitrary horizon cross-sections
when the dynamical entropy is used. At first order,
\begin{equation}
 \delta\cS(v)=\frac14 \big(1-v\pa_v\big)\d A(v).
\end{equation}
If the endpoint correction $v\partial_v\delta A$ vanishes
at both ends, then
\begin{equation}
 \delta \cE_{out}
 =\frac{\kappa_\beta}{8\pi}\Delta\delta A.
\end{equation}
This reproduces the usual Rindler physical process first
law, including its positive area coefficient
\cite{Jacobson:2018ahi}.

\subsection{The de Sitter static patch}\label{app:wz-desitter}
Consider four-dimensional de Sitter static patch metric,
\begin{equation}
 \ed s^2=-F(r)\ed t^2+\frac{\ed r^2}{F(r)}
                   +r^2\ed\tom_2^2,\qquad
 F(r)=1-\frac{r^2}{L^2},\qquad
 \Lambda=\frac{3}{L^2},
\end{equation}
where $L$ is the de Sitter radius. The static patch is the region $0\leq r<L$, and $r=L$ is the cosmological horizon. We normalise $\xi=\partial_t$ to have unit norm at $r=0$. Its positive surface gravity on the future cosmological horizon is
\begin{equation}
  \kappa_{dS}=\frac{1}{L}.
\end{equation}
To construct coordinates regular at the horizon, it is useful to introduce the tortoise coordinate $r_*$, defined by
$\ed r_*/\ed r=F(r)^{-1}$ and $r_*(0)=0$:
\begin{equation}
 r_*=\frac{L}{2}\log\frac{L+r}{L-r},
\end{equation}
and regular null coordinates
\begin{equation}
 u=-\sqrt{2}L\,e^{-(t+r_*)/L},
 \qquad
 v=\sqrt{2}L\,e^{(t-r_*)/L}.
\end{equation}
The coordinates $(u,v,x^A)$ are Kruskal-type double-null coordinates, with the angular coordinates unchanged. In the static region, $u<0$ and $v>0$. The future cosmological horizon is $u=0$, where $v$ is an affine parameter along its generators; its origin $v=0$ is the bifurcation sphere. The dimensionless combination $\tilde{u}=uv/(2L^2)$ brings the metric to the form
\begin{equation}
 d s^2
 =-\frac{2\ed u\ed v}{(1-\tilde{u})^2}
  +L^2\left(\frac{1+\tilde{u}}{1-\tilde{u}}\right)^2
        \ed\tom_2^2.
 \label{eq:examples-ds-regular-metric}
\end{equation}
The future horizon is $u=0$, and its bifurcation sphere
is at $u=v=0$. On the future horizon,
\begin{equation}
 \ell=\partial_v,\qquad n=\partial_u,\qquad
 q_{AB}=L^2\gamma_{AB},\qquad
 \vartheta=0,\qquad\sigma_{AB}=0,
\end{equation}
where $\gamma_{AB}$ is the metric on the unit two-sphere.
The coordinate $v$ is affine there: in particular, $\Gamma^v{}_{vv}=u/[L^2(1-\tilde{u})]$ vanishes at $u=0$. Moreover,
\begin{equation}
 \xi=\partial_t
 =\frac{1}{L} \left(v\partial_v-u\partial_u\right),
 \qquad
 \xi\big|_{\cN^+}=\frac{v}{L}\ell.
\end{equation}
Thus, the static Killing field has exactly the boundary boost normalisation required by the charge construction. Evaluating \cref{eq:examples-charge,eq:examples-flux} gives
\begin{equation}
 A_0=4\pi L^2,\qquad
 \sq_\xi=-\frac{L}{2},\qquad
 \cS_0=\pi L^2,\qquad
 \sf_\xi=0.
 \label{eq:examples-ds-background}
\end{equation}
For a perturbation at fixed $\Lambda$, the first-order gravitational flux again vanishes, $\delta\sf_\xi=0$. The sourced charge identity becomes
\begin{equation}
 -\frac{1}{2\pi L}\Delta\delta\cS
 =\Delta\delta\cE=-\delta \cE_{out}.
\end{equation}
Equivalently,
\begin{align}
 \frac{1}{2\pi L}\Delta\delta\cS
 &=\delta \cE_{out}
 \nonumber\\
 &=\int_{v_1}^{v_2}\ed v
       \int_{\cC(v)}\boep^+_2\,
                      \delta T_{ab}\xi^a\ell^b
 \nonumber\\
 &=\frac{1}{L} \int_{v_1}^{v_2}\ed v\,v
       \int_{\cC(v)}\boep^+_2\,
                      \delta T_{ab}\ell^a\ell^b.
 \label{eq:examples-ds-ppfl}
\end{align}
This is the physical process first law for the dynamical cosmological horizon entropy \cite{Zhao:2025zny}.
The change of the matter Killing energy remaining inside the static patch satisfies
\begin{equation}
 \Delta\delta M=-\delta \cE_{out},
\end{equation}
provided the indicated horizon flux accounts for the energy exchange. Hence, the same law can be written as
\begin{equation}
 \Delta\delta M
 =-\frac{1}{2\pi L}\Delta\delta\cS.
 \label{eq:examples-ds-interior-law}
\end{equation}
The positive coefficient in
\cref{eq:examples-ds-ppfl} and the negative coefficient in \cref{eq:examples-ds-interior-law} are compatible: the former uses positive outgoing energy, while the latter uses the decrease of the interior energy. These are the two forms of the de Sitter physical
process law discussed in \cite{Zhao:2025zny}.
For arbitrary cross-sections,
\begin{equation}
 \delta\cS(v)=\frac14 \big(1-v\pa_v\big)\d A(v).
 \label{eq:examples-ds-dynamical entropy}
\end{equation}
If the endpoint expansion corrections vanish,
\cref{eq:examples-ds-ppfl} reduces to
\begin{equation}
 \delta \cE_{out}
 =\frac{1}{8\pi L}\Delta\delta A,
 \qquad
 \Delta\delta M
 =-\frac{1}{8\pi L}\Delta\delta A.
\end{equation}
More generally, no such endpoint assumption is needed when \cref{eq:examples-ds-dynamical entropy} is retained.
Finally, differentiating the integrated charge identity with respect to its upper endpoint gives
\begin{equation}
 \frac{\ed}{\ed v}\delta\cS
 =2\pi v\int_{\cC(v)}\boep^+_2\,
                     \delta T_{ab}\ell^a\ell^b
 \geq0
 \label{eq:examples-ds-second-law}
\end{equation}
under the null energy condition. Thus, positive energy leaving the static patch increases the perturbed dynamical horizon entropy, while decreasing the energy remaining inside the patch.

\subsection{Entanglement wedge in three-dimensional Poincar\'e AdS}
\label{sec:poincare_entanglement_wedge}

Consider three-dimensional Einstein gravity with
$\Lambda=-1/L_{AdS}^2$, in units $G=\hbar=1$. The Poincar\'e
AdS metric is
\begin{equation}
 ds^2=\frac{L_{AdS}^2}{z^2}
 \left(-\ed t^2+\ed z^2+\ed x^2\right),
 \qquad z>0.
 \label{eq:ew3_poincare_metric}
\end{equation}
For the boundary interval $B=\{t=0,-R<x<R\}$, the
vacuum Ryu--Takayanagi surface and entanglement wedge
are
\begin{align}
 \gamma_B &: \quad t=0,\qquad z^2+x^2=R^2,
 \label{eq:ew3_rt_surface}\\
 \cW_B &: \quad |t|+\sqrt{z^2+x^2}\leq R.
 \label{eq:ew3_wedge}
\end{align}
The curve $\gamma_B$ is the bifurcation surface of the
wedge-preserving Killing horizon. Its length computes
the interval entanglement entropy at leading classical
order \cite{Ryu:2006bv,Casini:2011kv}.

We introduce Gaussian null coordinates
$(u,v,\upsilon)$, where $\upsilon$ is dimensionless and
labels the null generators. Define the auxiliary
functions
\begin{equation}
 b(u,v)=1-\frac{uv}{L_{AdS}^2},
 \qquad
 \mathsf{D}(u,v,\upsilon)
 =
 L_{AdS}b\cosh\upsilon+\frac{(b+1)v-u}{2}.
 \label{eq:ew3_auxiliary_functions}
\end{equation}
The transformation to Poincar\'e coordinates is
\begin{align}
 t&=\frac{R\big[(b+1)v+u\big]}{2\mathsf{D}},
 \nonumber\\
 x&=\frac{RL_{AdS}b\sinh\upsilon}{\mathsf{D}},
 \qquad
 z=\frac{RL_{AdS}}{\mathsf{D}}.
 \label{eq:ew3_coordinate_transformation}
\end{align}
We work in a neighbourhood of the horizon with
$b>0$ and $\mathsf{D}>0$, so that the coordinates are regular
and remain in the Poincar\'e patch. Direct substitution
into \cref{eq:ew3_poincare_metric} gives
\begin{equation}
 ds^2
 =
 -2\,\ed u\,\ed v
 -\frac{u^2}{L_{AdS}^2}\ed v^2
 +L^2\left(1-\frac{uv}{L_{AdS}^2}\right)^2
       \ed\upsilon^2.
 \label{eq:ew3_gnc_metric}
\end{equation}
Thus
\begin{equation}
 g_{uu}=g_{u\upsilon}=0,
 \qquad
 g_{uv}=-1,
 \qquad
 g_{v\upsilon}=0,
 \qquad
 g_{vv}=O(u^2),
 \label{eq:ew3_gnc_conditions}
\end{equation}
which is the affine Gaussian null form adapted to
the future horizon $u=0$.
At $u=v=0$, the transformation reduces to
\begin{equation}
 t=0,
 \qquad
 x=R\tanh\upsilon,
 \qquad
 z=R\,\operatorname{sech}\upsilon,
 \label{eq:ew3_bifurcation_coordinates}
\end{equation}
so that this cut is precisely $\gamma_B$.
On $u=0$, it becomes
\begin{equation}
 t=\frac{Rv}{L_{AdS}\cosh\upsilon+v},
 \qquad
 x=\frac{RL_{AdS}\sinh\upsilon}{L_{AdS}\cosh\upsilon+v},
 \qquad
 z=\frac{RL_{AdS}}{L_{AdS}\cosh\upsilon+v}.
 \label{eq:ew3_future_horizon_coordinates}
\end{equation}
These expressions satisfy
\begin{equation}
 t+\sqrt{z^2+x^2}=R.
 \nonumber
\end{equation}
Hence $u=0$, $v\geq0$ describes the future wedge
horizon $\cN_B^+$. Similarly, $v=0$, $u\leq0$
describes the past horizon. The coordinate $u$ is
transverse to $\cN_B^+$; away from that horizon its
constant-value hypersurfaces need not be null.
On $\cN_B^+$, 
\begin{equation}
 \ell^a=(\partial_v)^a,
 \qquad
 n^a=(\partial_u)^a.
 \label{eq:ew3_null_normals}
\end{equation}
They obey
$\ell^a\ell_a=n^an_a=0$ and $\ell^an_a=-1$ there.
Moreover, \cref{eq:ew3_gnc_metric} gives
\begin{equation}
 \Gamma^a{}_{vv}\big|_{u=0}=0,
 \qquad
 \ell^b\nabla_b\ell^a\nulleq0.
 \label{eq:ew3_affineness}
\end{equation}
Thus, $v$ is affine and vanishes at $\gamma_B$,
while $\upsilon$ is constant along the generators.
The transverse vector $\partial_u$ is also affinely
parametrised, since $\Gamma^a{}_{uu}=0$.
The wedge-preserving Killing vector takes the simple
form
\begin{equation}
 \xi
 =
 \kappa_\beta
 \left(v\partial_v-u\partial_u\right),
 \qquad
 \kappa_\beta=-\beta>0.
 \label{eq:ew3_killing_vector}
\end{equation}
Indeed, the metric is invariant under
$v\mapsto e^s v$, $u\mapsto e^{-s}u$.
In Poincar\'e coordinates, the same vector is
\begin{equation}
 \xi=
 \frac{\kappa_\beta}{2R}
 \left[
 (R^2-t^2-x^2-z^2)\partial_t
 -2t(x\partial_x+z\partial_z)
 \right].
 \label{eq:ew3_killing_poincare}
\end{equation}
It vanishes on $\gamma_B$, and its restriction to
the future horizon is
\begin{equation}
 \xi^a\nulleq
 \kappa_\beta v\ell^a=-\beta v\ell^a.
 \label{eq:ew3_horizon_boost}
\end{equation}
The choice $\kappa_\beta=2\pi$ gives the vacuum
modular-flow normalisation \cite{Casini:2011kv}.
The constant-$v$ cuts $\cC(v)$ of $\cN_B^+$ have
induced metric and positively oriented length form
\begin{equation}
 q_{\upsilon\upsilon}=L_{AdS}^2,
 \qquad
 \boep_1^+=L_{AdS}\,\ed\upsilon.
 \label{eq:ew3_cut_geometry}
\end{equation}
Consequently,
\begin{equation}
 K_{\upsilon\upsilon}
 =
 \frac12\cL_\ell q_{\upsilon\upsilon}=0,
 \qquad
 \vartheta
 =
 \frac12q^{\upsilon\upsilon}
       \partial_v q_{\upsilon\upsilon}=0.
 \label{eq:ew3_expansion}
\end{equation}
The shear vanishes identically for a one-dimensional
cut, since its null second fundamental form has no
trace-free part. The absence of expansion here is
therefore manifest directly in Gaussian null
coordinates.

The cut length requires a regulator. We introduce
$z=z_{\mathrm{UV}}$ on $\gamma_B$ and transport
the resulting endpoints along the generators.
By \cref{eq:ew3_bifurcation_coordinates}, this fixes
\begin{equation}
 |\upsilon|\leq\upsilon_{\max},
 \qquad
 \cosh\upsilon_{\max}=\frac{R}{z_{\mathrm{UV}}}.
 \label{eq:ew3_cutoff}
\end{equation}
The same range of $\upsilon$ is used on every cut,
and hence
\begin{equation}
 A[\cC(v)]
 =
 \int_{-\upsilon_{\max}}^{\upsilon_{\max}}
 L_{AdS}\,\ed\upsilon
 =
 2L_{AdS}\,\operatorname{arcosh}
 \left(\frac{R}{z_{\mathrm{UV}}}\right).
 \label{eq:ew3_length}
\end{equation}
Here $A$ denotes the codimension-two area, which
is a length in three dimensions. A fixed Poincar\'e
cutoff on every cut would instead change the
retained generators. The transported regulator
avoids this additional contribution to the length
variation.

The three-dimensional version of the dynamical
entropy functional is
\begin{equation}
 \cS[\cC(v)]
 =
 \frac14\int_{\cC(v)}
 \boep_1^+(1-v\vartheta).
 \label{eq:ew3_dynamical_entropy}
\end{equation}
Using \cref{eq:ew3_expansion,eq:ew3_length}, we obtain
\begin{align}
 \cS[\cC(v)]
 &=
 \frac{L_{AdS}}{2}\operatorname{arcosh}
 \left(\frac{R}{z_{\mathrm{UV}}}\right)
 \nonumber\\
 &=
 \frac{L_{AdS}}{2}\log\left(\frac{2R}{z_{\mathrm{UV}}}\right)
 +O\left(\frac{L_{AdS}z_{\mathrm{UV}}^2}{R^2}\right)
 \nonumber\\
 &=
 \cS_{\mathrm{EE}}^{\mathrm{vac}}(B),
 \label{eq:ew3_vacuum_entropy}
\end{align}
where the last equality is understood at leading
classical holographic order with matched regulators.
Thus the dynamical entropy agrees with the vacuum
interval entropy on every cut of this non-expanding
horizon. This equality transports the entropy of
$\gamma_B$ along the horizon; it does not imply that
the later cuts are themselves extremal curves.

To identify the role of the expansion correction,
consider a smooth perturbation about this background,
in an adapted GNC which fixes the
generator labels, the affine parameter and its
origin. In particular, on the identified perturbed
horizon we retain
$\ell=\partial_v$ and affine parametrisation.
All variations below are evaluated on the AdS
background, and the regulator is transported with
the generators. Since $\vartheta_0=0$,
\begin{align}
 \d\cS[\cC(v)]
 &=
 \frac14\left[
 \d A(v)
 -v\int_{\cC(v)}\boep_{1,0}^+\,\d\vartheta
 \right]
 \nonumber\\
 &=
 \frac14\left(
 1-v\frac{\ed}{\ed v}
 \right)\d A(v),
 \label{eq:ew3_entropy_variation}
\end{align}
where
\begin{equation}
 \frac{\ed\,\d A}{\ed v}
 =
 \int_{\cC(v)}\boep_{1,0}^+\,\d\vartheta.
 \nonumber
\end{equation}
This is the Einstein gravity form of the
Hollands-Wald-Zhang expansion correction
\cite{Hollands:2024vbe}.
In three dimensions the affine null constraint is
\begin{equation}
 \partial_v\vartheta
 =
 -\vartheta^2-8\pi T_{ab}\ell^a\ell^b.
 \label{eq:ew3_null_constraint}
\end{equation}
There is no shear term, and the cosmological term
vanishes upon contraction with $\ell^a\ell^b$.
Linearising about the vacuum horizon gives
\begin{equation}
 \partial_v\d\vartheta
 =
 -8\pi\,\d T_{ab}\ell^a\ell^b.
 \label{eq:ew3_linearised_constraint}
\end{equation}
Therefore,
\begin{align}
 \frac{\ed^2\,\d A}{\ed v^2}
 &=
 -8\pi\int_{\cC(v)}
 \boep_{1,0}^+\,\d T_{ab}\ell^a\ell^b,
 \nonumber\\
 \frac{\ed\,\d\cS}{\ed v}
 &=
 -\frac{v}{4}\frac{\ed^2\,\d A}{\ed v^2}
 \nonumber\\
 &=
 2\pi v\int_{\cC(v)}
 \boep_{1,0}^+\,\d T_{ab}\ell^a\ell^b.
 \label{eq:ew3_entropy_balance}
\end{align}
For $v_2>v_1\geq0$, this yields
\begin{align}
 &\frac{\kappa_\beta}{2\pi}
 \left(
 \d\cS[\cC(v_2)]-\d\cS[\cC(v_1)]
 \right)
 \nonumber\\
 &\qquad =
 \int_{v_1}^{v_2}\ed v
 \int_{\cC(v)}
 \boep_{1,0}^+\,\d T_{ab}\xi^a\ell^b.
 \label{eq:ew3_integrated_balance}
\end{align}
The right-hand side uses the positive
future directed flux measure. With the signed WZ
charge convention
$\sq_\xi=-\kappa_\beta\cS/(2\pi)$, the corresponding
change in $\d\sq_\xi$ is its negative.

At the reference cut $v=0$, the expansion correction
vanishes. With fixed boundary anchoring and the
usual extremal-surface boundary conditions,
displacement of the reference extremal curve
does not contribute to its first-order length
variation. In the classical holographic regime
\cite{Hubeny:2007xt}, one consequently obtains
\begin{equation}
 \d\cS[\gamma_B]
 =
 \frac14\d A(\gamma_B)
 =
 \d\cS_{\mathrm{EE}}(B)
 =
 \d\langle K_B\rangle,
 \label{eq:ew3_entanglement_first_law}
\end{equation}
where
\begin{equation}
 K_B=
 2\pi\int_{-R}^{R}\ed x\,
 \frac{R^2-x^2}{2R}\,
 T_{tt}^{\mathrm{CFT}}
 +\mathrm{constant}
 \label{eq:ew3_modular_hamiltonian}
\end{equation}
is the modular Hamiltonian of the reference vacuum
state \cite{Casini:2011kv,Faulkner:2013ica}.
Bulk quantum entropy corrections are not included.

On later cuts, the dynamical entropy has the
gravitational interpretation expressed by
\cref{eq:ew3_integrated_balance}: its first-order
change measures matter boost-energy flux through
the horizon. In particular, a flux-free perturbation
with $\d A(v)=a_0+a_1v$ has $\d\cS=a_0/4$.
The expansion correction removes this homogeneous
linear length evolution. An identification with
boundary von Neumann entropy on a general dynamical
cut would require an additional prescription,
because such a cut need not be extremal.

\section{GSL for black holes with INEC}
\label{sec:bh_gsl}

In \cref{sec:generalise_second_law_cd}, we used the improved
quantum null energy condition (INEC) to formulate a
semiclassical entropy balance for the future null boundary
of a causal diamond. In this appendix, we discuss the
corresponding construction for a black hole horizon. The
comparison makes explicit the relation to the generalised
second law obtained by Hollands, Wald and Zhang (HWZ)
\cite{Hollands:2024vbe}, including the role of the remainder
in the quantum energy bound. We work in four-dimensional
semiclassical Einstein gravity and set $G_N=\hbar=1$.

Consider the future event horizon $\cH^+$ of a stationary
black hole with a bifurcate Killing horizon. Let $k^a$ be
its future directed, affinely parametrised null generator,
and choose the affine parameter $v$ to vanish at the
bifurcation surface. On $\cH^+$,
\beq
 \xi^a=\kappa v\,k^a,
 \qquad
 \kappa>0,
 \label{eq:bh_killing_generator}
\eeq
where $\xi^a$ is the horizon-generating Killing vector.
The Hawking temperature is $T_{\rm BH}=\kappa/(2\pi)$.
We use uniform cuts $\cC(v)$ with $v\ge0$, and write
\beq
 \sfa=\sqrt{q},
 \qquad
 \boep_2=\sfa\,d^2y,
 \qquad
 \vartheta=\frac{1}{\sfa}\frac{d\sfa}{dv}.
 \label{eq:bh_area_expansion}
\eeq
Here $\boep_2$ is the positive area form on a cut.
The background expansion and shear vanish. The quantities
$k^a$, $v$, $\xi^a$, $\kappa$, and the identification of the
cuts are held fixed under the perturbation.

Let $\cS_{vN}[v]$ denote the entropy of the quantum fields
outside the horizon cut $v=v(y)$. We use nested exterior
regions and the same assumptions concerning renormalisation
and shape derivatives as in
\cref{sec:generalise_second_law_cd}. On uniform cuts we write
$\cS_{vN}(v)=\cS_{vN}[v(y)=v]$. The quantum expansion is
\beq
 \vartheta_\fq
 =
 \vartheta+\frac{4}{\sfa}
 \frac{\delta\cS_{vN}[v]}{\delta v(y)}.
 \label{eq:bh_quantum_expansion}
\eeq

For comparison with the causal diamond calculation, we
first retain the expansion dependent form of the quantum
bound,
\beq
 \sfa\,\langle T_{ab}\rangle k^ak^b
 \ge
 \frac{1}{2\pi}
 \left[
 \left(
 \frac{\delta^2\cS_{vN}[v]}{\delta v(y)^2}
 \right)_{\rm diag}
 -
 \frac{1}{2}\vartheta
 \frac{\delta\cS_{vN}[v]}{\delta v(y)}
 \right].
 \label{eq:bh_INEC}
\eeq
Away from a stationary horizon, this is the conditional
INEC discussed in \cref{sec:generalise_second_law_cd},
whose derivation uses vanishing quantum expansion and
a shear-free congruence \cite{Ben-Dayan:2023inz}.
For the stationary background considered here, the required
bound is the ordinary QNEC, which we assume directly on the
Killing horizon; no condition of vanishing quantum expansion
is imposed on the perturbed state.

For a deformation of a single infinitesimal pencil of
generators, the differential combination in
\cref{eq:bh_INEC} obeys
\beq
 \frac{d^2\cS_{vN}}{dv^2}
 -
 \frac{1}{2}\vartheta\frac{d\cS_{vN}}{dv}
 =
 \sfa^{1/2}\frac{d}{dv}
 \left(
 \sfa^{-1/2}\frac{d\cS_{vN}}{dv}
 \right).
 \label{eq:bh_INEC_identity}
\eeq
For uniform displacements of the whole cut, we instead
use the full second derivative and define
\beq
 \mathcal{I}_{\rm INEC}^{\rm BH}(v)
 \equiv
 \frac{d^2\cS_{vN}(v)}{dv^2}
 -
 \frac{1}{2}
 \int_{\cC(v)}d^2y\,\vartheta(v,y)
 \left.
 \frac{\delta\cS_{vN}[v]}{\delta v(y)}
 \right|_{v(y)=v}.
 \label{eq:bh_I_INEC}
\eeq
The mixed shape derivatives on distinct generators
contribute non-positively by strong subadditivity
\cite{Bousso:2015mna}. Thus, as in
\cref{sec:generalise_second_law_cd}, integration of the
local bound gives
\beq
 \int_{\cC(v)}\boep_2\,
 \langle T_{ab}\rangle k^ak^b
 \ge
 \frac{1}{2\pi}\mathcal{I}_{\rm INEC}^{\rm BH}(v).
 \label{eq:bh_integrated_INEC}
\eeq

We evaluate the quantum inequality on the fixed background
geometry and perturb the quantum state. We assume that
the first-order perturbation satisfies
\beq
 \int_{\cC(v)}\boep_2\,
 \d\langle T_{ab}\rangle k^ak^b
 \ge
 \frac{1}{2\pi}\d\mathcal{I}_{\rm INEC}^{\rm BH}(v).
 \label{eq:bh_pert_INEC}
\eeq
As in the causal diamond discussion, a sufficient condition
is a differentiable one-sided family of states satisfying
\cref{eq:bh_integrated_INEC}, with the reference state
saturating it at every cut considered. For a differentiable
family admissible for both signs of the perturbation
parameter, saturation forces equality at first order.
The gravitational backreaction is incorporated by treating
$\d\langle T_{ab}\rangle$ as an external first-order source;
the matter entropy derivatives are evaluated on the
reference geometry at leading semiclassical order.

The background satisfies
\beq
 \cL_\xi g_{ab}=0.
 \label{eq:bh_killing_condition}
\eeq
Together with the stationarity property of the
Wald-Zoupas potential, this implies that the first-order
gravitational flux vanishes. The local physical process
first law for the HWZ dynamical entropy $\cS_{\rm dyn}$
therefore reads \cite{Hollands:2024vbe}
\beq
 \frac{\kappa}{2\pi}
 \frac{d}{dv}\d\cS_{\rm dyn}
 =
 \int_{\cC(v)}\boep_2\,
 \d\langle T_{ab}\rangle\xi^ak^b.
 \label{eq:bh_local_ppfl}
\eeq
Using \cref{eq:bh_killing_generator,eq:bh_pert_INEC},
we obtain
\begin{align}
 \frac{\kappa}{2\pi}
 \frac{d}{dv}\d\cS_{\rm dyn}
 &=
 \kappa v
 \int_{\cC(v)}\boep_2\,
 \d\langle T_{ab}\rangle k^ak^b
 \nonumber\\
 &\ge
 \frac{\kappa v}{2\pi}\,
 \d\mathcal{I}_{\rm INEC}^{\rm BH}(v),
 \label{eq:bh_ppfl_INEC}
\end{align}
and hence
\beq
 \frac{d}{dv}\d\cS_{\rm dyn}
 \ge
 v\,\d\mathcal{I}_{\rm INEC}^{\rm BH}(v).
 \label{eq:bh_entropy_INEC}
\eeq
We define the modified matter contribution by
\beq
 \frac{d}{dv}\d\cS_q^{\rm BH}(v)
 \equiv
 -v\,\d\mathcal{I}_{\rm INEC}^{\rm BH}(v).
 \label{eq:bh_Sq_definition}
\eeq
Equivalently,
\beq
 \d\cS_q^{\rm BH}(v)
 =
 \d\cS_q^{\rm BH}(v_0)
 -
 \int_{v_0}^{v}d\bar v\,
 \bar v\,\d\mathcal{I}_{\rm INEC}^{\rm BH}(\bar v).
 \label{eq:bh_Sq_integrated}
\eeq
The corresponding non-negative first-order remainder is
\beq
 \mathsf{R}_{\rm BH}(v)
 \equiv
 2\pi\int_{\cC(v)}\boep_2\,
 \d\langle T_{ab}\rangle k^ak^b
 -
 \d\mathcal{I}_{\rm INEC}^{\rm BH}(v)
 \ge0.
 \label{eq:bh_INEC_remainder}
\eeq

\begin{theorem}[Linearised GSL for black holes]
\label{thm:bh_linearised_GSL}
For first-order perturbations of the stationary black hole
satisfying \cref{eq:bh_pert_INEC}, the sum of the dynamical
gravitational entropy and the modified matter contribution
obeys
\beq
 \frac{d}{dv}
 \left(\d\cS_{\rm dyn}+\d\cS_q^{\rm BH}\right)
 \ge0.
 \label{eq:bh_GSL_INEC}
\eeq
Consequently, for $\cC(v_2)$ to the future of $\cC(v_1)$,
\beq
 \Delta\d\cS_{\rm dyn}
 +
 \Delta\d\cS_q^{\rm BH}
 \ge0.
 \label{eq:bh_integrated_GSL}
\eeq
\end{theorem}

\begin{proof}
Combining \cref{eq:bh_local_ppfl,eq:bh_Sq_definition}
and using \cref{eq:bh_INEC_remainder} gives the exact
first-order balance
\beq
 \frac{\kappa}{2\pi}
 \frac{d}{dv}
 \left(\d\cS_{\rm dyn}+\d\cS_q^{\rm BH}\right)
 =
 \frac{\kappa}{2\pi}\,
 v\,\mathsf{R}_{\rm BH}(v).
 \label{eq:bh_combined_entropy_balance}
\eeq
The right-hand side is non-negative since
$\kappa>0$, $v\ge0$, and $\mathsf{R}_{\rm BH}\ge0$.
Integration yields
\beq
 \frac{\kappa}{2\pi}
 \left(
 \Delta\d\cS_{\rm dyn}
 +
 \Delta\d\cS_q^{\rm BH}
 \right)
 =
 \frac{\kappa}{2\pi}
 \int_{v_1}^{v_2}dv\,v\,\mathsf{R}_{\rm BH}(v)
 \ge0,
 \label{eq:bh_integrated_entropy_balance}
\eeq
which proves the result.
\end{proof}
To express the same statement in terms of the matter
energy flux, define
\beq
 \Delta\d\cE_{\rm BH}
 \equiv
 \kappa
 \int_{v_1}^{v_2}dv\,v
 \int_{\cC(v)}\boep_2\,
 \d\langle T_{ab}\rangle k^ak^b.
 \label{eq:bh_matter_energy_flux}
\eeq
The integrated balance can then be written as
\beq
 T_{\rm BH}
 \left(
 \Delta\d\cS_{\rm dyn}
 +
 \Delta\d\cS_q^{\rm BH}
 \right)
 =
 \Delta\d\cE_{\rm BH}
 +
 T_{\rm BH}\Delta\d\cS_q^{\rm BH}
 \ge0.
 \label{eq:bh_entropy_flux_condition}
\eeq
Since $T_{\rm BH}>0$, non-negativity of the full
right-hand side is equivalent to a non-negative combined
entropy change. The assumed quantum bound ensures this
positivity through \cref{eq:bh_integrated_entropy_balance}.
Strict positivity of the weighted remainder integral
gives a strict increase.

\begin{corollary}[Saturation of the quantum energy bound]
\label{cor:bh_saturated_GSL}
If \cref{eq:bh_pert_INEC} is saturated throughout the
interval, then $\mathsf{R}_{\rm BH}(v)=0$ and
\beq
 \frac{d}{dv}
 \left(\d\cS_{\rm dyn}+\d\cS_q^{\rm BH}\right)=0.
 \label{eq:bh_saturated_GSL}
\eeq
In particular,
$\Delta\d\cS_{\rm dyn}+\Delta\d\cS_q^{\rm BH}=0$.
\end{corollary}
This follows immediately from
\cref{eq:bh_combined_entropy_balance}. The individual
gravitational and matter entropy changes need not vanish;
they cancel in their sum at first order.

On the stationary background used throughout this
derivation, $\vartheta=0$, so the quantum inequality is
the QNEC,
\beq
 \sfa\,\langle T_{ab}\rangle k^ak^b
 \ge
 \frac{1}{2\pi}
 \left(
 \frac{\delta^2\cS_{vN}[v]}{\delta v(y)^2}
 \right)_{\rm diag}.
 \label{eq:bh_QNEC}
\eeq
Correspondingly, \cref{eq:bh_I_INEC} becomes
\beq
 \mathcal{I}_{\rm INEC}^{\rm BH}(v)
 =
 \frac{d^2\cS_{vN}(v)}{dv^2}.
 \label{eq:bh_I_QNEC}
\eeq
It follows from \cref{eq:bh_Sq_definition} that
\beq
 \d\cS_q^{\rm BH}
 =
 \d\cS_{vN}
 -
 v\,\frac{d}{dv}\d\cS_{vN}
 +
 \text{constant}.
 \label{eq:bh_modified_vN}
\eeq
Thus the modified contribution agrees, to first order,
with the HWZ modified von Neumann entropy. Substitution
into \cref{eq:bh_GSL_INEC} gives
\beq
 \frac{d}{dv}
 \left[
 \d\cS_{\rm dyn}
 +
 \d\left(
 \cS_{vN}-v\,\frac{d\cS_{vN}}{dv}
 \right)
 \right]
 \ge0,
 \label{eq:bh_HWZ_GSL}
\eeq
recovering the HWZ result \cite{Hollands:2024vbe}.
At this order, the expansion dependent term in INEC
therefore supplies no additional correction on the
stationary background.

The comparison with the causal diamond is most directly
made at the level of the exact entropy balances.
For the future boundary of the diamond,
\cref{sec:generalise_second_law_cd} gives
\beq
 -\frac{\kappa_\beta}{2\pi}
 \left(\Delta\d\cS+\Delta\d\cS_q\right)
 =
 \d\sf_\xi[\cN^+_{v_1v_2}]
 -
 \frac{\kappa_\beta}{2\pi}
 \int_{v_1}^{v_2}dv\,v\,\mathsf{R}(v).
 \label{eq:bh_cd_entropy_comparison}
\eeq
Positivity of the full right-hand side implies a decrease of the combined perturbed entropy. When the diamond's quantum bound is saturated, this condition reduces to positivity of $\d\sf_\xi[\cN^+_{v_1v_2}]$. For the stationary black hole, the right-hand side of \cref{eq:bh_integrated_entropy_balance} is non-negative and the temperature is positive, so the combined entropy
increases or remains constant. Saturation makes its first-order change vanish because there is no first-order gravitational flux.

The comparison concerns the full balances, including their oriented matter contributions and the gravitational flux present for a generic diamond.
Extending the black hole argument to a nonstationary background would require both an applicable quantum energy bound and the corresponding physical process balance with its first-order gravitational flux retained. The zero-flux first law used here applies to perturbations of the stationary background. The second order vacuum flux discussed in \cref{sec:physical_process_1st_law}
belongs to a different perturbative order and does not supply an additional term in the first-order GSL above.

\section{Symmetry of a null surface in Einstein-scalar gravity} \label{sec:ES-null-symmetry}
In this appendix we briefly discuss the symmetry of a null surface in the Einstein-scalar theory.
The minimally coupled scalar as considered in \cref{sec:EH-scalar_gravity} does not change the
kinematical definition of the null-boundary structure.
Accordingly, we start with the boundary-preserving
diffeomorphisms introduced in
\cref{eq:symm_generator_GNC}. In affine coordinates
$(v,x^A)$ on $\cN$, their generators are restrict to
\begin{equation}
 \xi\big|_{\cN}
 =f\ell+X^A\pa_A,\qquad
 f=\varsigma-v\beta,\qquad
 \ell=\pa_v,
 \label{eq:ES-symmetry-generator}
\end{equation}
where
\begin{equation}
 \pa_v\varsigma=0,\qquad
 \pa_v\beta=0,\qquad
 \pa_vX^A=0.
\end{equation}
These conditions imply
\begin{equation}
 [\xi,\ell]=\beta\ell.
 \label{eq:ES-generator-commutator}
\end{equation}
A boundary preserving diffeomorphism generally acts on both dynamical fields,
\begin{equation}
 \delta_\xi g_{ab}=\cL_\xi g_{ab},
 \qquad
 \delta_\xi\P=\cL_\xi\P.
\end{equation}
The presence of the scalar does not by itself require $\delta_\xi\P=0$. In the following, we impose this additional condition and determine the subalgebra preserving a prescribed scalar profile on $\cN$. Thus, we consider the scalar stabiliser within the null-boundary symmetry algebra, rather than the isometries of a complete Einstein-scalar solution.
The scalar-preservation condition is
\begin{equation}
 \cL_\xi\P
 =(\varsigma-v\beta)\pa_v\P
   +X^A\pa_A\P=0.
 \label{eq:ES-scalar-stabiliser}
\end{equation}
\\

\paragraph{Case I:}\label{sec:ES_symmetry_i} First, we choose that the scalar field is constant along the null generators, i.e.,
\begin{equation}
 \cL_\ell\P=\pa_v\P=0.
 \label{eq:ES-caseI}
\end{equation}
Then the scalar restricted to $\cN$ has the form
\begin{equation}
 \P(v,x^A)=\P_0(x^A).
\end{equation}
Substituting into \cref{eq:ES-scalar-stabiliser}
gives
\begin{equation}
 X^A\pa_A\P_0=0.
 \label{eq:ES-caseI-restriction}
\end{equation}
There is no additional restriction on $\varsigma$
or $\beta$. The affine supertranslations and the rescalings of the generators therefore survive, whereas the transverse diffeomorphisms must preserve the scalar profile.
If $\P_0$ is constant over the cross-section, \cref{eq:ES-caseI-restriction} is automatic and the complete null surface algebra is retained. If $\P_0$ is nonconstant, the allowed transverse fields form the subalgebra
\begin{equation}
 \mathfrak{diff}(\cZ)_{\P_0} =
 \left\{  X\in\mathfrak{diff}(\cZ): X^A\pa_A\P_0=0 \right\}.
 \label{eq:ES-transverse-stabiliser}
\end{equation}
Geometrically, these fields are tangent to the level sets of $\P_0$. On a patch where $\ed\P_0\neq0$, one may choose coordinates $(p,z)$ with $p=\P_0$.
Then
\begin{equation}
 X=Y(p,z)\pa_z,
\end{equation}
where $Y$ is arbitrary. This local description must extend smoothly across critical points of $\P_0$ when a globally defined symmetry is required.
Closure of the algebra is straightforward,
\begin{equation}
 [X_1,X_2]\P_0
 =X_1(X_2\P_0)-X_2(X_1\P_0)=0.
\end{equation}
The resulting algebra is consequently
\begin{equation}
 \mathfrak g_{\P_0} =
 \mathfrak{diff}(\cZ)_{\P_0}  \ltimes
 \left(\mathfrak w\ltimes\mathfrak t_0\right),
 \label{eq:ES-caseI-algebra}
\end{equation}
where $\mathfrak t_0$ denotes affine supertranslations and $\mathfrak w$ denotes the Weyl rescalings. The combined longitudinal sector is nonabelian: for $X_1=X_2=0$,
\begin{equation}
 [\xi_1,\xi_2]
 =\left(\varsigma_2\beta_1 -\varsigma_1\beta_2\right)\ell.
\end{equation}
Only the affine supertranslation sector is abelian by itself.
\vspace{.5cm}
\paragraph{Case II:}\label{sec:ES_symmetry_ii} Now we consider the scalar field varying along the null generators as follows
\begin{equation}
 \cL_\ell\P=\pa_v\P=\fa\,\P.
 \label{eq:ES-caseII}
\end{equation}
Here $\fa$ denotes the scalar evolution rate, distinct
from the area density $\sfa$ used in the GSL discussion.
Initially, we allow $\fa$ to depend on both $v$ and $x^A$. 
The scalar profile is
\begin{equation}
 \P(v,x^A)=\P_0(x^A) \exp\left[ \int_0^v\ed\bar v\,\fa(\bar v,x^A) \right].
 \label{eq:ES-caseII-solution}
\end{equation}
The condition $\cL_\xi\P=0$ now reads
\begin{equation}
 \fa(\varsigma-v\beta)\P + X^A\pa_A\P=0.
 \label{eq:ES-caseII-preservation}
\end{equation}
Compatibility with the evolution equation gives a second useful relation. Using \cref{eq:ES-generator-commutator} and $\cL_\xi\P=0$, we find
\begin{align}
 \cL_\xi(\cL_\ell\P)  &=\cL_\ell(\cL_\xi\P) +\cL_{[\xi,\ell]}\P
 =\beta\fa\P,\\
 \cL_\xi(\fa\P) &= (\cL_\xi\fa)\P.
\end{align}
Therefore,
\begin{equation}
 \left(\cL_\xi\fa-\beta\fa\right)\P=0.
 \label{eq:ES-caseII-compatibility}
\end{equation}
On a region where $\P\neq0$, this becomes
\begin{equation}
 (\varsigma-v\beta)\pa_v\fa + X^A\pa_A\fa=\beta\fa.
 \label{eq:ES-rate-condition}
\end{equation}
The rate is therefore not generally invariant under $\xi$: it transforms with the rescaling weight of the null generator.
On the same patch, scalar preservation on the initial cut requires
\begin{equation}
 \fa(0,x)\varsigma + X^A\pa_A\ln|\P_0| = 0.
 \label{eq:ES-caseII-initial-condition}
\end{equation}
Together, \cref{eq:ES-rate-condition,eq:ES-caseII-initial-condition}
are necessary and sufficient for scalar preservation along the generators. Indeed,
\begin{equation}
 \pa_v\left(\frac{\cL_\xi\P}{\P}\right) = \cL_\xi\fa-\beta\fa,
\end{equation}
so the compatibility condition propagates the initial scalar preservation condition. An explicit classification is obtained when the rate is constant along each generator,
\begin{equation}
 \pa_v\fa=0,\qquad \fa=\fa(x^A).
\end{equation}
In this case,
\begin{equation}
 \P=e^{\fa v}\P_0,
\end{equation}
and direct substitution gives
\begin{equation}
 \cL_\xi\P = 
 e^{\fa v} \left[  \fa\,\varsigma\P_0+X^A\pa_A\P_0
  + v\P_0\left(X^A\pa_A\fa-\fa\beta\right)  \right].
\end{equation}
Since this must vanish for every $v$, we get
\begin{align}
 \fa\varsigma\P_0+X^A\pa_A\P_0&=0,
 \label{eq:ES-caseII-varsigma}\\
 \P_0\left(X^A\pa_A\fa-\fa\beta\right)&=0.
 \label{eq:ES-caseII-beta}
\end{align}
These equations remain valid at zeros of $\P_0$ or $\fa$, without division by either quantity.
Where $\fa\P_0\neq0$, they imply
\begin{equation}
 \varsigma = -\frac{X^A\pa_A\ln|\P_0|}{\fa},
 \qquad
 \beta=\frac{X^A\pa_A\fa}{\fa}.
 \label{eq:ES-caseII-parameters}
\end{equation}
Thus, the allowed generator is
\begin{equation}
 \xi\big|_{\cN} =X^A\pa_A
 -\frac{ X^A\pa_A\ln|\P_0| + vX^A\pa_A\fa}{\fa}\,\ell.
 \label{eq:ES-caseII-generator}
\end{equation}
The longitudinal displacement compensates for the change of the scalar produced by the transverse diffeomorphism. Unlike case $I$, there are no independent longitudinal transformations in this region: $X^A=0$ implies $\varsigma=\beta=0$.
The local algebra becomes especially clear when expressed in the following coordinates,
\begin{equation}
 p=\ln|\P|=\ln|\P_0|+\fa v,\quad \text{and} \quad x^A.
\end{equation}
Since $\pa_vp=\fa\neq0$, these coordinates are locally
admissible, and \cref{eq:ES-caseII-generator} becomes
\begin{equation}
 \xi(X)=X^A(x)\left.\pa_A\right|_p.
\end{equation}
Hence,
\begin{equation}
 [\xi(X_1),\xi(X_2)]
 =\xi([X_1,X_2]).
\end{equation}
The stabiliser is locally isomorphic to $\mathfrak{diff}(\cZ)$, embedded in the original boundary algebra through \cref{eq:ES-caseII-parameters}. Global regularity and preservation of any specified endpoint cuts may impose further restrictions.

If $\fa$ is a nonzero constant over $\cN$, \cref{eq:ES-caseII-parameters} simplifies to
\begin{equation}
 \beta=0,\qquad
 \varsigma = -\frac{1}{\fa}X^A\pa_A\ln|\P_0|.
 \label{eq:ES-caseII-constant-rate}
\end{equation}
If $\P_0$ is also a nonzero constant, then $\varsigma=0$, and only the transverse diffeomorphisms remain.
Preserving the cut $v=0$ additionally requires $\varsigma=0$. Where $\fa\P_0\neq0$, the remaining conditions are
\begin{equation}
 X^A\pa_A\P_0=0,\qquad
 \beta=\frac{X^A\pa_A\fa}{\fa}.
\end{equation}
For a nonzero constant rate, these reduce to $\beta=0$ and
$X\in\mathfrak{diff}(\cZ)_{\P_0}$. The relation $\pa_v\P=\fa\P$ specifies the behaviour of the scalar on the boundary; it does not assert an internal scaling symmetry of the Einstein-scalar action. Moreover, the above classification is kinematical. An admissible on-shell configuration must also satisfy the coupled Einstein and scalar equations.
Finally, the stabiliser is defined for a prescribed scalar profile. Its generators can therefore be held fixed under the variations used in the charge construction. If instead the scalar profile is varied and the generators are recomputed from \cref{eq:ES-caseII-parameters}, they become field dependent, and their charge variations and
brackets require the corresponding modifications.

\let\c\oldc \let\d\oldd \let\i\oldi
\bibliography{ref}

\end{document}